\documentclass[11pt]{scrartcl}
\usepackage[utf8]{inputenc}
\usepackage[colorlinks=false,hidelinks]{hyperref}
\usepackage[square,sort&compress,comma,numbers]{natbib}
\usepackage{amsmath,amsfonts,amssymb,amsthm,mathrsfs}
\usepackage{geometry}
\usepackage{braket}
\usepackage{amsthm}
\usepackage{dsfont}
\usepackage{color}
\usepackage[title]{appendix}
\usepackage{setspace}
\usepackage{soul}

\newtheorem{vor}{Assumption}[section]
\newtheorem{theorem}[vor]{Theorem}

\newtheorem{lem}[vor]{Lemma}

\newtheorem{prop}[vor]{Proposition}

\theoremstyle{definition}

\newtheorem{note}[vor]{Remark}
\numberwithin{equation}{section}

\newcommand{\Pek}{\mathrm{Pek}}

\renewcommand{\P}{\mathrm{P}}
\newcommand{\mb}{\mathbb}
\newcommand{\ms}{\mathscr}
\newcommand{\mf}{\mathfrak}
\newcommand{\Ls}{\Big\langle}
\newcommand{\Rs}{\Big\rangle}
\newcommand{\ls}{\langle}
\newcommand{\rs}{\rangle}
\newcommand{\n}{{(n)}}
\newcommand{\m}{{(m)}}
\newcommand{\eps}{\varepsilon}
\newcommand{\peps}{\pmb{\varepsilon}}
\newcommand{\eB}{\mathsf}
\renewcommand{\d}{\text{d}}

\begin{document}

\title{Asymptotic series for low-energy excitations of the Fr\"ohlich Polaron at strong coupling
}
\author{Morris Brooks\thanks{Institut für Mathematik,~Universit\"at Z\"urich,~Winterthurerstrasse~190, CH-8057 Z\"urich, Switzerland. E-mail: \texttt{morris.brooks@math.uzh.ch}}\and David Mitrouskas\thanks{Institute of Science and Technology Auastria (ISTA), Am Campus 1, 3400 Klosterneuburg, Austria. E-mail: \texttt{mitrouskas@ist.ac.at}}}

\date{\today}
\maketitle

\frenchspacing

\begin{spacing}{1.1}

\begin{abstract} 
\noindent \textsc{Abstract}. 
We consider the confined Fr\"ohlich polaron and establish an asymptotic series for the low-energy eigenvalues in negative powers of the coupling constant. The coefficients of the series are derived through a two-fold perturbation approach, involving expansions around the electron Pekar minimizer and the excitations of the quantum field.\medskip



\end{abstract}

\tableofcontents

\allowdisplaybreaks

\section{Introduction and main results}

For an open bounded region $\Omega \subset \mb R^3$, let $\Delta_\Omega$ be the Dirichlet--Laplacian and set $v_x(y) = (-\Delta_\Omega)^{-1/2}(x,y)$. The Fr\"ohlich Hamiltonian of the confined polaron, acting on the Hilbert space $L^2(\Omega) \otimes \mathcal F$ with $\mathcal F$ the bosonic Fock space over $L^2(\Omega)$, is given by
\begin{align}\label{eq:Froehlich:Hamiltonian}
\mf H_\alpha = -\Delta_\Omega + \frac{1}{\alpha} ( a^*(v_x) + a(v_x) ) + \frac{1}{\alpha^2} \mathcal N
\end{align}
with $\alpha>0$ a dimensionless coupling constant.  The creation and annihilation operators satisfy the usual bosonic commutation relations $[a(f),a^*(g)]=\ls f,g \rs$ and the field energy is described by the number operator $\mathcal N = \sum_{j=0}^\infty a^*(\varphi_j) a(\varphi_j)$ with $\{ \varphi_j \}$ an orthonormal basis of $L^2(\Omega)$.  The definition \eqref{eq:Froehlich:Hamiltonian} is slightly formal since $v_x $ is not square-integrable. However, it is well known that $\mathfrak H_\alpha$ defines a semi-bounded self-adjoint operator, as follows for instance by considering the semi-bounded quadratic form associated with the above expression  \cite{GriesemerW18,LiebY58,LiebT97}. 

%Indeed, it is not difficult to see that the form domain $Q(\mf H_\alpha)$ coincides with the form domain $Q(-\Delta_\Omega +\mathcal N)$. 
%Also note that by unitary equivalence, we could assume $\alpha<0$ instead of $\alpha>0$

The polaron model was introduced by Herbert Fr\"ohlich \cite{Froehlich37} as a model of an electron interacting with the quantized optical phonons of a polarizable crystal. Typically, the model is defined on $\mb R^3$ with $v_x(y) = (2\pi^2)^{-1} |x-y|^{-2}$ and features translation invariance. It provides a simple and well-studied model of non-relativistic quantum field theory and we refer to \cite{DonskerV83,LiebT97,LiebY58,
Seiringer21,BS1,BS2,LMM23,
MMS23,Moeller06,Mukherjee,BetzP22,BetzP23,DybalskiS2020,
Polzer22,Spohn1988,Selke} for properties, results and further references. The confined version of the  polaron was brought to our attention by a recent work of Frank and Seiringer \cite{FrankS21}. While the  confinement breaks the translational symmetry, the ultraviolet behavior of the interaction remains unaffected. In \cite{FrankS21} the authors derive a two-term expansion of the ground state energy of $\mf H_\alpha$ for large $\alpha$ in terms of the semi-classical Pekar energy and a next-order correction of order $\alpha^{-2}$ that describes the quantum fluctuations around the classical field. In the present work, we study the complete low-energy spectrum of $\mf H_\alpha$ and derive an asymptotic series for each low-energy eigenvalue in negative powers of the coupling constant. 

In the case of the translation invariant model, low-energy bound states exist for the fiber Hamiltonians, that is, after fixing the value of the total momentum \cite{MitS2022}. We expect that these eigenvalues admit an asymptotic series expansion, similarly to the eigenvalues of the confined polaron. For the ground state energy, this was predicted already in the 1950s by Bogoliubov and Tjablikov \cite{Tjablikow54,Bog50}. 
\begin{note}
By scaling all lengths by a factor of $\alpha$, the operator $\alpha^2 \mathfrak H_\alpha$ becomes unitarily equivalent to the operator
\begin{align}
- \Delta_{\Omega/\alpha} + \sqrt{\alpha} ( a^*(\tilde v_x) + a(\tilde v_x) ) + \mathcal N
\end{align} 
with $\tilde v_x(y) = (-\Delta_{\Omega/\alpha})^{-1/2}(x,y)$. This Hamiltonian describes a polaron confined to the region $\Omega/\alpha$, which corresponds to the natural length scale for large $\alpha$. This choice of units explains why $\alpha \to \infty$ is called the strong-coupling limit. For this work, we find it more convenient to use \eqref{eq:Froehlich:Hamiltonian}.
\end{note}

We are interested in the eigenvalues of $\mathfrak H_\alpha$,
\begin{align}\label{eq:eigenvalues:H}
\ms E^{(1)}(\alpha) \le \ms E^{(2)}(\alpha) \le \ms E^{(3)}(\alpha) \le \ldots \le \ms E^{(n)}(\alpha) \le \ldots,
\end{align}
where we follow the convention to count degenerate eigenvalues according to their multiplicity, starting from the lowest eigenvalue $\ms E^{(1)}(\alpha)$. Our main result is an asymptotic series for each such eigenvalue in inverse powers of $\alpha$, 
\begin{align}\label{eq:asympt:expansion}
 \ms E^{(n)}(\alpha) =  e^\Pek + \alpha^{-2} E_0 +   \alpha^{-3} E_1  +  \alpha^{-4} E_2 + \ldots 
\end{align}
with $n$-independent coefficient $e^\Pek<0$ and $n$-dependent coefficients $E_0$, $E_1, E_2, \ldots$ The crucial point is that the $\alpha$-dependence is exclusively contained in the prefactors. Before stating our result in precise form, let us explain the coefficients in more detail and formulate the  assumptions regarding the region $\Omega$.

The first term is given by the semi-classical approximation called the Pekar energy. For the ground state energy, the first-order expansion $\ms E^{(1)}(\alpha) =  e^{\Pek} + o(1)$ as $\alpha \to \infty$ is a long-known result \cite{DonskerV83,LiebT97}. To define the semi-classical energy, we replace the creation and annihilation operators in \eqref{eq:Froehlich:Hamiltonian} by a real-valued field $\alpha \varphi \in L^2(\Omega)$ and take the expectation value in a normalized electron state $\psi  \in H^1_0(\Omega)$. This leads to
\begin{align}
\mathcal E( \psi ,\varphi ) = \int\limits_\Omega dx | \nabla \psi (x)|^2 + 2 \iint\limits_{\Omega\times \Omega} dxdy\, |\psi(x) |^2 ( - \Delta_\Omega)^{-1/2}(x,y) \varphi (y) + \int\limits_\Omega dy \, \varphi (y)^2.
\end{align}
It is straightforward to find the optimal field $\varphi(y) = - \langle \psi, v_\cdot(y) \psi\rangle $ for given $\psi$, which gives the Pekar energy functional
\begin{align}
\mathcal E^\P (\psi) = \inf_\varphi \mathcal E(\psi, \varphi) = \int\limits _\Omega dx\, |\nabla \psi(x)|^2 -  \iint\limits_{\Omega\times \Omega} dxdy\, |\psi(x) |^2 ( - \Delta_\Omega)^{-1}(x,y) |\psi(y)|^2 .
\end{align}
The Pekar energy is defined as 
\begin{align}
\label{eq:def:Pekar:energy}
e^\Pek : = \inf_\psi \mathcal E^\P(\psi) 
\end{align} 
where the infimum is taken over all $\psi \in H_0^1(\Omega)$ satisfying $\|\psi \|_2 =1$. We shall work under the assumption that a unique positive minimizer $\psi^\P$ exists and that $\mathcal E^{\textrm P}(\psi) - e^{\Pek}$ obeys a lower bound that is quadratic in the distance $\psi -\psi^\P$.\medskip

\noindent \textbf{\hypertarget{1}{Assumption 1}.} The Pekar functional $\mathcal E^\P$ has a unique $L^2$-normalized minimizer $\psi^\P$ satisfying $\psi^\P >0$. \medskip

\noindent \textbf{\hypertarget{2}{Assumption 2}.} There exists a constant $\tau > 0$ such that 
\begin{align}
\mathcal E^\P(\psi) \ge \mathcal E^\P(\psi^\P)   + \tau \| \psi  - \psi^\P \|_{H^1(\Omega)}^2 \quad \forall\ \psi \in H_0^1(\Omega),\ \|\psi \|_2 =1.
\end{align}
Both assumptions are satisfied if $\Omega$ is a ball \cite{Feliciangeli20}. For the translational-invariant polaron, unique minimizers exist only up to translations \cite{Lenzmann09,Lieb77}. 
%For the analysis of \cite{FrankS21}, an additional regularity assumption on the domain $\Omega$ was introduced. In view of this, let us note that the results that will be used from \cite{FrankS21} do not depend on this assumption. 
\medskip

The Euler--Lagrange equation for $\psi^\P$ reads  $ -\Delta_\Omega \psi^\P - ( 2 (-\Delta_\Omega)^{-1}|\psi^\P|^2 + \mu^\Pek )\psi^\P =0 $, where $\mu^\Pek =  e^\Pek - \| \varphi^\P \|_2^2$ with $\varphi^\P(y) = - \langle \psi^\P , v_\cdot (y) \psi^\P \rangle $ the optimal classical field. In Lemma \ref{lem:X:phiY:bound} we will show that $ (-\Delta_\Omega)^{-1}|\psi^\P|^2\in L^\infty(\Omega)$. By positivity of $\psi^\P$, it follows from standard arguments for Schr\"odinger operators that $\psi^\P$ is the unique ground state with eigenvalue zero of the Pekar Hamiltonian
\begin{align}\label{eq:def:pekar:hamiltonian}
H_0 = -\Delta_\Omega - 2 (-\Delta_\Omega)^{-1}|\psi^\P|^2  - \mu^\Pek
\end{align}
and that $H_0$ has a positive spectral gap above zero. Thus, the reduced resolvent
\begin{align} \label{eq:def:R}
 R := - Q \frac{1}{Q H_0 Q} Q
\end{align}
with $Q=\mathds {1} -P$ and $P=| \psi^\P \rs \ls \psi^\P |$ defines a bounded negative operator on $L^2(\Omega)$.

Let us next focus on the coefficient $E_0$ in \eqref{eq:asympt:expansion}. It corresponds to the $nt h$ eigenvalue of the quadratic Bogoliubov type Hamiltonian
\begin{align}\label{eq:Bog:Hamiltonian}
\mathbb H_0 = \mathcal N + \ls \psi^\P , \phi(v_\cdot) R \phi(v_\cdot) \psi^\P \rs_{L^2}
\end{align}
where $\phi( f  ) : = a^*( f  ) +a( f )$. This operator acts  on $\mathcal F$ and describes the fluctuations of the quantum field around its classical value $\varphi^\P$. Note that in the second term, we take the partial inner product with respect to $\psi^\P \in L^2(\Omega)$, hence this is a quadratic operator on $\mathcal F$. Using $ \mathcal W_{\alpha \varphi^{\textrm P}}^{-1} (\mathfrak H_\alpha -e^{\Pek} - \alpha^{-2} E_0) \mathcal W_{\alpha \varphi^{\textrm P}} = H_0 + V$ with $V=\alpha^{-1}\phi(v_x+\varphi^\P) + \alpha^{-2} (\mathcal N -E_0 )   $, where $\mathcal W_{\alpha \varphi^\P }$ is the unitary Weyl operator associated with $ \alpha \varphi^\P \in L^2(\Omega) $, we can see that the Bogoliubov Hamiltonian emerges naturally from second-order perturbation theory. Concretely, the ansatz $\Psi_\xi =  (1 +  RV +    RV RV  )  \psi^\P \otimes \xi$ for $\xi \in \mathcal F$
satisfies
\begin{subequations}\label{eq:heuristics:Bog}
\begin{align}
P (H_0+V)   \Psi_\xi & = \psi^\P \otimes \alpha^{-2} ( \mb  H_0 -E_0) \xi  + O(\alpha^{-3}) ,\label{eq:heuristics:Bog:a}\\
Q(H_0+V) \Psi_\xi & = O(\alpha^{-3}).
\end{align}
\end{subequations}
Unlike in usual perturbation theory, the second- and higher-order terms in the eigenspace $\text{Ran}P \otimes \mathcal F$ of the unperturbed Hamiltonian still involve non-trivial operators. This is a consequence of the fact that $H_0$ acts trivially on $\mathcal F$ and, as a result, the Fock space factor $\xi $ is determined  only at order $\alpha^{-2}$ by the operator $\mb H_0$. Being a quadratic operator, it follows that  $\mb H_0$ is unitarily diagonalizable and that its low-energy spectrum consists of a sequence of discrete eigenvalues $\eB E^{(1)}< \eB E^{(2)} \le \eB E^{(3)} \le  \ldots$ (See Section \ref{sec:Bog:Hamiltonian} for details.) This reveals that $E_0 = \eB E^\n$ by choosing $\xi = \Gamma^{(n)}$ as the $nth$ eigenstate of $\mb H_0$.

The remaining terms in \eqref{eq:asympt:expansion} are obtained through a two-fold perturbation approach. The first expansion is around the ground state of $H_0$ and extends \eqref{eq:heuristics:Bog} to higher orders. The second expansion revolves around the low-energy eigenstates $\Gamma^{(n)}$ of $\mb H_0$. Here, we treat the higher-order terms in \eqref{eq:heuristics:Bog:a} that result from the first expansion as a perturbation of the operator $\alpha^{-2}\mb H_0$. For a detailed explanation, we refer to the sketch of the proof of Proposition \ref{prop:asym:expansion:HG:1}. The precise formulas for the coefficients $(E_\ell)_{\ell\ge 1}$ are summarized in Theorems \ref{thm:non-degnerate-formula} and \ref{thm:recursive:formula:generalcase} below.

The unique characteristic of this perturbation series is the presence of two natural scales, which sets it qualitatively apart from the conventional Rayleigh--Schr\"odinger perturbation expansion. In \cite{Tjablikow54,Bog50}, this is referred to as adiabatic perturbation theory.  The origin of this feature can be attributed to the small prefactor in the field energy term of \eqref{eq:Froehlich:Hamiltonian}. As a result, our approach may be of interest also for other quantum models exhibiting a similar adiabatic (or semi-classical) decoupling between two subsystems.

Our first result is a precise version of the asymptotic series  \eqref{eq:asympt:expansion}.

\begin{theorem}\label{thm:asymptotic:expansion:1} For $n\in \mb N$ let $\ms E^{(n)}(\alpha)$ denote the $n th$ eigenvalue of the Fr\"ohlich Hamiltonian $\mathfrak H_\alpha$ and $\mathsf E^\n$ the $n th$ eigenvalue of the Bogoliubov Hamiltonian $\mathbb H_0$. There exists a sequence $(E_\ell )_{\ell \in \mb N_0 }$ with $E_0 = \mathsf E^\n$ so that for every $b\in \mathbb N_0$ there are constants $C(b),\alpha(b)>0$ such that
\begin{align}
\Bigg| \, \alpha^2  \ms E^{(n)}(\alpha) - \alpha^2 e^\Pek -  \sum_{\ell=0}^{b} \alpha^{-\ell } E_\ell   \, \Bigg| \le \frac{C(b)}{ \alpha^{b+1} }
\end{align}
for all $\alpha \ge \alpha(b)$.
\end{theorem}

Under the assumption that $\mathsf E^\n$ is a non-degenerate eigenvalue, we can improve the statement by providing explicit estimates of the coefficients and the remainder of the series. In Remark \ref{remark:Borel} we comment on the potential relevance of these estimates. While we focus on the non-degenerate case in the proof, we expect the result to hold for degenerate eigenvalues as well. Note that the theorem applies, for example, to the ground state energy $\ms E^{(1)}(\alpha)$.

\begin{theorem}\label{thm:asymptotic:expansion:2} Under the additional assumption that $\mathsf E^\n$ is a non-degenerate eigenvalue of $\mb H_0$, there exists a constant $C>0$ such that $|E_\ell | \le C^\ell \sqrt{\ell !}$ for all $\ell \ge 0$. Moreover, there is a constant $\alpha_0> 0$ such that
\begin{align}\label{eq:mainbound}
\Bigg| \, \alpha^2  \ms E^{(n)}(\alpha) - \alpha^2 e^\Pek -  \sum_{\ell=0}^{b} \alpha^{-\ell } E_\ell   \, \Bigg| \le \frac{ C^{b+1} (b+1)! }{ \alpha^{b+1} }
\end{align}
for all $\alpha \ge \alpha_0$ and $b\in \mb N_0$.
\end{theorem}

\begin{note} \label{remark:Borel}  The bounds on $E_\ell$ imply the existence of the Borel transform  $B(\alpha) := \sum_{\ell=0}^\infty \alpha^{-\ell} E_\ell  / \ell!$. Using $ \int_0^\infty t^\ell e^{-t \alpha} \d t = \alpha^{-\ell-1} \ell! $ together with an unjustified interchange of summation and integration, we would thus obtain the convergent result
$\ms E^\n(\alpha) - e^\Pek = \alpha^{- 1 }\int_0^\infty  B(1/t) e^{-t \alpha} \d t$. If this identity holds, the series is called Borel summable. By the Watson-Nevanlinna theorem \cite{Sokal80}, Borel summability of $\ms E^\n(\alpha)$ is rigorously justified under the following condition: There exists a function $f:\mb C \to \mb C$ that is analytic on a disk $D\subseteq \{ z \in \mb C : \Re (z) > 0 \}$ with $0\in \partial D$, such that (i) $f(z) = \ms E^\n(1/z) $ for all $z\in D \cap \mb R_0^+$ and (ii) $f$ has an asymptotic series with remainder estimates of the form \eqref{eq:mainbound} uniformly in $z\in D$. While we do not address the question of analytic continuation of $\ms E^\n(\alpha)$ in this work, we believe that it poses an interesting mathematical problem, in particular since it evades some techniques from singular perturbation theory. The usual  approach, as deployed for instance in \cite{Auberson82,Simon70,Simon71,
Simon82,GGS70}, to infer the required analytic structure of an eigenvalue of an operator $H(z)$ relies on the non-degeneracy of the eigenvalue of $H(0)$ together with the convergence property $\lim_{z\to 0:z\in D} H(z) = H(0)$ in the norm resolvent sense. This approach can not be employed for eigenvalues of $\mf H_\alpha$ because the eigenvalue of the unperturbed Hamiltonian $H_0 = H_0 \otimes \mathds{1}$ has  infinite degeneracy. 
%Moreover, the perturbation $\mathcal N$ is highly singular, in the sense that it lacks any form of relative boundedness with respect to $H_0$.
\end{note}

Having stated our results about the existence of an asymptotic series for $\ms E^\n(\alpha)$, our next two theorems present formulas for the coefficients of this series. This requires some additional preparation. For fixed eigenvalue $\ms E^\n(\alpha)$, let $(E_\ell)_{\ell \in \mb N_0}$ denote the coefficients from Theorem \ref{thm:asymptotic:expansion:1} and set
\begin{align}\label{eq:def:V_1}
V_1  & := \phi(v_x + \varphi^\P), \\[0.5mm]
V_2 & := \mathcal N - E_0  \\[0.5mm]
V_{\ell} & : = -E_{\ell-2}\quad  \text{for}\quad \ell \ge 3. \label{eq:def:V_3}
\end{align}
Recalling $R$ from \eqref{eq:def:R} and $P = |\psi^\P \rangle \langle \psi^\P|$, we introduce for $\ell \ge 1$
\begin{align}
\mb V_\ell = \mb V_\ell ( E_0, \ldots , E_\ell) & :=   \sum_{i=1}^{\ell + 2} \sum_{\substack{ \peps  \in \mathbb N^i \\ | \peps | = \ell +2 } }  P V_{\varepsilon_1}  (R V_{\varepsilon_2} )\ldots  (R V_{\varepsilon_i}) P \label{eq:def:mbV:a}
\end{align}
where $\peps = (\eps_1,\ldots , \eps_i)$ and $| \peps | = \eps_1 + \ldots + \eps_i$. Note that $\mathbb V_\ell$ acts on $\text{Ran} P \otimes \mathcal F \cong \mathcal F$, so we can view it also as an operator on $\mathcal F$. In fact, we shall employ the identification of $\text{Ran} P \otimes \mathcal F $ with $\mathcal F$ throughout the entire paper, without mentioning it at each instance. We further define $ \widetilde { \mb V }_\ell := \mb V_\ell + E_{\ell } $ such that
\begin{align}
\widetilde{ \mb V}_\ell =\widetilde{ \mb V}_\ell (E_0, \ldots , E_{\ell-2})  =   \sum_{i=2}^{\ell + 2 }   \sum_{\substack{ \peps  \in \mathbb N^i \\  | \peps |= \ell +2 } }  P V_{\varepsilon_1}  (R V_{\varepsilon_2} )\ldots  (R V_{\varepsilon_i}) P. \label{eq:def:mbV:b}
\end{align}
Note that the coefficient $ E_{\ell} = - V_{\ell +2}$ does not appear on the right-hand side, while $E_{\ell-1} = - V_{\ell +1}$ only appears as $ P V_{\ell+1} R V_1 P + \text{h.c.}$, which vanishes for $\ell \ge 1$ since $PR=0$. This explains why $\widetilde {\mb V}_\ell$ does not depend on these coefficients.

Let $\mb P$ be the orthogonal projection onto the finite-dimensional eigenspace of $\mb H_0$ associated with the eigenvalue $\eB E^\n$ and define for $\mb Q = \mathds{1}-\mb P $ the reduced resolvent
\begin{align}
\mathbb R :=  -  \mathbb Q  \frac{1}{\mb Q( \mb H_0 - \eB E^\n )\mb Q  }\,  \mathbb Q  .
\end{align}
By Lemma \ref{lem:diagonalization of HBog} below, this is a bounded operator for every $n\in \mathbb N$. 

The next statement provides an iterative formula for the coefficients $E_\ell$ in case $\mathsf E^\n$ is non-degenerate.

\begin{theorem}\label{thm:non-degnerate-formula} Assume $\eB E^\n$ to be a non-degenerate eigenvalue of $ \mb H_0$ with normalized eigenfunction $\Gamma\in \mathcal F$. The corresponding sequence from Theorem \ref{thm:asymptotic:expansion:1} satisfies $E_0 = \eB E^\n$, $E_\ell = 0$ for all $\ell$ odd and the iterative relation
\begin{align}
E_{\ell} & =  \bigg\langle \psi^\P \otimes \Gamma  ,\bigg[ \widetilde {\mb V}_{\ell}  + \sum_{i=2}^{\ell}   \sum_{ \substack{ \peps \in \mathbb N^i \\ |\peps| = \ell } }    \mb V_{ \varepsilon_1} (\mb R \mb V_{ \varepsilon_2} ) \ldots (\mb R \mb V_{ \varepsilon_i } ) \bigg] \psi^\P \otimes \Gamma  \bigg\rangle  \label{eq:energy:order:s}
\end{align}
for all $\ell \ge 2 $ even.
\end{theorem}

Note that the right-hand side of \eqref{eq:energy:order:s} depends only on $E_0,\ldots ,E_{\ell-2}$. Besides providing a tool for the approximate computation of $\ms E^\n(\alpha)$, the theorem confirms the prediction from the physics literature \cite{Gross76,Miyake76} that the ground state energy $\ms E^{(1)}(\alpha)$, and indeed all eigenvalues that are non-degenerate in the limit, admits an expansion in powers of $\alpha^{-2}$. Before we continue with the discussion for the general case, let us illustrate the result by computing the first terms of the ground state energy $
\alpha^2 \ms E^{(1)}(\alpha) = \alpha^2 e^\Pek +  E_0 + \alpha^{-2} E_2 +  \alpha^{-4} E_4  + O(\alpha^{-6})$.
An explicit formula for $E_0 = \eB E^{(1)}$ is given in Lemma \ref{lem:diagonalization of HBog}. Using the shorthand notation $\phi  = \phi (v_x + \varphi^\P) $, one finds
\begin{align}
E_2 
% & =  \Big\langle \psi^ \P\otimes \Gamma, \widetilde {\mb V}_{2} \psi^\P \otimes \Gamma  \Big\rangle -  \Big\langle \psi^\P \otimes \Gamma , (\widetilde { \mb V}_{ 1}  \mb R \widetilde { \mb V}_{1}) \psi^\P \otimes \Gamma  \Big\rangle  \notag\\
& = \Big\langle \psi^\P \otimes \Gamma   , \big(   \phi R  \mathcal N  R \phi + \phi R \phi R \phi R \phi \big) \psi^\P \otimes \Gamma \Big\rangle  + \eB E^{(1)} \Big\langle \psi^\P \otimes \Gamma  ,  ( \phi R^2 \phi   ) \psi^\P \otimes \Gamma   \Big\rangle  \notag\\[1mm]
& \quad +  \Big\langle  \psi^\P \otimes \Gamma  , ( \phi R \phi R \phi )  (P\otimes \mb R ) ( \phi R \phi R \phi  ) \psi^\P \otimes \Gamma  \Big\rangle, \label{eq:ground:state:E2} \\[1mm]
E_4 & =  \Big\langle \psi^\P \otimes \Gamma  , \Big(  \widetilde {\mb V}_{4} +  \widetilde {\mb V}_{2}  \mb R \widetilde { \mb V}_{2} + \widetilde { \mb V}_1 \mb R \mb V_2  \mb R \widetilde { \mb V}_1 + \widetilde { \mb V}_{1} \mb R \mb V_{1}  \mb R \mb V_{1}  \mb R  \widetilde { \mb V}_{1} \Big) \psi^\P \otimes \Gamma  \Rs\notag\\
& \quad + 2\Re   \Big\langle  \psi^\P \otimes \Gamma  , \Big(  \widetilde { \mb V}_3  \mb R \widetilde { \mb V}_1 +   \widetilde { \mb V}_1 \mb R \mb V_{1}  \mb R \widetilde { \mb V}_2 \Big)  \psi^\P \otimes \Gamma  \Big\rangle. \label{eq:ground:state:E4}
\end{align}
Inserting $\widetilde{ \mb V }_\ell$ and  $\mb V_\ell$ into $E_4$ already leads to lengthy expressions, so we omit the details. Considering that $\Gamma$ is a quasi-free state, that is, $\Gamma = \mb U^* \Omega$ with $\Omega$ the Fock space vacuum and $\mb U$ a unitary Bogoliubov transformation (see Section \ref{sec:Bog:Hamiltonian}), the coefficients $E_\ell$ can  be explicitly computed. As a result, the computation of $\ms E^\n(\alpha)$ with non-degenerate $\eB E^\n$ effectively simplifies to determining the spectral properties of the one-particle operators $H_0$ and  $\mf h$, with $\mf h$ introduced in \eqref{eq:Hessian:kernel}. The operator $\mf h$  determines the spectrum of $\mb H_0$ as will be explained in the next section.
 
Turning to the general case, let $\d: =\text{dim(Ran} \mb P) \ge 1$ be the degeneracy of the eigenvalue $\eB E^\n$ of $\mb H_0$. For $\d \ge 2$, we enumerate the degenerate eigenvalues as $\eB  E^{(n_1)} = \ldots = \eB E^{(n_\d)} (=\eB E^\n)$ and select an orthonormal basis $\{ \Gamma^{(n_1)}, \ldots , \Gamma^{(n_\d)}\}\subseteq \text{Ran} \mb P$. According to Theorem \ref{thm:asymptotic:expansion:1}, each eigenvalue $\ms E^{(n_s)}(\alpha)$, $s\in \{ 1,\ldots, \d\}$, has an asymptotic expansion in terms of a sequence $(E_\ell)_{\ell \in \mb N_0}$ with $E_0 = \eB E^\n$. Fixing $s\in \{1, \ldots , \d\}$ we define the hermitian matrix $M^{(\ell)} : = \sum_{k=1}^\ell \alpha^{- k } M_k \in \mb C^{\d\times \d}$ for $\ell \ge 1$  with coefficients
\begin{align}\label{eq:matrix:definitions}
(M_k)_{rt} := \bigg\langle \psi^\P \otimes \Gamma^{(n_r)} , \bigg[ \widetilde {\mb V}_k +  \sum_{j=2}^{k} \sum_{\substack{ \peps \in \mb N^j \\ |\peps|=k } } \,  \mb V_{ \varepsilon_1 } (  \mb R  \mb V_{ \varepsilon_2} ) \ldots ( \mb R \mb V_{ \varepsilon_j})  \bigg] \psi^\P \otimes \Gamma^{(n_t)}\bigg\rangle.
\end{align} 
Note that through $\mb V_\ell$ and $\widetilde {\mb V}_\ell$, the matrix $M^{(\ell)}$ depends on the coefficients $E_0, \ldots, E_{\ell-2}$ associated with the eigenvalue $\ms E^{(n_s)}(\alpha)$. Our next result shows that the coefficients $E_{\ell}$ are determined by the $sth$ eigenvalue of $M^{(\ell)}$.

\begin{theorem}\label{thm:recursive:formula:generalcase} Choose $\ms E^{(n_s)}(\alpha)$, $s\in \{1,\ldots, \textnormal{d}\}$ with $\textnormal{d} \ge 1$, such that $\mb H_0$ has a $\textnormal{d}$-fold degenerate eigenvalue group $\eB E^{(n_1)} = \ldots = \eB E^{(n_\textnormal{d})}$. Moreover, let $ M^{(\ell)}\in \mb C^{\textnormal{d} \times \textnormal{d}}$ be defined as explained above. The coefficients from Theorem \ref{thm:asymptotic:expansion:1} associated with the eigenvalue $\ms E^{(n_s)}(\alpha)$ satisfy the recursive relation $E_0 = \eB E^{(n_s)}$ and 
\begin{align}\label{eq:limit:formula}
E_\ell  = \lim_{\alpha \to \infty} \alpha^\ell \bigg(  \mu^{(s)} ( M^{(\ell)}  ) - \sum^{\ell-1}_{k=1 }\alpha^{-k} E_{k} \bigg) 
\end{align}
where $\mu^{(s)}(M^{(\ell)} )$ denotes the $sth$ eigenvalue of $M^{(\ell)}$ counted in increasing order, i.e., $\mu^{(1)}( M^{(\ell)} ) \le \ldots \le \mu^{(\textnormal{\d})}( M^{(\ell)} )$.
\end{theorem}

This immediately implies Theorem \ref{thm:non-degnerate-formula} when $\d=1$. Hence it is sufficient to prove the general statement.

\begin{note} \label{rem:eigenvalue:M}The  theorem is based on the observation that the spectrum of $M^{(\ell)}$ consists of $\d$ analytic branches in the variable $\alpha^{-1}$ that do not intersect when restricted to a suitable interval $(0,1/\alpha_0)$. This property follows from \cite[Theorem 1.10 and Chapter 2 §6]{Kato66}. Consequently, the ordered eigenvalues admit a power series expansion, that is, for every $\ell \ge 1$ and $s\in \{1,\ldots,\d \}$ there exists a sequence $(\mu^{(s)}_k )_{k \in \mb N}$ such that
\begin{align}\label{eq:powerseries:ev}
\mu^{(s)}( M^{(\ell)} ) = \sum_{k =1}^\infty \alpha^{- k } \mu_k^{(s)} 
\end{align}
on $[\alpha_0,\infty)$. From this, we can verify that the limit in \eqref{eq:limit:formula} indeed exists by an iterative argument. For this purpose, assume that $E_{\ell-1} = \lim_{\alpha \to \infty } \alpha^{\ell-1} (\mu^{(s)}( M^{(\ell-1)} )  - \sum_{ k = 1}^{\ell-2} \alpha^{-k} E_k )$ exists. Note that the expansion of eigenvalues for the matrices $M^{(\ell)}$ and $M^{(\ell-1)}$ share the same coefficients $\mu^{(s)}_k$ for $k \in \{ 1,\ldots, \ell-1\}$, and therefore $\mu^{(s)}_k = E_k $ for all such $k$. Consequently, $\mu^{(s)}(M^{(\ell)}) - \sum_{k=1}^{\ell-1} \alpha^{-k} E_k = \sum_{k=\ell}^\infty \alpha^{-k}\mu^{(s)}_k$, which concludes the argument. For completeness, let us mention that the definition of $E_\ell$ is independent of the choice of the basis $\{\Gamma^{(n_1)}, \ldots, \Gamma^{(n_\d)} \}$.
\end{note}

As a simple yet illustrative example, consider the case $\d = 2$, i.e., assume that $\eB E^{(n-1)} < \eB E^{(n)} = \eB E^{(n+1)}  < \eB E^{(n+2)}$ for some $n\in \mb N$. Theorem \ref{thm:recursive:formula:generalcase} states that $\alpha^2 \ms E^{(n+j)}(\alpha) - \alpha^2 e^\Pek = \eB E^\n + \alpha^{-1} \mu^{(j+1)}(M_1)+O(\alpha^{-2})$ for $j\in \{ 0,1 \} $ with $M_1 \in \mb C^{2 \times 2}$ defined by \eqref{eq:matrix:definitions}. Concretely, one finds $(M_1)_{11}=(M_1)_{22}=0$ and 
\begin{align}
(M_1)_{12} = \Ls  \psi^\P \otimes  \Gamma^{(n)}, P \phi(v_\cdot) R \phi(v_\cdot + \varphi^\P ) R   \phi(v_\cdot) P \psi^\P \otimes \Gamma^{(n+1)} \Rs 
\end{align}
and thus the eigenvalues of $M_1$ are $\pm |(M_1)_{12}|$. If $M_1 $ has two identical eigenvalues, the coefficient $E_1=0$ will be the same for $n$ and $n+1$ and thus the degeneracy remains preserved at this order. On the other hand, if $(M_1)_{12} \neq 0$, then $E_1 = - |(M_1)_{12}|$ for $n$ and $E_1 = |(M_1)_{12} | $ for $n+1$. In this case, the eigenvalues  $\ms E^{(n+j)}(\alpha)$, $j\in \{0,1\}$, split at this particular order. The extensions to higher orders and $\d \geq 3$, where partial lifting of the degeneracy becomes a possibility, are straightforward.

\subsection{Outline of the proof}
\label{sect:outline}

Our proof consists of two distinct parts. The key idea, which is made precise in the second part, is to consider approximate eigenstates of $\mathfrak H_\alpha$ with corresponding approximate eigenvalues $\ms E^\n_b (\alpha)= e^\Pek +  \sum_{\ell=0}^{b} \alpha^{-2-\ell} E_\ell$. We refer to $\Psi_b \in L^2(\Omega) \otimes \mathcal F$ as an approximate eigenstate of order $b$ and $\ms E_b^\n(\alpha) $ the corresponding approximate eigenvalue if they satisfy an inequality of the form
\begin{align}\label{eq:quasi:eigenstate}
\big\| \big( \mathfrak H_\alpha -\ms E_b^\n \big)  \Psi_b \big\| \le C(b) \alpha^{-b-3}
\end{align}
for some constant $C(b)>0$ and large $\alpha$. An immediate consequence of the existence of an approximate eigenstate is that the spectrum $\sigma (\mathfrak H_\alpha)$ contains eigenvalues $\ms E(\alpha)$ that satisfy $|\ms E(\alpha) - \ms E_b^\n(\alpha)| \le C(b) \alpha^{-b-3}$. In Sections \ref{sec:Asymptotic:expansion} and \ref{sec:improved:remainder:estimates} we will construct approximate eigenstates for any $b\in \mathbb N_0$. Section \ref{sec:Asymptotic:expansion} covers the general case, while Section \ref{sec:improved:remainder:estimates} focuses on the non-degenerate case with optimized remainder estimates. It is important to note, however, that the existence of approximate eigenstates does not exclude the existence of other eigenvalues of $\mf H_\alpha$. 

This brings us to the first part of our proof, where we establish an a priori localization of the low-energy eigenvalues in the vicinity of the eigenvalues $\eB E^\n$ of $\mb H_0$. For this purpose, we derive a two-term expansion of the form
\begin{align}\label{eq:second-order-expansion}
\ms E^\n(\alpha) = e^\Pek + \alpha^{-2} \eB E^\n + o(\alpha^{-2}).
\end{align} 
Based on this two-term localization, it follows that the approximate eigenvalues $\ms E^\n_b(\alpha)$ obtained from \eqref{eq:quasi:eigenstate} approximate indeed all eigenvalues of $\mathfrak{H}_\alpha$. Combining the two findings implies that every eigenvalue admits an asymptotic series. To demonstrate \eqref{eq:second-order-expansion}, we first establish suitable condensation properties for low-energy states of $\mathfrak{H}_\alpha$, following the strategy deployed in \cite{BS1,BS2}. Subsequently, we utilize these properties to compare $\mf {H}_\alpha$ with the Bogoliubov Hamiltonian $\mb {H}_0$ through a Feshbach--Schur decomposition. This approach can be regarded as a rigorous second-order perturbation argument. The details of this proof are presented in Section \ref{sec:two:term:expansion}. Under an additional regularity assumption on $\Omega$, a two-term expansion of the form \eqref{eq:second-order-expansion} for the ground state energy $\ms {E}^{(1)}(\alpha)$ was previously obtained by Frank and Seiringer \cite{FrankS21}. Our proof takes a distinct approach since the one employed in \cite{FrankS21} does not directly extend to excited eigenvalues and, in particular, does not provide a spectral gap above the ground state energy. 

\begin{note} Let us compare our work with \cite{BPS2021} about the asymptotic expansion for the low-energy eigenvalues of $N$ bosons in the mean-field limit in inverse powers of the coupling constant $\lambda_N = (N-1)^{-1}$. Although the models and the considered limits are obviously different,  the results may appear somewhat similar. Therefore, we believe it would be helpful to highlight some interesting differences that are not immediately evident. The eigenvalue expansion in \cite{BPS2021} requires a power series expansion of the Hamiltonian itself, more specifically, of the $N$-dependent excitation Hamiltonian $\mb H$, see \cite[Def. 2.3 and Eq. (1.8)]{BPS2021}. This type of expansion is not needed in our case since \eqref{eq:Froehlich:Hamiltonian} and \eqref{eq:Gross:transformed:Froehlich} have already the desired form of a polynomial in the coupling constant $\alpha^{-1}$. On the other hand, the interaction in the polaron model is more singular than the one for bosons in the mean-field limit, the implications of which are discussed at the beginning of Section \ref{sec:Asymptotic:expansion}. An important property inherent in the strong coupling limit of the polaron pertains to the adiabatic decoupling between the electron and the quantum field, i.e., the small prefactor of $\mathcal N$ in \eqref{eq:eigenvalues:H}. As explained earlier, this gives rise to a two-scale character of the expansion, along with the fact that the leading-order eigenvalue equation exhibits infinite degeneracy. For the Bose gas, the  leading-order problem is at most finitely degenerate. Perhaps the most interesting difference concerns the method of proof. As outlined above, our approach is based on the construction of approximate eigenstates. This differs from the approach used in \cite[cf. Sect. 3.2]{BPS2021}, which utilizes an expansion of the difference of two Riesz projections 
\begin{align}
\mb P^{(n)} - \mb P^{(n)}_0 = \frac{1}{2\pi i} \oint_{\gamma^{(n)}} \bigg[ \frac{1}{z - \mb H } -   \frac{1}{z - \mb H_0 } \bigg] \, \textnormal{d}z
\end{align}
where $\mb H_0$ is an explicitly $N$-independent Bogoliubov type operator and $\gamma^{(n)} \subset \mb C$ an $N$-independent contour that encompasses both, the $n$th eigenvalue of $\mb H$ and $\mb H_0$, while it excludes all other eigenvalues   (for simplicity, we assume that the $n$th eigenvalue of $\mb H_0$ is non-degnerate), see \cite[Def. 3.10]{BPS2021}. That such a contour exists follows from Bogoliubov theory for the Bose gas in the mean-field limit \cite{LNSS,Seiringer11,GrechS13}. In the next step, the authors expand the difference of the resolvents $ ( z - \mb H )^{-1} - ( z - \mb H_0 )^{-1}$ in a power series of $\lambda_N$ with operator-valued coefficients that are of order one when integrated along $\gamma^{(n)}$ and tested on low-energy eigenstates of $\mb H$. This leads to an expansion for $\mb P^{(n)}$ which is then used to derive a series expansion for the low-energy eigenvalues. We now explain why this approach does not work for the Fr\"ohlich Hamiltonian. The analogous starting point for the expansion of the  Riesz projection would be
\begin{align}\label{eq:Riesz:projection}
\mb P^{(n)} - \mb P^{(n)}_0 =  \frac{1}{2\pi i} \oint_{\gamma^{(n)}} \bigg[ \frac{1}{z -  \mathcal W_{\alpha \varphi^{\textrm P}}^{-1}  \alpha^2  ( \mathfrak H_\alpha - e^\Pek )  \mathcal W_{\alpha \varphi^{\textrm P}}  } -   \frac{1}{z - P\otimes \mb H_0 } \bigg] \, \textnormal{d}z
\end{align}
with $\mb H_0$ defined by \eqref{eq:Bog:Hamiltonian}. Here, $\gamma^{(n)}$ is an $\alpha$-independent contour that encompasses the $n$th eigenvalues of $\alpha^2 ( \mathfrak H_\alpha - e^\Pek )$ and $\mb H_0$ and excludes all other eigenvalues of these operators. If the considered eigenvalue of $\mb H_0$ is non-degenerate, the existence of such a contour is ensured by Proposition \ref{prop:two:term:expansion}. However, we do not have an expansion of the difference of the two resolvents in terms of operators that are small compared to one when integrated along $\gamma^{(n)}$ and tested on low-energy eigenstates. By closer inspection, Lemmas \ref{Lemma-Weak_BEC} and \ref{Lemma-Algebra} show that apriori, we only have $\alpha^2 (\mf H_\alpha - e^{\Pek} ) \cong P\otimes \mb H_0 + O(\alpha^{2-\eta})$ for some small $\eta$ and thus, the perturbation would be much larger than the distance of the eigenvalue of $\mb H_0$ to the contour $\gamma^{(n)}$. Consequently, a direct attempt to expand \eqref{eq:Riesz:projection} is deemed to fail. The method of approximate eigenstates, on the other hand, avoids this difficulty as the perturbation is much smaller when applied to the correct states, cf. \eqref{eq:heuristics:Bog}. Conversely, the method of approximate eigenstates can be directly applied to the Bose gas in the mean-field limit as well. Another crucial advantage of this method is that it requires only estimates for the approximate eigenstates, which are usually much easier to establish compared to estimates for the true eigenstates needed in the expansion of the Riesz projection. While our approach does not directly give an expansion of the spectral projections, let us stress that aposteriori a series expansion for the projections is almost trivial. In the non-degenerate case, for instance, this is an immediate consequence of Theorem \ref{thm:asymptotic:expansion:2} and Proposition \ref{prop:asym:expansion:HG}.  A final aspect worth noting is that we derive the asymptotic expansion for each low-energy eigenvalue, whereas \cite[Theorem 2]{BPS2021} considers the averaged sum of eigenvalues for those eigenvalues that become degenerate in the limit.
\end{note}

\subsection{Spectrum of the Bogoliubov Hamiltonian}

\label{sec:Bog:Hamiltonian}

To diagonalize the Hamiltonian \eqref{eq:Bog:Hamiltonian}, we introduce the operator $\mf h :L^2(\Omega) \to L^2(\Omega)$ with integral kernel
\begin{align}\label{eq:Hessian:kernel}
\mf h(x,y) = \delta(x-y) + 4 \Re \langle \psi^{\P},  v_{\cdot}(x)  R  v_{\cdot}(y)  \psi^\P \rangle_{L^2}.
\end{align}
This operator arises as the Hessian of the Pekar functional $ \mathcal E(\psi^{\textrm P}, \varphi ) = \inf \sigma ( -\Delta_\Omega - 2  \langle v_x , \varphi \rangle) + \| \varphi \|_2^2 $ for real-valued $\varphi \in L^2(\Omega )$, evaluated at its minimizer $\varphi^{\P}$. For a detailed discussion, we refer to \cite[Sec. 2.2]{FrankS21}. Anticipating that $\mf h$ is positive  and has a bounded inverse, we can define the operator
\begin{align}
\mf  B \, = \, \tfrac{1}{2} \big( \mf h^{-1/4} - \mf h^{1/4} \big)   . \label{eq: A and B on Pi0}
\end{align}
\begin{lem}\label{lem:Hessian:properties} There exist constants $\tau \in (0,1)$ and $C>0$ such that $\tau  \le \mf h \le 1$, $ {\mathrm{Tr}}_{L^2}(1- \mf h) \le C $ and $\mf B^2 \le C( 1- \mf h)$. Moreover, the discrete part of the spectrum satisfies $  \sigma_{\textnormal{dsic}}(\mf h)  = \sigma ( \mf h ) \setminus \{1\} $ and $| \sigma_{\textnormal{disc}} (\mf h) | = \infty$.
\end{lem}
\begin{proof} For the first two properties, see \cite[Sec. 3.1]{FrankS21} or alternatively \cite[Lem. 1.1 and 2.3]{MMS23}. The bound on $\mf B^2$ follows by using $\mf h  \le 1$ and the inequality $(1-x)^{-1/2} \le 1+\lambda^{-3/2} x$ for $x\in (0,1-\lambda)$. Since $1 -\mf h$ is trace-class, it is compact and thus $\sigma_{\textnormal{dsic}}(\mf h)  = \sigma ( \mf h ) \setminus \{1\}$. Using \eqref{eq:Hessian:kernel}, it is easy to see that $ \varphi \notin \text{Ker}(1-\mf h)$ if $(-\Delta_\Omega)^{-1/2}\varphi \neq \text{constant}$ in $L^2(\Omega)$. Choosing $\varphi = w_j$ with $(w_j)_{j\in \mathbb N}$ the set of eigenfunctions of $-\Delta_\Omega$, this proves the last statement of the lemma.
\end{proof} 
 
We enumerate the eigenvalues of $\mf h $ by $ 0< \tau_1  \le \tau_2 \le \ldots < 1$ and the corresponding normalized eigenfunctions by $\mathfrak u_n \in L^2(\Omega)$. Note that $\tau_j\to 1$ as $j\to \infty$. 
%They satisfy $\tau_n \to 1$ as $n\to \infty$ and $(\mathfrak u_n)_{n\ge 1}$ forms an orthonormal basis of $ L^2(\Omega)$. 

Since $\mf B$ is Hilbert--Schmidt, there exists a unitary map $\mathbb U$ on $\mathcal F$, called a Bogoliubov transformation, satisfying
\begin{align} \label{eq: def of U}
\mathbb U a^*(f) \mathbb U^* \,  & = \,  a^*(    ( 1 + \mf B^2 )^{1/2}  f ) + a (  \mf B \overline{ f} ) 
\end{align}
for all $f\in L^2(\Omega)$.
%For an introduction to Bogoliubov transformations and related concepts, see \cite{JPS2007} and \cite[Sec. 4]{Bossmann2019}. 
The operator $\mathbb H_0$ defined in \eqref{eq:Bog:Hamiltonian} is diagonalized by $\mathbb U $ in the sense of the following lemma, whose proof follows from a direct computation, see \cite[Lemma 2.5]{MMS23}. We note that by Lemma \ref{lem:Hessian:properties} the value of the trace is finite and negative. 

\begin{lem}\label{lem:diagonalization of HBog} We have
\begin{subequations}
\begin{align}\label{eq: diagonalization of HBog} 
\mathbb U  \mathbb H_0 \mathbb U^*  \, & = \, \textnormal{d} \Gamma ( \mf h ^{1/2} )  +   \tfrac{1}{2}\textnormal{Tr}_{ L^2  } \big(  \mf h^{1/2} - 1 \big)  \\
\label{rem: Bog ground state energy} 
\inf \sigma( \mathbb H_0) \, &  = \,   \tfrac{1}{2}\textnormal{Tr}_{ L^2  } \big(  \mf h^{1/2} - 1 \big) . 
\end{align}
\end{subequations}
\end{lem}

We are interested in the excitation spectrum of $\mb H_0$ below $\inf \sigma  ( \mb H_0 )+1$.  This part of the spectrum contains only discrete eigenvalues of finite multiplicity, as it corresponds to the free bosonic excitation spectrum with one-particle energies determined by $\mathfrak h^{1/2}$. Denoting the lowest eigenvalue by $\eB E^{(1)} = \inf \sigma (\mathbb H_0)$, the spectrum below $\eB E^{(1)} + 1$ can be enumerated as $
\eB E^{(1)}  < \eB E^{(2)} \le \eB E^{(3)} \le  \ldots <  \eB E^{(1)} + 1 $. More specifically, for any $n \in \mathbb N$ there exists some $k\in \mb N$ and a tuple $(\nu_1,\ldots,\nu_k)\in \mb N^k_0$ such that
\begin{align}
  \label{eq:Lambda:eigenvalue}
\eB E^\n = \eB E^{(1)} + \nu_1 \tau_1^{1/2} + \ldots + \nu_{k} \tau_k^{1/2}.
\end{align}
Since $\tau_1 > 0$, the value of $k$ is uniformly bounded in $n$. In fact, there exist $j,n_0\in \mb N_0$ so that for all $n\ge n_0$, the eigenvalues are of the form $\eB E^{(n)}= \eB E^{(1)} + \tau_{n+j}^{1/2}$, and thus $\eB E^{(n)} \to \eB  E^{(1)} +1 $ as $n\to \infty$. For later purposes, let us also introduce the corresponding normalized eigenfunctions of $\mb H_0$, given by
\begin{align}\label{eq:eigenfunctions:H0}
\Gamma^\n = \mathbb U^* \gamma^\n \quad \text{with}  \quad \gamma^\n = \frac{ a^*(\mathfrak u_{\nu_1})^{\nu_1} }{\sqrt{\nu_1 !} } \cdots \frac{a^*(\mathfrak u_{\nu_k})^{\nu_k}}{\sqrt{\nu_k !} } \Omega ,
\end{align}
where $\Omega$ is the vacuum state in $\mathcal F$.

\section{Two-term expansion of low-energy eigenvalues}\label{sec:two:term:expansion}

The following statement provides a two-term localization of the low-energy spectrum of $\mf H_\alpha$ around the eigenvalues of $\mb H_0$.

\begin{prop} \label{prop:two:term:expansion}For $n\in \mathbb N$ let $\ms E^\n(\alpha)$ and $\eB E^\n$ be the $nth$ eigenvalues of $\mf H_\alpha$ and $\mb H_0$. There exist constants $C,\delta,\alpha_0>0$ such that
\begin{align}\label{eq:2term:bound:minmax}
\Big| \ms E^\n(\alpha) - e^\Pek - \alpha^{-2} \eB E^\n \Big| \le C \alpha^{-2-\delta} 
\end{align}
for all $\alpha \ge \alpha_0$.
\end{prop}

To prove the proposition, we derive upper and lower bounds for the min-max values of the Fr\"ohlich Hamiltonian, 
\begin{align}\label{def:minmax:values}
\ms E^\n( \alpha ) := \inf_{ \mathcal  V \subset L^2(\Omega) \otimes \mathcal F} \, \sup \big\{ \langle \Psi ,  \mf H_\alpha \Psi \rangle  \, | \, \Psi \in \mathcal V ,\ \| \Psi  \| = 1 \big\}, \quad n \in \mathbb N,
\end{align}
where the infimum is taken over all subspaces $\mathcal V \subset   Q(\mf H_\alpha)$ with $\text{dim} \mathcal  V = n $. Here, $ Q(\mf H_\alpha) = Q(-\Delta_\Omega +\mathcal N)$ is the form domain of $ \mf H_\alpha$ \cite{FrankS21,GriesemerW18}. 

%For completeness, we note that \eqref{eq:2term:bound:minmax} for the min-max values, combined with \eqref{eq:Lambda:eigenvalue}, implies  that \eqref{def:minmax:values} coincides with $nth$ eigenvalue of $\mf H_\alpha$.

For convenience of the reader, we recall the definition of the unitary Weyl operator $ \mathcal W_\varphi  := e^{a(\varphi) - a^*(\varphi)}$ for $\varphi \in L^2(\Omega)$ and the shift relations $ \mathcal W_{\varphi}  a(f)  \mathcal W_{\varphi}^{-1} = a(f) + \langle f,\varphi \rangle$ and
\begin{align}
\mathcal W^{-1}_\varphi \mathcal N 
\mathcal W_\varphi & = \mathcal N + \phi(\varphi) + \| \varphi \|_2^2\\[1mm]
 \mathcal W^{-1}_\varphi \phi(f) \mathcal W_{\varphi} & = \phi(f) + 2\Re \langle f ,\varphi \rangle
\end{align}
where $\phi(f) = a^*(f) + a(f)$. We denote the  coherent state associated with $\varphi \in L^2(\Omega)$ by $\Omega_\varphi  \in \mathcal F$, that is, $\Omega_\varphi = \mathcal W^{-1}_\varphi \Omega$ where $\Omega$ denotes the vacuum state in $\mathcal F$.

\subsection{Preliminary estimates}

In this section, we state and prove some preliminary estimates. The letter $C$ is used to denote generic constants, whose value may change from one line to the next.

\begin{lem}
\label{Lemma-Regularity_w}
Let $ g_x(y) = (-\Delta_\Omega)^{-3/2}(x,y)$ and $g^\Lambda_x = (\Pi_\Lambda  -1)g_x$ with the orthogonal projection $\Pi_\Lambda=\mathds{1}(-\Delta_\Omega \leq \Lambda^2)$ and recall that $v_x(y) = (-\Delta_\Omega)^{-1/2}(x,y)$. There exists a constant $C>0$ such that for all $x\in \Omega$ and $\Lambda\ge 0$,
\begin{align}
\| g^\Lambda_x\|_2  & \le C (1+\Lambda)^{-1}\\[1mm]
 \|(-\Delta_\Omega)^{1/2} g_x^\Lambda \|_2 & \leq C \sqrt{ (1+\Lambda)^{-1} \ln  (2 + \Lambda ) } \\[1mm]
%\| (-\Delta_\Omega )^{-1/2}v_x \|_2 & \le C \\[1mm]
 \| \Pi_\Lambda v_x \|_2 & \le C \Lambda^{1/2}.
\end{align}
\end{lem}

\begin{proof} We recall that a function $f$ on $[0,\infty)$ is called completely monotone if it is continuous, infinitely differentiable on $(0,\infty)$ and satisfies $(-1)^n \tfrac{d^n}{dt^n}f(t) \ge 0$ for all $n \in \mb N_0$ and $t\in (0,\infty)$. As explained in \cite[Appendix~C.1]{FrankS21}, such a function $f$ satisfies
\begin{align}\label{eq:totally:monotone:bound}
f(-\Delta_\Omega)(x,y) \le  
f(-\Delta_{\mb R^3}) (x,y) = \int_{\mb R^3} f(k^2) e^{ik(x-y)} \frac{dk}{(2\pi)^3}
\end{align}
for all $x,y\in \Omega$. We shall use this bound for the function $t\mapsto (1+t)^{-\delta}$ for $\delta>0$.

According to the definition of  $\Pi_\Lambda$ and $g_x = (-\Delta_\Omega)^{-3/2} \delta_x$, we can bound
\begin{align}
(1+\Lambda^2) \| g_x^\Lambda \|^2_2 \leq  \|  (1-\Delta_\Omega )^{1 / 2} \, g_x \|^2_2 
& \le  C (1-\Delta_\Omega)^{-2}(x,x) 
\end{align}
for suitable constant $C>0$, where we used that $-\Delta_\Omega$ has a spectral gap above zero. Applying \eqref{eq:totally:monotone:bound} we can estimate
\begin{align}
(1-\Delta_\Omega )^{-2}(x,x) &  \leq C (2\pi)^{-3}\int_{\mathbb{R}^3} (1+|k|^2)^{-2} \mathrm{d}k < \infty.
\end{align}
This proves the first bound and the other estimates are obtained in analogy.
\end{proof}

In the next lemma, we bound the Fr\"ohlich Hamiltonian in terms of the Fr\"ohlich Hamiltonian with cutoff  $\Lambda>0$ in the interaction, given by
\begin{align}\label{eq:Fr:Ham:cutoff}
\mathfrak H_\alpha^\Lambda = - \Delta_\Omega + \alpha^{-1} \phi (\Pi_\Lambda v_x) + \alpha^{-2} \mathcal{N}.
\end{align}
Also recall $\varphi^\P(y) = - \langle \psi^\P, v_\cdot(y) \psi^\P\rangle$.
\begin{lem}
\label{Lemma-Relative_Bound}
There exists a constant $C$ such that for all $\Lambda \ge 2$
\begin{align}
\mathfrak H_\alpha & \geq \mathfrak H_\alpha^\Lambda- C \sqrt{\frac{\ln \Lambda }{\Lambda}} (\mathfrak H_\alpha +1 ), \label{Lemma-Relative_Bound_1} \\[1mm]
-\Delta_\Omega + \alpha^{-2}\mathcal{N} & \leq C ( \mathfrak H_\alpha + 1). \label{Lemma-Relative_Bound_2}
\end{align}
Furthermore, we have $\|(1-\Pi_\Lambda )\varphi^\mathrm{P}\|^2_2  \leq  C\sqrt{ \Lambda^{-1} \ln \Lambda  }$.
\end{lem}

\begin{proof} Writing $(1-\Pi_\Lambda) v_x= - ( \nabla_\Omega \cdot \nabla_\Omega ) g^\Lambda_x = - \sum_{i=1}^3 [\nabla_\Omega^{(i)} ,  \nabla_\Omega^{(i)} g_x^\Lambda]$, with $g^\Lambda_x$ defined in Lemma \ref{Lemma-Regularity_w}, we obtain using the Cauchy--Schwarz inequality
\begin{align}
\label{eq:Relative_Bound_Tool}
\mf H_\alpha &= \mf H_\alpha^\Lambda - \frac{1}{\alpha}\left[\nabla_{\Omega} , \phi ( \nabla_\Omega g_x^\Lambda )\right]  \geq \mf H_\alpha^\Lambda-2\| (-\Delta_\Omega)^{1/2} g^\Lambda_x\|_2 \left(-\Delta_\Omega +\frac{1}{\alpha^2}\mathcal{N}+\frac{1}{2}\right) \notag\\
& \geq \mf H_\alpha^\Lambda - 2C \sqrt{ \frac{\ln \Lambda  }{\Lambda }} \left(-\Delta_\Omega +\frac{1}{\alpha^2}\mathcal{N}+\frac{1}{2}\right),
\end{align}
where we have also used Lemma \ref{Lemma-Regularity_w}. Choosing $\Lambda_0>0$ such that $2C \frac{\ln  \Lambda_0 }{\Lambda_0} \le \frac{1}{3}$ and making use of the fact $\mf H_\alpha^\Lambda\geq -\Delta_\Omega +\frac{2}{3\alpha^2}\mathcal{N}-3 \sup_{x\in \Omega }\|\Pi_{\Lambda}v_x\|^2_2$ yields
\begin{align}
\frac{1}{3}\bigg(-\Delta_\Omega + \frac{1}{\alpha^2}\mathcal{N} \bigg) \leq \mf H_\alpha + \frac{1}{6}+3\sup_{x\in \Omega } \|\Pi_{\Lambda_0}v_x\|^2_2
\end{align}
and therefore concludes the proof of \eqref{Lemma-Relative_Bound_2} since $\|\Pi_{\Lambda_0}v_x\|_2<\infty$. Combining Eq.~\eqref{eq:Relative_Bound_Tool} with \eqref{Lemma-Relative_Bound_2} concludes the proof of \eqref{Lemma-Relative_Bound_1} as well. Regarding the estimate on the norm $\|(1-\Pi)\varphi^\mathrm{P}\|^2_2$, note that by definition of the Pekar energy and Eq.~\eqref{Lemma-Relative_Bound_1}
\begin{align}
e^\Pek &=\mathcal{E}(\psi^\P,\varphi^\P) = \braket{\psi^\P \otimes \Omega_{\alpha\varphi^\P}  , \mf H_\alpha \psi^\P  \Omega_{\alpha\varphi^\P}   } \notag \\
& \geq \braket{\psi^\P \otimes \Omega_{\alpha\varphi^\P}, \mf H_\alpha^\Lambda \psi^\P \otimes \Omega_{\alpha\varphi^\P}  } -C \sqrt{\tfrac{\ln \Lambda }{\Lambda}}(e^\mathrm{Pek}+C) \notag  \\
&=\mathcal{E}(\psi^\P,\Pi_\Lambda \varphi^\P)+\|(1-\Pi_\Lambda )\varphi^\P \|^2-C \sqrt{\tfrac{\ln  \Lambda }{\Lambda}}(e^\mathrm{Pek}+C ).
\end{align}
Since $\mathcal{E}(\psi^\P,\Pi_\Lambda \varphi^\P )\geq e^\Pek$, this yields  $\|(1-\Pi _\Lambda )\varphi^\P\|^2_2 \leq C \sqrt{\frac{\ln \Lambda}{\Lambda}}(e^\Pek+C)$.
\end{proof}

We recall the definition \eqref{eq:def:R}, $P=|\psi^{\P}\rangle \langle \psi^{\P}| $ and let
\begin{align}\label{eq:full:resolvent}
\mathfrak R_\alpha : = - Q \frac{1}{Q( \mathfrak H_\alpha - e^{\Pek} )Q  } Q.
\end{align}
Anticipating Lemma \ref{Lemma-Weak_BEC} below, $\mf R_\alpha$ is a negative bounded operator on $L^2(\Omega) \otimes \mathcal F$ satisfying $\sup_{\alpha>0} \| \mf R_\alpha \| < \infty$. The next lemma is a simple variant of the commutator method by Lieb and Yamazaki \cite{LiebT97,LiebY58}.

\begin{lem}\label{lem:X:phiY:bound}  Let $X,Y \in \{ P,(-R)^{1/2},(-\mf  R_\alpha)^{1/2} \}$. There is a constant $C>0$ such that
\begin{align}
\sup_{\alpha>0} \| X \phi(v_\cdot ) Y (\mathcal N+1)^{-1/2}\| \le C.
\end{align}
\end{lem}
\begin{proof} Similarly as in the proof of Lemma \ref{Lemma-Relative_Bound}, write $\phi (v_x) =  \sum_{i=1}^3 [\nabla_\Omega^{(i)}, \phi( \nabla_\Omega^{(i)} g_x ) ]$ with $g_x$ defined in Lemma \ref{Lemma-Regularity_w}. With the usual bounds for creation and annihilation operators, this gives
\begin{align}
\| X \phi(v_\cdot) Y \Psi \| & \le \sup_{x\in\Omega} \| (-\Delta_\Omega)^{1/2} g_x  \|_2  ( \| X \nabla_\Omega \|\, + \| \nabla_\Omega Y \|) \|(\mathcal N+1)^{1/2}\Psi\| \notag\\
& \le C\| (\mathcal N+1)^{1/2}\Psi\|,
\end{align}
where the second step follows from  Lemma \ref{Lemma-Regularity_w} and
\begin{align}\label{Rnabla:bounds}
\| P\nabla_\Omega  \| +  \| (-R)^{1/2} \nabla_\Omega \| + \sup_{\alpha>0} \| (-\mf R_\alpha)^{1/2}\nabla_\Omega  \|   < \infty.
\end{align}
The bound involving $R$ is a consequence of $ (-\Delta_\Omega)^{-1}|\psi^\P|^2 \in L^\infty(\Omega)$, which follows from H\"older's inequality using that $ |\psi^\P|^2 \in L^2(\Omega)$ by Sobolev's inequality and that $(-\Delta_\Omega)^{-1}(x,\cdot)\in L^2(\Omega)$ for every $x\in \Omega$, as can be inferred from \eqref{eq:totally:monotone:bound}. The bound involving $\mf R_\alpha$ is obtained by using $\pm \phi(v_x) \le \tfrac12 (-\Delta_\Omega ) + C( \tfrac{1}{\alpha^2} \mathcal N + 1)$.
\end{proof}

The next lemma provides bounds for higher moments of the particle number operator after conjugation with the unitary defined in \eqref{eq:2term:bound:minmax}. For a proof, see \cite[Lemma 4.4]{Bossmann2019}.

\begin{lem}\label{lem:number:op:bounds} For every $b\in \mathbb N$ we have
\begin{align}
\mathbb U(\mathcal N+1)^b\, \mathbb U^*  \le (b\, \mathfrak c)^b (\mathcal N+1)^b , \quad  \mathbb U^* (\mathcal N+1)^b\, \mathbb U \le  (b\, \mathfrak c)^b (\mathcal N+1)^b
\end{align}
with $\mathfrak c = 3 \| \mathfrak B \|_{\mf S^2} +3 $, where $\| \cdot \|_{\mf S^2}$ denotes the Hilbert--Schmidt norm.
\end{lem}

\subsection{Upper bounds}\label{sect:upper:bounds} We consider the Weyl-transformed Fr\"ohlich Hamiltonian
\begin{align}
\widetilde {\mathfrak H}_\alpha  := \mathcal W_{\alpha \varphi^{\P}}^{-1}   ( \mathfrak H_\alpha - e^{\Pek} ) \mathcal W_{\alpha \varphi^{\P}}  = H_0 + \alpha^{-1} \phi( v_x + \varphi^\P) + \alpha^{-2} \mathcal N
\end{align}
and for $n\in \mb N$ the trial state
\begin{align}
\Psi^\n = (1 - \tfrac{1}{\alpha} R  \phi(v_\cdot) ) \psi^\P  \otimes \mathbb U^* \gamma^\n  =: \Upsilon^\n_0  - \tfrac{1}{\alpha} \Upsilon^\n_1
\end{align}
where $\mb U^*\gamma^\n$ are the normalized eigenfunctions of $\mathbb H_0$, see \eqref{eq:eigenfunctions:H0}.  
A straightforward computation, using $H_0 \psi^\P = 0$, $P\phi(v_x+ \varphi^\P) P =0$, $PR=0$ and Lemma \ref{lem:diagonalization of HBog}, leads to
\begin{align}
 \langle \Psi^\n ,  \widetilde{ \mf H}_\alpha  \Psi^\m \rangle  - & \alpha^{-2} \delta_{nm} \eB E^\n \notag \\[1mm]
& = \alpha^{-3}  \langle \psi^{\P}  \otimes \mathbb U^* \gamma^\n,  \phi( v_x  ) R   \phi( v_x+\varphi^\P  )  R  \phi( v_x   )   \psi^{\P} \otimes \mathbb U^* \gamma^\m \rangle \notag\\[1mm]
& \quad - 2 \alpha^{-3}  \Re \langle \psi^{\P}  \otimes \mathbb U^* \gamma^\n, \phi( v_x ) R  \mathcal N   \psi^{\P} \otimes \mathbb U^* \gamma^\m \rangle \notag\\[1mm]
&\quad + \alpha^{-4} \langle \psi^{\P}  \otimes \mathbb U^* \gamma^\n, \phi( v_\cdot) R \mathcal N R \phi ( v_\cdot ) \psi^{\P} \otimes \mathbb U^* \gamma^\m \rangle.
\end{align}
Using Lemmas \ref{lem:X:phiY:bound}, \ref{lem:number:op:bounds}, one verifies that the modulus of the right-hand side is bounded uniformly in $n,m$, by $C \| (\mathcal N+1)^2 \mb U^*\gamma^\n\| \alpha^{-3} \le C \alpha^{-3}$. Similarly, one shows that the trial states are almost orthonormal, i.e. that there is a constant $C$ so that $| \langle \Psi^\n, \Psi^{(m)} \rangle - \delta_{nm}| \le C \alpha^{-1} $ for all  $n,m \in \mathbb N$ and $\alpha>0$. Now let $\mathcal V = \text{Span}(\Psi^{(1)}, \ldots, \Psi^\n)$. There exists $\alpha(n)>0$ such that for all $\alpha \ge \alpha(n)$, $\textnormal{dim} \mathcal V = n$. Thus, we can use $\mathcal V $ as a test subspace in \eqref{def:minmax:values}. This implies the desired upper bound for the $nth$ min-max value $\ms E^\n(\alpha)  \le e^\Pek + \alpha^{-2} \eB E^\n +C \alpha^{-3}$.

\subsection{Lower bounds}\label{sec:lower bounds}

To prove the lower bound on the min-max values, we first estimate $\mf H_\alpha$ in terms of its Feshbach--Schur complement associated with the projection $P$. In the second step, we estimate the min-max values of the Feshbach--Schur complement in terms of the eigenvalues of the Bogoliubov Hamiltonian. For the second step, we need to establish suitable condensation properties for low-energy states of $\mf H_\alpha$.

Let $f:\mathbb R \to [0,1]$ be a smooth function satisfying $f|_{ [0,\frac{1}{2} ]}=1$ and $f|_{ [1,\infty )}=0$, such that $f_\perp:=\sqrt{1-f^2}$ is smooth as well. We use this to define the operators 
\begin{align}\label{eq:def:pi}
\pi:=f \Big (\alpha^{\kappa-2} \mathcal{W}_{\alpha \varphi^{\mathrm{P}}}^{-1}\mathcal{N}\mathcal{W}_{\alpha \varphi^{\mathrm{P}}} \Big   ) ,\quad \pi_\perp:=f_\perp \Big  (\alpha^{\kappa -2}  \mathcal{W}_{\alpha \varphi^{\mathrm{P}}}^{-1}\mathcal{N}\mathcal{W}_{\alpha \varphi^{\mathrm{P}}} \Big   )
\end{align}
with $\kappa \in (0,2)$.

\begin{lem}
\label{Lemma-Weak_BEC} Let $Q=\mathrm {1} - |\psi^\P \rangle \langle \psi^\P |$. There exist constants $\delta,\alpha_0>0$, such that $Q(\mathfrak H_ \alpha -e^\mathrm{Pek})Q\geq \delta Q$ for all $\alpha\geq \alpha_0$. Further, let  $\pi$ be defined as above for $\kappa \in (0,\tfrac{2}{7})$. Then, there exists a constant $C> 0 $ so that
\begin{align*}
\|Q \Psi\|^2+1-\|\pi\Psi\|^2\leq C \alpha^{\kappa-\frac{2}{7}}\ln \alpha
\end{align*}
for all states $\Psi\in L^2 (\Omega,\mathcal{F}\big(L^2 (\Omega ) ) )$ satisfying $\braket{\Psi, \mathfrak H_\alpha \Psi}\leq e^\mathrm{Pek}+\alpha^{-\frac{2}{7}}\ln \alpha$. 
\end{lem}
\begin{proof}For the porpuse of the proof, recall the Fr\"ohlich Hamiltonian with cutoff given in  \eqref{eq:Fr:Ham:cutoff}, where $\Pi_\Lambda $ denotes the orthogonal projection $ \mathds{1}(-\Delta_\Omega \leq \Lambda^2)$. Applying Lemma \ref{Lemma-Relative_Bound} for $\Lambda =\alpha^{\frac{4}{7}}$ yields $(1+ C \alpha^{-\frac{2}{7}} \ln \alpha ) \mathfrak H_\alpha \geq \mathfrak H_\alpha^\Lambda - C \alpha^{-\frac{2}{7}}\ln \alpha$ for a suitable constant $C$. Let us now identify $L^2 ( \Omega , \mathcal{F} (L^2 ( \Omega ) ) )\cong\mathcal{F} (\Pi_\Lambda L^2 (\Omega ) ) \otimes L^2 (\Omega,\mathcal{F} (\Pi_\Lambda L^2 (\Omega )^\perp ) )$ and let $\Theta_\varphi$ be the orthogonal projection onto elements of the form $\Omega_\varphi\otimes \widetilde{\Psi}$ with $\widetilde{\Psi}\in L^2 (\Omega,\mathcal{F} (\Pi_\Lambda L^2 (\Omega )^\perp ) )$, where $\Omega_\varphi$ is the coherent state defined by $a(f)\Omega_\varphi=\braket{f,\varphi}\Omega_\varphi$ for some $\varphi\in \Pi_\Lambda L^2 (\Omega )$.
Furthermore, let $\mathcal{N}_\perp$ be the particle number operator on $\mathcal{F} (\Pi_\Lambda L^2\left(\Omega\right)^\perp )$ and let $f_1,\dots ,f_N$ be an orthonormal basis of $\Pi_\Lambda L^2\left(\Omega\right)$. With this at hand, we can write 
\begin{align}
\mathfrak H_\alpha^\Lambda&=-\Delta_\Omega + \alpha^{-1} \phi (\Pi_\Lambda v_x) +\alpha^{-2}\sum_{j=1}^N a(f_j)a^*(f_j)-\frac{N}{\alpha^2}+\alpha^{-2}\mathcal{N}_\perp \notag \\
&=\frac{1}{\pi^N}\int_{\Pi_\Lambda L^2\left(\Omega\right)}\Theta_\varphi H_{\alpha^{-1}\varphi}\Theta_\varphi\mathrm{d}\varphi-\frac{N}{\alpha^2}+ \alpha^{-2} \mathcal{N}_\perp
\end{align}
with $H_\varphi:=-\Delta_\Omega - \braket{v_x,\varphi} -  \braket{\varphi,v_x}+\|\varphi\|_2^2$. As a consequence of Assumption \hyperlink{2}{2}, we obtain $\braket{\psi,H_\varphi \psi}\geq \mathcal{E}^\mathrm{P}(\psi)\geq e^\mathrm{Pek}+\tau_1 \|Q\psi\|^2_2$. Furthermore, $\braket{\psi,H_\varphi \psi}\geq e^\mathrm{Pek}+\tau_2 \|\varphi-\varphi^\mathrm{P}\|^2_2$ by \cite[Proposition 3.2]{FrankS21}, and therefore $H_\varphi\geq e^\mathrm{Pek}+\tau Q+\tau\|\varphi-\varphi^\mathrm{P}\|_2^2$ with $\tau:=\frac{1}{2}\min\{\tau_1,\tau_2\}$. (Note that the proof of \cite[Proposition 3.2]{FrankS21} does not depend on \cite[Assumption 3]{FrankS21}.) Using this inequality, we proceed with
\begin{align}
& \frac{1}{\pi^N}\int_{\Pi_\Lambda L^2\left(\Omega\right)}\! \! \! \Theta_\varphi H_{\alpha^{-1}\varphi}\Theta_\varphi\mathrm{d}\varphi \geq \frac{1}{\pi^N}\int_{\Pi_\Lambda L^2\left(\Omega\right)}\! \! \! \Theta_\varphi\big(e^\Pek +\tau Q+\tau \alpha^{-2}\|\varphi-\alpha \varphi^\P \|^2_2\big)\Theta_\varphi\mathrm{d}\varphi  \notag \\[1mm]
&\qquad =e^\mathrm{Pek}+\tau Q+\tau \alpha^{-2}\mathcal{W}_{\alpha \varphi^{\mathrm{P}}}^{-1}( \mathcal N - \mathcal N_\perp ) \mathcal{W}_{\alpha \varphi^{\mathrm{P}}}+ \tau \| (1-\Pi_\Lambda) \varphi^\P \|_2^2 + \frac{\tau N}{\alpha^2}.
\end{align}
Abbreviating $\varphi^{\P}_\perp = (1-\Pi_\Lambda) \varphi^{\P}$ and using the inequality $ \alpha^{-1}\phi(\varphi_\perp^\P) \le \tfrac{1}{2}\alpha^{-2}(\mathcal N_\perp+1) + 2 \|\varphi_\perp^\P\|_2^2$ together with Lemma \ref{Lemma-Relative_Bound} with $\Lambda = \alpha^{\frac{4}{7}}$, we further have
\begin{align}
\alpha^{-2} \mathcal N_\perp 
& = \mathcal{W}_{\alpha \varphi^{\mathrm{P}}}^{-1} \big( \alpha^{-2}\mathcal{N}_\perp  - \alpha^{-1} \phi( \varphi^{\P}_\perp ) +  \|  \varphi^{\P}_\perp \|^2_2  \big) \mathcal{W}_{\alpha \varphi^{\mathrm{P}}}  \notag\\
& \ge  \mathcal{W}_{\alpha \varphi^{\mathrm{P}}}^{-1} \tfrac{1}{2}\alpha^{-2}\mathcal{N}_\perp   \mathcal{W}_{\alpha \varphi^{\mathrm{P}}}  - 2  \|  \varphi^{\P}_\perp \|^2_2 - \tfrac{1}{2\alpha^2}\notag\\
& \ge  \mathcal{W}_{\alpha \varphi^{\mathrm{P}}}^{-1} \tfrac{1}{2}\alpha^{-2}\mathcal{N}_\perp   \mathcal{W}_{\alpha \varphi^{\mathrm{P}}}  -C\alpha^{-\frac{2}{7}}\ln \alpha.
\end{align}
Combining what we have so far and using the Weyl asymptotics $N\lesssim \Lambda^3=\alpha^{\frac{12}{7}}$, we arrive at  
\begin{align}
\label{Equation-LT_Estimate}
(1+ C \alpha^{-\frac{2}{7}}\ln \alpha ) \mathfrak H_\alpha \geq e^\mathrm{Pek}+\tau Q+ \tau'  \alpha^{-2}\mathcal{W}_{\alpha \varphi^{\mathrm{P}}}^{-1}\, \mathcal{N}\mathcal{W}_{\alpha \varphi^{\mathrm{P}}}-C \alpha^{-\frac{2}{7}}\ln \alpha
\end{align}
for some suitable constant $C $ and $\tau'=\min\{\tau,\tfrac{1}{2}\}$. As an immediate consequence, we obtain $Q(\mathfrak H_\alpha -e^\mathrm{Pek})Q\geq (\tau- C \alpha^{-\frac{2}{7}}\ln \alpha) Q$, which concludes the proof of the first statement of the lemma. 

Regarding the second statement, let  $\braket{\Psi, \mathfrak H_\alpha \Psi}\leq e^\mathrm{Pek}+\alpha^{-\frac{2}{7}}\ln \alpha$. Applying the inequality~(\ref{Equation-LT_Estimate}) yields 
\begin{align*}
&  \! \tau \|Q\Psi\|^2 \! + \!  \tau' \alpha^{-2} \braket{\Psi,\mathcal{W}_{\alpha \varphi^{\mathrm{P}}}^{-1}\, \mathcal{N}\mathcal{W}_{\alpha \varphi^{\mathrm{P}}}\Psi} \\[1mm]
 & \qquad \qquad \le  (1 \! + \! C \alpha^{-\frac{2}{7}} \! \ln \alpha )   \big(  e^\mathrm{Pek} \!  + \! \alpha^{-\frac{2}{7}} \! \ln \alpha  \big) \! - e^\mathrm{Pek} \! +  \! C \alpha^{-\frac{2}{7}}\ln \alpha.
\end{align*}
Since $1-\|\pi\Psi\|^2=\braket{\Psi,\pi_\perp^2\Psi}\leq 2 \braket{\Psi,\alpha^{\kappa -2}\mathcal{W}_{\alpha \varphi^{\mathrm{P}}}^{-1}\,\mathcal{N}\mathcal{W}_{\alpha \varphi^{\mathrm{P}}}\Psi}$, we conclude that\vspace{-1mm}
\begin{align}
\|Q\Psi\|^2+1-\|\pi\Psi\|^2\leq C \alpha^{\kappa -\frac{2}{7}}\ln \alpha.
\end{align}
This completes the proof of the lemma.
\end{proof}
The next result provides a lower bound for the energy evaluated in a normalized low-energy state $\Psi$ in terms of the energy evaluated in the state $ \pi \Psi / \| \pi \Psi \|$. Effectively, it shows that cutting off the number of excitations in an approximate ground state only comes at the cost of a subleading error.
\begin{lem}
\label{Lemma-Strong_BEC}
Let $D>0$ and consider a normalized state $\Psi\in L^2 (\Omega ) \otimes \mathcal{F} (L^2 (\Omega))$ satisfying $\braket{\Psi, \mathfrak H_\alpha \Psi}\leq \ms E^{(1)}(\alpha) + D \alpha^{-2}$ with $\ms E^{(1)} (\alpha) =\inf \sigma( \mathfrak H_\alpha )$. Then, there exists a constant $C>0$ such that $\braket{\pi \Psi, \mathfrak H_\alpha \pi \Psi } \| \pi \Psi \|^{-2} \leq \braket{\Psi, \mathfrak H_\alpha  \Psi} +C \alpha^{-(2+\frac{2}{7}-\kappa)}\ln \alpha$.
\end{lem}
\begin{proof}
According to the IMS identity, see \cite[Theorem A.1]{LS} and \cite[Proposition 6.1]{LNSS}, we have $\pi \mathfrak H_\alpha  \pi+\pi_\perp \mathfrak H_\alpha \pi_\perp =  \mathfrak H_\alpha -\left[\left[ \mathfrak H_\alpha ,\pi\right],\pi\right]-\left[\left[\mathfrak H_\alpha ,\pi_\perp\right],\pi_\perp\right]$. Furthermore, we can evaluate the double commutators as 
\begin{align*}
\braket{\Psi,\left[\left[\mathfrak H_\alpha ,\pi_{ \bullet }\right],\pi_{ \bullet }\right]\Psi}=\mathfrak{Re}\Big\langle \Psi,\gamma_{ \bullet  }^2\left(\alpha^{-1}a(v_x + \varphi^\mathrm{P})-\braket{v_x + \varphi^\mathrm{P},\varphi^\mathrm{P}}\right)\Psi \Big\rangle
\end{align*}
with
\begin{align*}
\gamma_{ \bullet }:=f_{ \bullet  }\!\left(\alpha^{\kappa-2}\left(\mathcal{W}_{\alpha \varphi^{\mathrm{P}}}^{-1}\, \mathcal{N}\mathcal{W}_{\alpha \varphi^{\mathrm{P}}}+ 1  \right)\right)-f_{ \bullet  }\!\left(\alpha^{\kappa-2}\mathcal{W}_{\alpha \varphi^{\mathrm{P}}}^{-1}\, \mathcal{N}\mathcal{W}_{\alpha \varphi^{\mathrm{P}}}\right),
\end{align*}
where $\bullet \in \{ \varnothing , \perp\}$. Using $v_x=  [\nabla_{\Omega} ,  \nabla_\Omega g_x]$ with $g_x(y) = (-\Delta_\Omega)^{-3/2}(x,y)$, see Lemma \ref{Lemma-Regularity_w}, in combination with the Cauchy--Schwarz inequality yields
\begin{align}
\braket{\Psi,\left[\left[ \mathfrak H_\alpha ,\pi_{ \bullet }\right],\pi_{ \bullet }\right]\Psi} & =  \mathfrak{Re} \langle \Psi,\left[\nabla_\Omega ,\gamma_{ \bullet }^2\left(\alpha^{-1}a(\nabla_\Omega g_x )-\braket{\nabla_\Omega g_x  ,\varphi^\mathrm{P}}\right)\right]\Psi \rangle \notag \\
&\quad  -\mathfrak{Re} \langle \Psi,\gamma_{ \bullet}^2\left(\alpha^{-1}a(\varphi^\mathrm{P})-\braket{\varphi^\mathrm{P},\varphi^\mathrm{P}}\right)\Psi \rangle \notag \\
& \leq \alpha^{2\kappa-4}\braket{\Psi,(-\Delta_\Omega ) \Psi}+\tfrac{3}{2}T_{1,\bullet}+\tfrac{3}{2}T_{2,\bullet}+T_{3,\bullet}+T_4 \label{eq:double:commutator}
\end{align}
with 
\begin{align}
T_{1,\bullet} & :=\alpha^{4-2\kappa}\big\|\gamma_{ \bullet }^2\big (\alpha^{-1}a(\nabla_\Omega g_x )-\braket{\nabla_\Omega g_x ,\varphi^\mathrm{P}}\big)\Psi\big\|^2 , \\
T_{2,\bullet} & :=\alpha^{4-2\kappa}\big\| \big(\alpha^{-1}a^*(\nabla_\Omega g_x )-\braket{\nabla_\Omega g_x ,\varphi^\mathrm{P}}\big) \gamma_{ \bullet }^2\Psi\big\|^2 \\ 
T_{3,\bullet} & :=\alpha^{4-2\kappa}\big\|\gamma_{ \bullet }^2 \Psi\big\|^2,\\
T_{4} & :=\alpha^{2\kappa-4}\big\| \big(\alpha^{-1}a(\varphi^\mathrm{P})-\braket{\varphi^\mathrm{P},\varphi^\mathrm{P }}\big)\Psi\big\|^2 .
\end{align}
Using $\|\gamma_{\bullet}\|\leq \alpha^{\kappa-2}\|f_{\bullet}'\|_\infty$, we can estimate $T_{3,\bullet}\leq \alpha^{2\kappa-4}\|f_{\bullet}'\|_\infty^4 $ and
\begin{align*}
T_{1,\bullet}& \leq 2\alpha^{2\kappa -4}\|f_{\bullet}'\|_\infty^4\big(\alpha^{-2}\|a(\nabla_\Omega g_x  )\Psi\|^2+\|\braket{\nabla_\Omega g_x ,\varphi^\mathrm{P}}\Psi\|^2\big)\\[1mm]
&\leq 2\alpha^{2\kappa-4}  \|f_{\bullet}'\|_\infty^4 \sup_{x\in \Omega}\|\nabla_\Omega g_x \|^2_2 \big( \braket{\Psi,\alpha^{-2}\mathcal{N}\Psi}+ \|\varphi^\mathrm{P}\|^2_2 \big) .
\end{align*}
Similarly $T_4 \leq 2\alpha^{2\kappa - 4} \|\varphi^\mathrm{P} \|^2_2  ( \braket{\Psi,\alpha^{-2}\mathcal{N}\Psi}+ \| \varphi^\mathrm{P} \|^2_2 ) $. In order to estimate $T_2$, note the operator inequality
\begin{align}
& \big(\alpha^{-1}a(\nabla_\Omega g_x )-\braket{\nabla_\Omega g_x ,\varphi^\P}\big) \! \big(\alpha^{-1}a^*(\nabla_\Omega g_x )-\braket{\nabla_\Omega g_x ,\varphi^\P }\big) \!   \notag\\[1mm]
& \hspace{6.5cm} \le \! \sup_{x\in \Omega}\|\nabla g_x \|^2_2 \alpha^{-2}\big(\mathcal{W}_{\alpha \varphi^\P}^{-1}\mathcal{N}\mathcal{W}_{\alpha \varphi^\P}+1\big)
\end{align}
and therefore
\begin{align*}
T_{2,\bullet} &\leq \alpha^{4-2\kappa}\sup_{x\in \Omega}\|\nabla_\Omega g_x \|^2_2\braket{\Psi,\gamma_{ \bullet }^2\big(\alpha^{-2}\mathcal{W}_{\alpha \varphi^\P}^{-1}\, \mathcal{N}\mathcal{W}_{\alpha \varphi^\P }+\alpha^{-2}\big)\gamma_{ \bullet }^2\Psi}\\
&=\alpha^{4-2\kappa}\sup_{x\in \Omega}\|\nabla_\Omega g_x \|_2^2\braket{\Psi,g_\bullet\big(\alpha^{-2}\mathcal{W}_{\alpha \varphi^{\mathrm{P}}}^{-1}\, \mathcal{N}\mathcal{W}_{\alpha \varphi^{\mathrm{P}}}\big)\Psi}
\end{align*}
with $g_\bullet(t):= \big(f_{ \bullet } (\alpha^{\kappa } (t+\frac{1}{\alpha^2} ) )-f_{ \bullet }  (\alpha^{\kappa }t ) \big)^4(t+\alpha^{-2})$. Since $|g_\bullet (t)|\leq 2\alpha^{4\kappa-8}\|f_\bullet '\|_\infty^4 $ we obtain $T_{2,\bullet}\leq 2\alpha^{2\kappa-4}\sup_{x\in \Omega}\|\nabla_\Omega g_x \|^2_2 \|f_\bullet '\|_\infty^4$. Combining the above estimates with \eqref{eq:double:commutator}, therefore yields $\braket{\Psi,\left[\left[ \mathfrak H_\alpha ,\pi_{ \bullet }\right],\pi_{ \bullet }\right]\Psi}\leq C \alpha^{2\kappa-4}\braket{\Psi,(-\Delta_\Omega +\alpha^{-2}\mathcal{N}+1)\Psi}$ for a suitable constant $C$. Using Lemma \ref{Lemma-Relative_Bound} and the fact that $\braket{\Psi,\mathfrak H_\alpha  \Psi}\leq \ms E^{(1)}( \alpha) +\frac{D}{\alpha^2}$ is bounded, we obtain for $\delta =2-2\kappa$
\begin{align*}
\braket{\pi \Psi, \mathfrak H_\alpha \pi\Psi}+\braket{\pi_\perp \Psi, \mathfrak H_\alpha \pi_\perp\Psi} & =\braket{\Psi,\big( \mathfrak H_\alpha -\left[\left[\mathfrak H_\alpha ,\pi\right],\pi\right]-\left[\left[\mathfrak H_\alpha ,\pi_\perp\right],\pi_\perp\right]\big)\Psi}\\[1mm]
&\leq \braket{\Psi,\mathfrak H_\alpha \Psi}+D_1 \alpha^{-(2+\delta)}
\end{align*}
and therefore
\begin{align}
\label{Equation-IMS}
\braket{\pi \Psi, \mathfrak H_\alpha  \pi\Psi}+(1-\|\pi\Psi\|^2) \ms E^{(1)}( \alpha ) & \leq \braket{\pi \Psi, \mathfrak H_\alpha \pi\Psi} +\braket{\pi_\perp \Psi, \mathfrak H_\alpha \pi_\perp\Psi} \notag\\[1mm]
& \leq \braket{\Psi, \mathfrak H_\alpha \Psi}+ C  \alpha^{-(2+\delta )}.
\end{align}
As a direct consequence of our assumption regarding the state $\Psi$, we further have $\braket{\Psi, \mathfrak H_\alpha  \Psi}\leq (1-\|\pi\Psi\|^2) ( \ms E^{(1)} ( \alpha ) +\frac{C}{\alpha^2} )+\|\pi\Psi\|^2 \braket{\Psi , \mathfrak H_\alpha \Psi}$. In combination with Eq.~(\ref{Equation-IMS}) this yields
\begin{align*}
\Big\langle\frac{\pi \Psi}{\|\pi \Psi\|}, \mathfrak H_\alpha \frac{\pi \Psi}{\|\pi \Psi\|}\Big\rangle\leq \braket{\Psi, \mathfrak H_\alpha   \Psi}+\frac{C}{\alpha^2} \, \frac{1-\|\pi\Psi\|^2}{\|\pi\Psi\|^2} + C  \frac{\alpha^{-(2+\delta )}}{\|\pi\Psi\|^2}.
\end{align*}
By Lemma \ref{Lemma-Weak_BEC} we know that $0\leq 1-\|\pi\Psi\|^2\leq C \alpha^{\kappa-\frac{2}{7}}\ln \alpha$, which concludes the proof of Lemma \ref{Lemma-Strong_BEC} since $\frac{2}{7}-\kappa \leq \delta =2-2\kappa $ for $\kappa< \frac{2}{7}$.
\end{proof}

Next, we consider a block decomposition of $ {\mathfrak H}_\alpha$ w.r.t. the projections $P$, $Q$ and estimate $\mathfrak H_\alpha - e^\Pek$ by the corresponding Feshbach--Schur complement. To this end, recall the definition of the resolvent $\mf R_\alpha$ from \eqref{eq:full:resolvent}.
\begin{lem}
\label{Lemma-Algebra}
The following expression defines a quadratic form on $Q(\mathfrak H_\alpha)$,
\begin{align}\label{eq:def:FP}
\mathfrak {F}_P :=P\Big( \mathfrak H_\alpha -e^\mathrm{Pek} + ( \mathfrak H_\alpha  -e^\mathrm{Pek}) \mf R_\alpha (\mathfrak H_\alpha  -e^\mathrm{Pek})\Big)P,
\end{align}
and satisfies the estimate $\mathfrak H_\alpha - e^\mathrm{Pek}  \geq \mathfrak {F}_{P}$.
\end{lem}
\begin{proof}
By Lemma \ref{Lemma-Weak_BEC} it is clear that $Q(\mathfrak H_\alpha  - e^\mathrm{Pek})Q$ is invertible as an operator on $Q L^2\!\left(\Omega \right)\otimes \mathcal{F} $ with bounded inverse $ - \mf R_\alpha$, and therefore $\mathfrak {F}_P$ is well-defined. Furthermore, we have the algebraic identity $\mathfrak H_\alpha -e^\mathrm{Pek}=\mathfrak{F}_P+X^* Q( \mathfrak H_\alpha -e^\mathrm{Pek})Q X$ with $X:=Q -  \mf R_\alpha Q(\mathfrak H_\alpha -e^\mathrm{Pek})P$. Since $Q(\mathfrak H_\alpha -e^\mathrm{Pek})Q\geq 0$ by Lemma \ref{Lemma-Weak_BEC} it is clear that $X^* Q(\mathfrak H_\alpha -e^\mathrm{Pek})Q X\geq 0$ as well, which concludes the proof of Lemma \ref{Lemma-Algebra}.
\end{proof}

In the next lemma, we compare the min-max values of $\mf H_\alpha$ with the min-max values of the operator $\pi \mf F_P \pi $ with $\pi$ defined by \eqref{eq:def:pi}.

\begin{lem}
\label{Lemma-Min_Max}
Let $\ms E^\n( \alpha )$ denote the $nth$ min-max value of $\mathfrak H_\alpha$ as defined in \eqref{def:minmax:values} and assume that $\ms E^\n( \alpha ) \le \ms E^{(1)}(\alpha) + D \alpha^{-2}$ for some $D>0$ and all large $\alpha$. Then, there exist $C>0$, $\kappa\in(0,\tfrac{2}{7})$ such that
\begin{align*}
\ms E^{(n)}(\alpha) \geq e^\mathrm{Pek}+\underset{\mathrm{dim} \mathcal{V}=n }{\inf_{\mathcal{V}\subset \pi\mathcal{F}  }}\, \, \,  \sup_{\Psi\in \mathcal{V}:\|\Psi\|=1}\braket{\Psi,\mathfrak{F}_P\Psi}-C \alpha^{-(2+\frac{2}{7}-\kappa) }\ln \alpha,
\end{align*}
where we interpret $\mathfrak {F}_P$ as a quadratic form on $\mathcal{F}$ via $ \textnormal{Ran}P \otimes \mathcal F \cong \mathcal  F$.
\end{lem}
\begin{proof} Given an arbitrary $0 \le \epsilon \le C\alpha^{-2}$, let $\mathcal V  \subseteq L^2(\Omega ) \otimes \mathcal F $ be an $n$ dimensional subspace satisfying $\sup_{\Psi\in \mathcal{V}:\|\Psi\|=1}\braket{\Psi , \mathfrak H_\alpha \Psi}\leq \ms E^{(n)}(\alpha) +\epsilon $.  By our assumption on $\ms E^{(n)}(\alpha)$, any state $\Psi\in \mathcal{V}$ satisfies the assumption of Lemma \ref{Lemma-Strong_BEC}. Consequently,
\begin{align*}
\ms E^{(n)}(\alpha) \geq  \sup_{\Psi\in \mathcal{V}:\|\Psi\|=1}\Big\langle \tfrac{\pi \Psi}{\|\pi \Psi\|}, \mathfrak H_\alpha  \tfrac{\pi \Psi}{\|\pi \Psi\|}\Big\rangle-C \alpha^{-(2+\frac{2}{7}-\kappa )}\ln \alpha-\epsilon.
\end{align*}
In combination with the estimate $\mathfrak H_\alpha \geq e^\mathrm{Pek}+\mathfrak{F}_P $ derived in Lemma \ref{Lemma-Algebra}, and the trivial observation that $\mathfrak{F}_P =P\mathfrak{F}_P P$, this yields
\begin{align}
\label{Equation-Min_Max}
\ms E^{(n)}(\alpha) &\geq  e^\mathrm{Pek}+\sup_{\Psi\in \mathcal{V}:\|\Psi\|=1}\Big\langle \tfrac{P\pi \Psi}{\|\pi \Psi\|},\mathfrak {F}_P \tfrac{P\pi \Psi}{\|\pi \Psi\|}\Big\rangle-C \alpha^{-(2+\frac{2}{7}-\kappa)}\ln \alpha-\epsilon \notag \\ 
&= e^\mathrm{Pek}+ \! \! \sup_{\Psi\in \mathcal{V}:\|\Psi\|=1} \! \! \left(1-\tfrac{\|Q\pi \Psi\|^2}{\|\pi\Psi\|^2}\right)\Big\langle \tfrac{P\pi \Psi}{\|P\pi \Psi\|},\mathfrak {F}_P \tfrac{P\pi \Psi}{\|P\pi \Psi\|}\Big\rangle-C \alpha^{-(2+\frac{2}{7}-\kappa )}\ln \alpha-\epsilon.
\end{align}
By Lemma \ref{Lemma-Weak_BEC} and $0\le \pi \le 1$, we have $1- C  \alpha^{\kappa -\frac{2}{7}}\ln \alpha \leq 1-\tfrac{\|Q\Psi\|^2}{\|\pi\Psi\|^2}\leq 1-\tfrac{\|Q\pi \Psi\|^2}{\|\pi\Psi\|^2}\leq 1$. Furthermore, Eq.~(\ref{Equation-Min_Max}) together with the assumption $\ms E^\n(\alpha) \leq \ms E^{(1)}( \alpha ) +\frac{C}{\alpha^2}\leq e^\mathrm{Pek}+\frac{C}{\alpha^2}$ imply
\begin{align*}
&  \Big\langle \tfrac{P\pi \Psi}{\|P\pi \Psi\|},\mathfrak {F}_P \tfrac{P\pi \Psi}{\|P\pi \Psi\|}\Big\rangle\\ 
& \quad \le \left(1-\tfrac{\|Q\pi \Psi\|^2}{\|\pi\Psi\|^2}\right)^{-1}  \left( \ms E^{(n)}(\alpha) -e^\mathrm{Pek}+\epsilon+C \alpha^{-(2+\frac{2}{7}-\kappa)}\ln \alpha\right)\leq \frac{C}{\alpha^{2}}
\end{align*}
for any $\Psi\in \mathcal{V}$. Consequently, 
\begin{align}
\left(1-\tfrac{\|Q\pi \Psi\|^2}{\|\pi\Psi\|^2}\right)\Big\langle \tfrac{P\pi \Psi}{\|P\pi \Psi\|},\mathfrak{F}_P \tfrac{P\pi \Psi}{\|P\pi \Psi\|}\Big\rangle\geq \Big\langle \tfrac{P\pi \Psi}{\|P\pi \Psi\|},\mathfrak{F}_P \tfrac{P\pi \Psi}{\|P\pi \Psi\|}\Big\rangle- C \alpha^{-(2+\frac{2}{7}-\kappa)}\ln \alpha 
\end{align} 
for any $\Psi\in \mathcal{V}$. Using Eq.~(\ref{Equation-Min_Max}) again, we arrive at
\begin{align}
\ms E^\n ( \alpha)  & \geq  e^\mathrm{Pek}+\sup_{\Psi\in \mathcal{V}:\|\Psi\|=1}\Big\langle \tfrac{P\pi \Psi}{\|P\pi \Psi\|},\mathfrak{F}_P \tfrac{P\pi \Psi}{\|P\pi \Psi\|}\Big\rangle- C \alpha^{-(2+\frac{2}{7}-\kappa )}\ln \alpha \notag \\
&=e^\mathrm{Pek}+\sup_{\Psi\in P\pi\mathcal{V}:\|\Psi\|=1}\Big\langle \Psi ,\mathfrak{F}_P \Psi \Big\rangle- C \alpha^{-(2+\frac{2}{7}-\kappa )}\ln \alpha.
\end{align}
Choosing $\alpha$ large enough such that $\|P\pi \Psi\|\geq 1-\|Q\Psi\|-\|(1-\pi)\Psi\|\geq 1-C \alpha^{\kappa -\frac{2}{7}}\ln \alpha>0$ by Lemma \ref{Lemma-Weak_BEC}, we obtain that $P\pi:\mathcal{V}\longrightarrow P\pi\mathcal{V}$ is injective and consequently $\dim P\pi\mathcal{V}=\dim \mathcal{V}=n  $. Furthermore, $P\pi \mathcal{V}\subset P L^2\!\left(\Omega \right)\otimes \pi\mathcal{F} \cong \pi\mathcal{F} $, and thus the proof is completed by taking the infimum over all $n$ dimensional subspaces of $\pi \mathcal F$.
\end{proof}

Finally, we compare $\mf F_P$ with the Bogoliubov Hamiltonian \eqref{eq:Bog:Hamiltonian}.

\begin{lem}
\label{Lemma-Bogoliubov}
Let $\pi$ and $\mf F_P$ be defined as in \eqref{eq:def:pi} and \eqref{eq:def:FP}. There exist constants $C,\gamma > 0$ such that
\begin{align*}
 \pi   \mf{F}_P    \pi \geq (1-C \alpha^{-\gamma} )\, \pi (  \mathcal W_{\alpha \varphi^{\P}}^{-1}  \, \alpha^{-2} \mb H_0 \mathcal W_{\alpha \varphi^{\P}} )  \pi -C \alpha^{-(2+\gamma)}
\end{align*}
as quadratic forms on $\mathcal F$ for all large $\alpha$.
\end{lem}
\begin{proof} 
Let us first note the following inequality,
\begin{align}\label{eq:N:Bog:bound}
\mathcal N \le C (\mathbb H_0 + 1)
\end{align}
for some $C>0$, which follows immediately from Lemma \ref{lem:diagonalization of HBog},  $\text{d}\Gamma(\mathfrak h^{1/2}) \ge \tau_1^{1/2} \mathcal N$ with $\tau_1>0$, and $
\mathbb U (\mathcal N+1) \mathbb U^* \le  C (\mathcal N+1)$.

For the proof of the lemma, we start  by conjugating $\mf F_P$ with the Weyl operator $\widetilde {\mf F}_P  := \mathcal W_{\alpha \varphi^{\P}}^{-1}  \mathfrak F_P \mathcal W_{\alpha \varphi^{\P}}$. This gives
\begin{align}
\widetilde {\mf F}_P &  =    P \big( H_0 + \alpha^{-2} \mathcal N + \alpha^{-1} \phi(v_\cdot + \varphi^{\P}) + \alpha^{-2} \langle \psi^{\P} , \phi( v_\cdot )  \widetilde {\mathfrak R} _\alpha   \phi( v_\cdot   ) \psi^{\P} \rangle _{L^2} \big) P \notag\\
& = \alpha^{-2}   \big( \mathcal N + \langle \psi^{\P} , \phi( v_\cdot ) \widetilde {\mathfrak R} _\alpha \phi( v_\cdot   ) \psi^{\P} \rangle _{L^2} \big),
\end{align}
where we omit the trivial tensor product with $P$ in the second line and view it as an operator on $\mathcal F$. We also introduced the Weyl transformed resolvent
\begin{align}
\widetilde{ \mathfrak R}_\alpha : =   \mathcal W_{\alpha \varphi^{\P}}^{-1}  \mf R_\alpha \mathcal W_{\alpha \varphi^{\P}} =  - Q \frac{1}{  H_0 +  \alpha^{-2} \mathcal N + \alpha^{-1} \phi( v_x  + \varphi^\P) }Q.
\end{align}
Recalling \eqref{eq:Bog:Hamiltonian}, we proceed with
\begin{align*}
\widetilde {\mf F}_P &  =  \alpha^{-2} \mathbb H_0   + \alpha^{-2} \langle \psi^{ \P} , \phi(  v_\cdot  ) (  \widetilde{ \mf  R} _\alpha - R  ) \phi( v_\cdot ) \psi^{\P}\rangle_{L^2} \notag\\[1mm]
&  =  \alpha^{-2} \mathbb H_0 +  \alpha^{-2} \Big( \tfrac12 \langle \psi^{\P} , \phi( v_\cdot  )  R ( \alpha^{-2} \mathcal N + \alpha^{-1} \phi(  v_\cdot + \varphi^\P  ) ) \widetilde{ \mf  R}_\alpha \phi( v_\cdot  ) \psi^{\P} \rangle _{L^2} +\text{h.c.} \Big) 
\end{align*}
where we used the second resolvent identity. Using Lemma \ref{lem:X:phiY:bound} and $[\mathcal N, R] = 0$, it is not difficult to show that the second term is bounded by
\begin{align}
&\pm  \Big( \langle \psi^{\P} , \phi( v_\cdot)  R ( \alpha^{-2} \mathcal N + \alpha^{-1} \phi( v_\cdot + \varphi^\P ) ) \mf R_\alpha \phi( v_\cdot) \psi^{\P} \rangle _{L^2} + \text{h.c.} \Big)  \notag\\[1mm]
& \hspace{5cm}  \le C \Big(  \alpha^{-4+\delta} (\mathcal N+1)^3 + \alpha^{-\delta } (\mathcal N+1) \Big).
\end{align}
Now write $\mathcal W_{\alpha \varphi^{\P}}\, ( \pi   \mf{F}_P    \pi ) \mathcal W_{\alpha \varphi^{\P}}^{-1} = \widetilde \pi \, \widetilde {\mf F}_P \widetilde  \pi  $ with $ \widetilde \pi =  \mathcal W_{\alpha \varphi^{\P}}\,  \pi  \mathcal W_{\alpha \varphi^{\P}}^{-1} =  f (\alpha^{\kappa-2}  \mathcal{N} ) $ and  choose $\delta = \kappa$, such that
\begin{align}
\widetilde \pi \, \widetilde {\mf F}_P \widetilde  \pi  & \ge \alpha^{-2}  \widetilde \pi \big(  {\mathbb H}_0 - \big(  \alpha^{-4 + \delta } (\mathcal N+1)^3 + \alpha^{ - \delta} (\mathcal N+1) \big) \widetilde  \pi \notag\\
& \ge \alpha^{-2} \widetilde \pi \big(  {\mathbb H}_0 - C \alpha^{-\kappa }  (\mathcal N+1) \big) \widetilde \pi \ge \alpha^{-2} \widetilde \pi \big(  {\mathbb H}_0 (1-C \alpha^{-\kappa}) - C \alpha^{-\kappa}   \big) \widetilde \pi,
\end{align}
where we used \eqref{eq:N:Bog:bound} in the last step. Conjugating both sides again with the Weyl operator  completes the proof of the lemma.
\end{proof}
We are now prepared to derive the lower bound in Proposition \ref{prop:two:term:expansion}.
\begin{proof}[Proof of Prop.~\ref{prop:two:term:expansion}: Lower bound.]
Let $\widetilde {\mb H}_0 = \mathcal W_{\alpha \varphi^{\P}}^{-1}  \, \mb H_0 \mathcal W_{\alpha \varphi^{\P}}$. Under the assumptions of Lemma \ref{Lemma-Min_Max}, we can apply the latter in combination with Lemma \ref{Lemma-Bogoliubov}. With $\kappa\in (0,\tfrac{2}{7})$ from Lemma \ref{Lemma-Min_Max} and $\gamma>0$ from Lemma \ref{Lemma-Bogoliubov}, this gives for large $\alpha$
\begin{align}
& \ms E^{(n)}(\alpha) - e^{\Pek} \notag\\[2mm]
& \quad \geq   (1-C \alpha^{-\gamma})\alpha^{-2}\underset{\mathrm{dim} \mathcal{V}=n }{\inf_{\mathcal{V}\subset \pi\mathcal{F}:}}\, \, \,  \sup_{\Psi\in \mathcal{V}:\|\Psi\|=1}\braket{\Psi,\widetilde {\mb H}_0  \Psi} \! - \! C\alpha^{-(2+\min\{\gamma,\frac{2}{7}-\kappa\})}\ln \alpha \notag \\
&\quad \geq  (1-C \alpha^{-\gamma})\alpha^{-2}\underset{\mathrm{dim} \mathcal{V}=n }{\inf_{\mathcal{V}\subset  \mathcal{F}  :}}\, \, \, \, \, \,  \sup_{\Psi\in \mathcal{V}:\|\Psi\|=1}\braket{\Psi,\widetilde {\mb H}_0  \Psi}- C \alpha^{-(2+\min\{\gamma,\frac{2}{7}-\kappa\})}\ln \alpha \notag \\[2mm]
\label{Equation-Upper_Bound_EV}
&\quad = (1-C \alpha^{-\gamma})\alpha^{-2}\eB E^\n  -C \alpha^{-(2+\min\{\gamma,\frac{2}{7}-\kappa\})}\ln \alpha,
\end{align}
where we used the min-max principle in the last step. Since for $n=1$ the assumption of Lemma \ref{Lemma-Min_Max} is trivially satisfied, this implies the desired lower bound on $\ms  E^{(1)}(\alpha)$. Now, since $\eB E^\n \le \eB E^{(1)} + 1$, see Section \ref{sec:Bog:Hamiltonian}, we can use the lower bound on $\ms E^{(1)}(\alpha)$ together with the upper bounds derived in Section \ref{sect:upper:bounds}, to see that
\begin{align}
\ms E^\n(\alpha) \le e^\Pek + \alpha^{-2} (\eB E^{(1)} +1) + C\alpha^{-3} \le \ms E^{(1)}( \alpha) + C \alpha^{-2 }
\end{align}
for all $n\ge 2$ and $\alpha$ large. This implies that the assumption of Lemma \ref{Lemma-Min_Max} holds for every min-max value and thus the lower bounds follow from the above argument. This completes the proof of Proposition \ref{prop:two:term:expansion}.
\end{proof}

\section{Asymptotic series for low-energy eigenvalues}\label{sec:Asymptotic:expansion}

The main goal of this section is to derive an inequality of the form \eqref{eq:quasi:eigenstate}. Combined with Proposition \ref{prop:two:term:expansion}, this will be used to prove Theorem \ref{thm:asymptotic:expansion:1}. For the definition of approximate eigenstates, we need to pay attention to the domain of $\mf H_\alpha$. To this end, let us recall the fact that the domain of $\mf H_\alpha$ is related to the domain of $-\Delta_\Omega + \mathcal N$ via an $x$-dependent Weyl operator, called the Gross transformation \cite{FrankS21,GriesemerW18}. In the following sections, we will directly work with the Gross transformed Fr\"ohlich Hamiltonian so that our approximate eigenstates will be elements of the domain of the non-interacting Hamiltonian.

\subsection{Gross transformed Hamiltonian}
\label{sect:Gross Hamiltonian}
The Gross transformation is defined for $\Lambda \ge 0$ as the $x$-dependent Weyl operator 
\begin{align}\label{eq:Gross:trafo}
U_\Lambda := \mathcal W_{ g_x^\Lambda / \alpha } = e^{\frac{1}{\alpha} (  a(g^\Lambda_{x}) - a^*(g^\Lambda_{x})) } 
\end{align}
with $g^\Lambda_x =  (\Pi_\Lambda - 1) g_x $, $g_x(y)=(-\Delta_\Omega)^{-3/2}(x,y)$ and $\Pi_\Lambda = \mathds{1} (-\Delta_\Omega \le \Lambda^2)$. Note that for $\Lambda=\infty$, we have $g_x^\infty \equiv 0$ for all $x\in \Omega$ and thus $U_\infty= \mathds{1}$.

To simplify the notation, let us from now on abbreviate $p:=-i \nabla_\Omega$. A simple computation using the shift properties of Weyl operators shows that
\begin{align}
%U p U^* & = p + \tfrac{1}{ \alpha} ( %a^*(p g_x) + a(p g_x ))   \\[1mm]
U_\Lambda p^2 U^*_\Lambda & = p^2 + \tfrac{1}{ \alpha}   ( 2 a^*(p g^\Lambda_x) p + 2 p a(p g^\Lambda _x) + \phi(p^2 g^\Lambda _x) ) +  \tfrac{1}{\alpha^2} \phi ( p g^\Lambda _x)^2  \\[1mm]
U_\Lambda \mathcal N U^*_\Lambda & = \mathcal N + \tfrac{1}{\alpha} \phi(  g^\Lambda _x) + \tfrac{1}{\alpha^2}  \| g^\Lambda _x\|^2_2   \\[1mm]
U_\Lambda \phi ( v_x + \varphi^\P ) U^*_\Lambda & = \phi(v_x + \varphi^\P) + \tfrac{2}{ \alpha} \Re \langle v_x+ \varphi^\P, g^\Lambda _x \rangle.
\end{align}
With these relations, the Gross transformed Fr\"ohlich Hamiltonian (after conjugation with the Weyl operator $\mathcal W_{\alpha \varphi}$; note that $[U_\Lambda, \mathcal W_{\alpha \varphi}]=0$) is found to be
\begin{align}
\mathfrak H_{\alpha,\Lambda}^{{\textnormal{G}}}  & : =  U_\Lambda  \mathcal W_{\alpha \varphi} ( {\mf H}_\alpha - e^\Pek)   \mathcal W_{\alpha \varphi}^{-1} U^*_\Lambda  \notag\\[1mm]
& = U_\Lambda \big( H_0 + \tfrac{1}{\alpha} \phi(v_x + \varphi^\P) +\tfrac{1}{\alpha^2} \mathcal N \big) U^*_\Lambda \notag\\
&  = H_0 +  \frac{1}{ \alpha} \Big( \phi(v_x + \varphi^\P + p^2 g_x^\Lambda) + 2 a^*( p g^\Lambda_{x}) p + 2  p a (p g^\Lambda_{x}) \Big) \notag\\
& \quad  + \frac{1}{\alpha^2} \big( \mathcal N + \phi(p g_x^\Lambda)^2 + 2 \Re \ls v_x + \varphi^\P, g^\Lambda_x \rs \Big)  +  \frac{1}{\alpha^3} \phi(g^\Lambda_x) + \frac{1}{\alpha^4}  \| g^\Lambda_x \|^2_2 \label{eq:Gross:transformed:Froehlich}
\end{align}
where $H_0$ is the Pekar Hamiltonian defined in \eqref{eq:def:pekar:hamiltonian}.
\allowdisplaybreaks
For coefficients $(E_\ell)_{\ell \ge 0} $ we set
\begin{subequations}
\begin{align}
K_{1} & := \phi(v_x + \varphi^\P + p^2 g_x^\Lambda) + 2 a^*( p g_{x}^\Lambda) p + 2  p a (p g^\Lambda_{x}) \label{eq:def:K_1} \\[0.5mm]
K_{2} & :=  \mathcal N + \phi(pg^\Lambda_x)^2 + 2 \Re \ls v_x + \varphi^\P, g^\Lambda_x \rs - E_0 \label{eq:def:K2}\\[0.5mm]
K_{3} & := \phi(g^\Lambda_x) - E_1 \\[0.5mm]
K_{4} & := \| g^\Lambda_x \|^2_2 - E_2 \\[0.5mm]
 K_{\ell} & := - E_{\ell - 2 } \quad \forall \ell \ge 5 \label{eq:def:K:ell+1}.
\end{align}
\end{subequations}
For $\Lambda=\infty$, we have $K_{\ell} = V_\ell$ with $V_\ell $ defined by \eqref{eq:def:V_1}--\eqref{eq:def:V_3}. 

With the above definitions, we can write
\begin{align}\label{eq:H:G:formula:sum:K}
\mathfrak H_{\alpha,\Lambda }^{\textnormal{G}} - \sum_{\ell=0}^{b} \alpha^{-\ell-2} E_\ell = H_0 + K \quad \text{with} \quad K := \sum_{\ell=1}^{b+2} \alpha^{-\ell} K_{\ell} .
\end{align}
The next lemma shows that up to second order, the expansion around $P$ of the Gross transformed Hamiltonian coincides with the expansion of the original Fr\"ohlich Hamiltonian.  We recall that $\mb H_0$ does not depend on $\Lambda$ and note that for $\Lambda = \infty$, the statement of the lemma holds by definition.

\begin{lem} \label{lem:Bogoliubov:identity} Let $P=|\psi^\P \rangle \langle \psi^\P |$, $\mb H_0$ defined by \eqref{eq:Bog:Hamiltonian} and $ E_0 \in \mb R$. For every $\Lambda \ge 0$, the operators $K_1$, $K_2$ defined by \eqref{eq:def:K_1} and \eqref{eq:def:K2}  satisfy 
\begin{align}\label{eq:identity:Bog}
PK_1P + P ( K_2 + K_1 R  K_1 ) P = P\otimes ( \mb H_0 - E_0) .
\end{align}
\end{lem}
%
%In analogy to \eqref{eq:def:mbV:a} and \eqref{eq:def:mbV:b}, we now define the $\Lambda$-dependent operators 
%\begin{align}\label{def:mb:V:k}
%\mb K_{\ell,\Lambda}  = \mb K_\ell \big(  E_0, \ldots, E_{\ell } \big) & :=   \sum_{i=1}^\ell \sum_{\substack{ \peps \in \mathbb N^i \\ |\peps|  = \ell+2 }}  P  K_{\varepsilon_1}  (R K_{\varepsilon_2} )\ldots  (R K_{\varepsilon_i}) P 
%\end{align}
%and $\widetilde {\mb K}_\ell : = \mb K_\ell + E_{\ell}$. Note that $\widetilde {\mb K}_\ell = \widetilde{ \mb K }_\ell \big(  E_0 , \ldots, E_{\ell-2} \big)$, as can be seen from
%\begin{align}\label{def:mb:tilde:V:k}
%\widetilde{ \mb K}_{\ell,\Lambda} & =   P \widetilde K_{\ell+2}  P +    \sum_{i=2}^\ell  \sum_{\substack{ \peps \in \mathbb N^i \\ |\peps|  = \ell+2 }}P  K_{\varepsilon_1} (R K_{\varepsilon_2} )\ldots  (R K_{\varepsilon_i}) P
%\end{align}
%where $\widetilde K_3 = \phi(g_x^\Lambda)$, $\widetilde K_4 = \| g_x^\Lambda \|_2^2$ and $\widetilde K_\ell =0$ else. Also note that $\mb K_\ell$, $\widetilde {\mb K}_\ell$ for $\ell$ even (odd) are operators with an even (odd) number of creation and annihilation operators.

\begin{proof}[Proof of Lemma \ref{lem:Bogoliubov:identity}] First note that $PK_1P=0 $, which is a consequence of $\ls \psi^\P , v_\cdot \psi^\P \rs = - \varphi^\P$ and $2 \ls \psi^\P , ( p g^\Lambda_\cdot)  p \psi^\P \rs  = - \ls \psi^\P , ( p^2 g^\Lambda_\cdot)   \psi^\P \rs$ together with integration by parts and the fact that $\psi^\P$ is real-valued.

For the remaining term, note that
\begin{align}
\phi(p g^\Lambda_{x})^2 
& =  \frac{1}{2}\big[  [ H_0  , a^*(g^\Lambda_{x}) - a(g^\Lambda_{x})] , a^*(g^\Lambda_{x}) - a(g^\Lambda_{x}) \big] 
\end{align}
and thus, using $ H_0 \psi = 0$, 
\begin{align}
\Ls \psi^\P ,  \phi( p g^\Lambda_{x})^2 \psi^\P \Rs 
& = - \Ls \psi^\P , ( a^*(g^\Lambda_{x}) - a(g^\Lambda_{x} ) ) H_0  ( a^*(g^\Lambda_{x}) - a(g^\Lambda_{x} ) ) \psi^\P \Rs_{L^2}  .
\end{align}
On the other hand, we use
\begin{align}
2 p a(p g^\Lambda_{x}) + 2 a^*( g^\Lambda_{ x}) p & = \big\{ p , a(p g^\Lambda_{ x}) + a^*( p g^\Lambda_{ x}) \big\} - \phi(p^2 g^\Lambda_{x}) \notag\\
& = \big[ H_0  , a^*( g^\Lambda_{x}) - a( g^\Lambda_{x}) \big] - \phi(p^2 g^\Lambda_{x})
\end{align}
and again $P H_0 = 0$ in order to find $
P K_1 Q =  P \phi(v_x ) Q  - P ( a^*( g^\Lambda_{x} ) - a(g^\Lambda_{x}))  Q H_0  $
and thus
\begin{align}
 \ls \psi , K_1 R K_1 \psi \rs_{L^2}
& =  \langle \psi^\P, \phi(v_{\cdot } )    R \phi(v_ \cdot   )  \psi^\P \rangle_{L^2}  - \langle \psi^\P, \phi(v_{\cdot}  )  Q ( a^*( g^\Lambda_{ \cdot } ) - a( g^\Lambda_{\cdot }) )   \psi^\P \rangle_{L^2} \notag \\ 
& \quad + \langle \psi^\P ,    ( a^*( g^\Lambda_{\cdot } ) -  a ( g^\Lambda_{\cdot } ) )  Q \phi(v_{ \cdot  }   )   \psi^\P \rangle_{L^2} \notag\\
& \quad +  \langle \psi^\P,    ( a^*( g^\Lambda_{ \cdot }) -  a ( g_{\cdot } ) )  (H_0-\mu)  ( a^*( g^\Lambda_{ \cdot } ) -  a (g^\Lambda_{ \cdot } ) ) \psi^\P \rangle_{L^2} .
\end{align}
Further note that
\begin{align}
& \langle \psi^\P , \phi(v_{\cdot}  )  Q ( a^*( g^\Lambda_{ \cdot } ) - a( g^\Lambda_{\cdot }) )   \psi^\P \rangle_{L^2} - \langle \psi^\P ,    ( a^*( g^\Lambda_{\cdot } ) -  a ( g^\Lambda_{\cdot } ) )  Q \phi(v_{ \cdot  }   )   \psi^\P \rangle_{L^2} \notag\\[1mm]
& \hspace{5.5cm} = \langle \psi^\P, [ \phi(v_{\cdot}+\varphi^\P ) ,   a^*( g^\Lambda_{ \cdot } ) - a( g^\Lambda_{\cdot }) ]  \psi^\P \rangle_{L^2}
\end{align}
since $P \phi (v_\cdot) Q = P \phi (v_\cdot+\varphi^\P) Q =  P \phi (v_\cdot+\varphi^\P) $. 

Taking all the above into account, we arrive at
\begin{align}
\ls \psi^\P, ( K_2 + K_1 R K_1 ) \psi^\P \rs & = \mathcal N +  \langle \psi^\P, \phi(v_{\cdot } )    R \phi(v_ \cdot   )  \psi^\P \rangle   + 2 \ls \psi^\P, \Re \ls v_x + \varphi^\P, g^\Lambda_x \rs \psi^\P \rs - E_0 \notag\\ 
& \quad - \langle \psi^\P, [ \phi(v_{\cdot}+\varphi^\P ) ,   a^*( g^\Lambda_{ \cdot } ) - a( g^\Lambda_{\cdot }) ]  \psi^\P \rangle.
\end{align}
The commutator gives $
[ \phi( v_{x} + \varphi^\P ) ,  a^*( g^\Lambda_x ) - a( g^\Lambda_x) ]  = 2 \Re \ls v_x + \varphi^\P , g^\Lambda_x \rs$
and thus
\begin{align}
\ls \psi^\P, ( K_2 + K_1 R K_1 ) \psi^\P \rs  & = \mathcal N  +  \langle \psi^\P, \phi(v_{\cdot } )    R \phi(v_ \cdot   )  \psi^\P \rangle   - E_0 = \mb H_0 -E_0
\end{align}
as claimed.
\end{proof}

\subsection{Approximate eigenstates}
\label{sec:approximate:eigenstates}

We fix $\eB E^\n$ to be the $nth$ eigenvalue of $\mb H_0$ with spectral projection $\mb P$ and degeneracy $\d=\text{dim}(\text{Ran}\mb P) \ge 1$. For $\d \ge 2$ we enumerate the degenerate eigenvalues as $\eB E^{(n_1)} = \ldots = \eB E^{(n_\d)} (= \eB E^\n)$ and fix an orthonormal basis $\{ \Gamma^{(n_1)}, \ldots, \Gamma^{(n_\d)}\}$ of $\text{Ran}\mb P$.  Also recall the definitions of  $\mb Q=\mathds {1} - \mb P$ and $\mb R = \mb Q( \mb Q( \eB E^\n - \mb H_0) \mb Q )^{-1} \mb Q$.

For $s\in \{1,\ldots,\d\}$ and $b\in \mb N_0$, we consider the coefficients $(E_\ell)_{\ell \in \mb N_0}$ defined by Equation \eqref{eq:limit:formula} in Theorem \ref{thm:recursive:formula:generalcase} and define the state
\begin{align}\label{eq:Phi:n:m:1}
\Psi_b  & =    \bigg( \sum_{i=0}^{2b+3} ( R  K )^i \, \bigg)  \psi^\P \otimes \xi_b \quad \text{with} \quad \xi_b =   \sum_{j=0}^{b}  ( \mb R  \mb V )^j   \ms U \Gamma^{(n_s)} ,
\end{align}
where $\ms U$ is a yet-to-be-determined unitary on $\text{Ran}\mb P$, $K = \sum_{\ell =1}^{b+2} \alpha^{-\ell } K_\ell$ as in \eqref{eq:H:G:formula:sum:K} and
\begin{align}\label{eq:def:mathbb V}
\mb V := \sum_{\ell=0}^b \alpha^{-\ell} \mb V_\ell
\end{align}
with $\mb V_\ell$ defined by \eqref{eq:def:mbV:a}. Note that $\Psi_b$ depends on $(E_\ell)_{\ell \in \mb N_0}$ through the definitions of $K_\ell$ and $\mb V_\ell$ and also on $\Lambda \ge 0$ through the definition of $K_\ell$.  Expanding $\Psi_b$ up to $(RK)^{2b+3}\psi^{\P}\otimes \xi_b$ allows us to compensate for the large $\Lambda$-dependence of $K_1$.

The next statement is the main result of this section. It shows that $\Psi_b$ is an approximate eigenstate of order $b$.

\begin{prop} \label{prop:asym:expansion:HG:1} Let $\eB E^{(n_s)}$ for $s\in \{1,\ldots, \textnormal{d}\}$, $\textnormal{d} \ge 1$, be a $\textnormal{d}$-fold degenerate eigenvalue of $\mb H_0$ with normalized eigenfunction $\Gamma^{(n_s)} \in \mathcal F$. For every $b\in \mb N_0$, the coefficients $(E_\ell)_{\ell \in \mb N_0}$ from Theorem \ref{thm:recursive:formula:generalcase} and the state $\Psi_b$ given by \eqref{eq:Phi:n:m:1} satisfy the following property. There are constants $\alpha(b),C(b)>0$ such that for $\Lambda = \alpha^{2b+2}$ and all $\alpha \ge \alpha(b)$
\begin{align}\label{eq:bound:ev:equation:H_G}
 \bigg\| \bigg(  \mf H_{\alpha,\Lambda}^{\textnormal{G} }  - \sum_{\ell =0}^b \alpha^{-\ell-2 } E_{\ell}  \bigg)   \Psi_b \bigg\| \le \frac{C(b)}{\alpha^{b+3}}.
\end{align}
\end{prop} 

Since the idea of the proof is simple and potentially interesting also beyond the Fr\"ohlich polaron model, let us explain it in more detail.  For notational convenience, let $\d=1$. We first consider the state $\Psi_{\xi,b} = \sum_{i=0}^{2b+3} (  R K )^i \psi^{\P} \otimes \xi$ for some $\xi \in \mathcal F$. This corresponds to a perturbation ansatz for the ground state of $ H_0 + K$; compare with \eqref{eq:H:G:formula:sum:K} and recall that $\psi^\P$ is the unique ground state of $H_0$ with eigenvalue zero. The state $\xi$ remains unspecified at this point, since the unperturbed Hamiltonian $H_0$ acts as the identity on $\mathcal F$. Concretely, one uses $\mathds 1=P+Q$, $PH_0=0$ and $H_0 R = - Q$ to compute 
\begin{align}
&  (H_0 +K) \Psi_{\xi,b} =    \bigg( \sum_{i=0}^{2b+3} PK (  R K  )^i \bigg)   \psi^\P \otimes \xi +  Q  K (  R K  )^{2b+3}   \psi^\P \otimes \xi. \label{eq:heuristic:argument}
\end{align}
Since the operators $(-R)^{1/2}K (- R)^{1/2}$ and $(-R)^{1/2}KP$ each contribute a factor $\alpha^{-1}$ and the operator $QK$ a factor $\alpha^{-1}\sqrt \Lambda$, the second term is of order $O(\alpha^{-2b-4}\sqrt \Lambda ) = O(\alpha^{-b-3})$ if we choose $\Lambda = \alpha^{2b+2}$. The first term, on the other hand, involves a complicated operator acting on $\mathcal F$. It follows from Lemma \ref{lem:Bogoliubov:identity}, the choice $E_0 = \eB E^\n$ and a short computation  that
\begin{align}\label{eq:heuristic:discussion:mbK:1}
\sum_{i=0}^{2b+3} P K  ( RK)^i P &   = \alpha^{-2} \big( \mb H_0 - \eB E^\n  + \mb V  \big) +O\big(  \alpha^{-b-3}\big),
\end{align}
with the precise statement given in Lemma \ref{lem:PKRKP:approximation}. Together with \eqref{eq:heuristic:argument}, this provides an extension of \eqref{eq:heuristics:Bog} (with $V$ replaced by $K$) to higher orders.  This identity is the reason for the definition of $\mb V = \sum_{\ell =1 }^b  \alpha^{-\ell} \mb V_\ell$, which we now treat it as a perturbation of $\mb H_0 - \eB E^\n$, which in turn motivates the choice of $\xi=\xi_b$ from \eqref{eq:Phi:n:m:1}. Invoking $\mathds 1= \mb P + \mb Q $, $(\mb H_0-\eB E^\n) \mb P =0 $, $(\mb H_0 - \eB E^\n)\mb R = - \mb Q$, one finds
\begin{align}
\alpha^{-2}(\mb H_0 - \eB E^\n + \mb V) \xi_b = \alpha^{-2} \ls \Gamma  , \mb V \xi_b \rs \Gamma +\alpha^{-2} \mathbb Q \mathbb V ( \mb R  \mb V)^{b} \Gamma 
\end{align}
where $\Gamma$ denotes the eigenstate associated with $\eB E^\n$. While the second term is of order $O(\alpha^{-b-3})$, the first term is now a multiple of the unperturbed state $\Gamma$ and the coefficients $E_1, \ldots , E_b$ are chosen such that this term is of order $O(\alpha^{-b-3})$. Hence we arrive at $(H_0+K)\Psi_b = O\big(  \alpha^{-b-3}\big)$.

The rigorous version of this argument is based on the next two lemmas. Recall that $K_\ell$, $V_\ell$ and $\mb V= \sum^b_{\ell=1} \alpha^{-\ell} \mb V_\ell$ are all defined in terms of the coefficients $E_\ell$ as given in Theorem \ref{thm:recursive:formula:generalcase}. 
\begin{lem}\label{lem:estimates:approx:groundstate:3x} Let $X_0,Y \in \{ P, (-R)^{1/2}\}$, $X_1=\mathds{1}$, $\mb X_0 = \mb R$ and $\mb X_1 = \mathds{1}$. There exists a constant $C>0$ and for all $\ell , \ell'  \in \mb N_0$ there are $C(\ell), C(\ell,\ell')>0$ such that for all $\Lambda \in [2,\infty]$
\begin{align}
\| X_i K_\ell Y (\mathcal N+1)^{-\min\{1,\frac{\ell}{2} \} } \| & \le C(1+\sqrt \Lambda)^{i}    , \quad i \in \{ 0,1\} \label{eq:XKY:bound}\\[0mm]
\| X_0 (K_\ell - V_\ell) Y (\mathcal N+1)^{-1} \| & \le C (\Lambda^{-1} \ln \Lambda )^{1/2}   \\[0mm]
\| (\mathcal N+1)^{\ell'/2}  \mb X_i \mb V_{\ell} (\mathcal N+1)^{-(\ell' + \ell + 2i ) / 2 }\| & \le C(\ell,\ell')\, \qquad \ \quad i \in \{ 0,1\}. \label{eq:bound:mbKell:2}
\end{align}
\end{lem}
\begin{proof} Recall that $p= - i\nabla_\Omega$. First note that $K$ is at most quadratic in creation and annihilation operators. For $i=0$ one proceeds in close analogy to the proof of Lemma \ref{lem:X:phiY:bound}. That is, we use $\|p P\| + \| p (-R)^{1/2}\| < \infty $ in combination with the estimates from Lemma \ref{Lemma-Regularity_w}, to verify
\begin{align} \label{eq:bound:Kell}
\| X_0 K_\ell Y (\mathcal N+1)^{-\min\{1,\frac{\ell}{2} \} } \| \le C
\end{align}
uniformly in $\Lambda\ge 0$ for all $\ell \ge 1$. For $i=1$, there is no operator on the left of $K_\ell$ that can be used to improve the $\Lambda$-dependence. Here we rely on writing
\begin{align}
K_1 = \phi(v_x + \varphi^\P + p^2 g_x^\Lambda) + 2 a^*( p g_{x}^\Lambda) p + 2  a (p g^\Lambda_{x}) p + 2  a (p^2 g^\Lambda_{x}) 
\end{align}
and using \cite[Lemma 5.1]{FrankS21}, which states that $a^*(f)a( f)\leq \|f\|_{\textnormal{C}}^2\,  p^2 \mathcal{N}$ with the Coulomb norm $\|f\|^2_{\textnormal{C}}:=\frac{1}{4\pi}\int_{\mb R^6} \frac{\overline{f(x)}f(y)}{|x-y|}\mathrm{d}x\mathrm{d}y$. Since $\|p^2 g_x^\Lambda \|_{ \textnormal{C}} \leq \| p g_x^\Lambda \|_2$, this implies
\begin{align}
\|  K_1 Y (\mathcal N +1)^{-1/2} \| \le C   \sup_{x \in \Omega}\Big( \| v_x  + \varphi^\P  +  p^2 g_x^\Lambda  \|_2 + \| p g_x^\Lambda \|_2 \Big).
\end{align}
By definition of $g_x^\Lambda$, we have $p^2 g_x^\Lambda = (\Pi_\Lambda -1 )v_x$ and thus we can apply Lemma \ref{Lemma-Regularity_w} to estimate $\| v_x + \varphi^\P +  p^2 g_x^\Lambda \|_2 + \| pg_x^\Lambda \|_2 \le C(1+\sqrt \Lambda)$.  For $K_\ell$ with $\ell \ge 2$ the bound follows from the usual estimates for creation and annihilation operators together with Lemma \ref{Lemma-Regularity_w}. This completes the bound for $i=1$.

For the second bound, note that by Lemma \ref{Lemma-Regularity_w}, $\| g_x^\Lambda \|_2 + \| p g_x^\Lambda \|_2 \le C (\Lambda^{-1} \ln \Lambda )^{1/2}   $ and 
\begin{align}
| \langle v_x ,  g_x^\Lambda \rangle |    =  \|  p  g_x^\Lambda \|_2^2 \le C \Lambda^{-1} \ln \Lambda 
\end{align} 
for $\Lambda \ge 2 $. This implies that $\| X_0 (K_\ell - V_\ell) Y (\mathcal N+1)^{-\ell} \| \le C  (\Lambda^{-1} \ln \Lambda )^{1/2}   $ for all $\ell \ge 2$. For $\ell =1 $, we use  $v_x + p^2 g_x^\Lambda = \Pi_\Lambda v_x$ and thus
\begin{align}
\| X_0 ( K_{1} - V_1) Y (\mathcal N+1)^{-1/2} \| & \le C  (\Lambda^{-1} \ln \Lambda )^{1/2}   \notag\\
&\hspace{0.5cm} + \| X_0 \phi ( (1-\Pi_\Lambda) v_x ) (\mathcal N+1)^{-1/2} \| .
\end{align}
That $\| X_0 \phi ( (1-\Pi_\Lambda) v_x )Y (\mathcal N+1)^{-1/2} \| \le C  (\Lambda^{-1} \ln \Lambda )^{1/2}   $ follows by the same argument as in the proof of Lemma \ref{lem:X:phiY:bound} with $v_x\mapsto (1-\Pi_\Lambda) v_x$ and $g_x\mapsto -g_x^\Lambda$.

For the third bound, note that $\mb V_{\ell}$ consist of products of operators of the form $X_0 V_k Y$ and that each $\mb V_\ell$ has at most $\ell+2$ creation and annihilation operators.  Since $V_\ell = K_\ell$ for $\Lambda =\infty$, we can apply \eqref{eq:XKY:bound} to infer
\begin{align}\label{eq:Kellbound}
\|(\mathcal N+1)^{\ell'/2} \mb V_{\ell} (\mathcal N+1)^{-(\ell' + \ell + 2 )/2} \| \le C(\ell,\ell').
\end{align}
For $\mb X_0=\mb R$, we use that $ \mb U \mb R \mb U^* $ commutes with $ \mathcal N$ and that $\| \mb U \mb R \mb U^* (\mathcal N+1) \| \le 1 $ by Lemma \ref{lem:diagonalization of HBog}. Invoking  \eqref{eq:Kellbound} and Lemma \ref{lem:number:op:bounds}, we thus have
\begin{align}
& \| (\mathcal N+1)^{\ell'/2} \mb V_\ell \mb R (\mathcal N+1)^{-( \ell' + \ell )/ 2 } \|  \le C(\ell ,\ell' ) \|    (\mathcal N+1)^{(\ell' + \ell + 2 ) /2} \mb R (\mathcal N+1)^{-( \ell' + \ell) / 2 } \|  \notag\\
& \qquad \qquad \le  C(\ell,\ell')  \|    (\mathcal N+1)^{(\ell'+\ell + 2 ) /2} (\mb U \mb R  \mb U^*) \mb U (\mathcal N+1)^{-( \ell' + \ell ) /2 } \|  \notag\\
& \qquad \qquad \le  C(\ell,\ell')  \|    (\mathcal N+1)^{(\ell'+\ell ) /2} \mb U (\mathcal N+1)^{-( \ell' + \ell ) /2 } \| \le C(\ell,\ell').
\end{align}
This completes the proof of the lemma.
\end{proof}
\begin{lem}\label{lem:PKRKP:approximation} Choose $\eB E^{(n_s)}$ as in Proposition \ref{prop:asym:expansion:HG:1} and let $\zeta \in \mathcal F$ satisfy $\mathds{1}(\mathcal N\le m) \mb U \zeta = \mb U \zeta$ with $\mb U$ defined by \eqref{eq: def of U} and some $m\in \mb N$. For every $b\in \mb N_0$ there is a constant $C(b)>0$ such that for all $\Lambda \ge 2 $, $\alpha>0$  
\begin{align}\label{eq:PKRKP:mbV}
\bigg\| \bigg( \sum_{i=0}^{2b+3} P K  ( RK)^i P -  \alpha^{-2} \big( \mb H_0 - \eB E^{(n_s)} + \mb V  \big) \bigg) \zeta  \bigg\| \le C(b) \bigg( \sqrt{\frac{\ln \Lambda}{\Lambda}} \alpha^{-3}  + \alpha^{-b-3} \bigg).
\end{align}
\end{lem}
\begin{proof} By Lemma \ref{lem:Bogoliubov:identity}, we have $
\sum_{i=0}^{2b+3} P K  ( RK)^i P -  \alpha^{-2} \big( \mb H_0 - \eB E^\n  \big) = \mb W_K $ with
\begin{align}
\mb W_K =  \sum_{\ell =3}^{b+2} \alpha^{-\ell} PK_\ell P + \sum_{\substack { k,\ell =1 \\ k+\ell \ge 3 } }^{b+2} \alpha^{-\ell - k} PK_\ell P K_k P + \sum_{i=2}^{2b+3} P K  ( RK)^i P.
\end{align}
We compare this operator with
\begin{align}
\mb W_V =  \sum_{\ell =3}^{b+2} \alpha^{-\ell} PV_\ell P + \sum_{\substack { k,\ell =1 \\ k+\ell \ge 3 } }^{b+2} \alpha^{-\ell-k} PV_\ell R V_k P + \sum_{i=2}^{2b+3} P V  ( R V )^i P,
\end{align}
where $V= \sum_{\ell=1}^{b+2} \alpha^{-\ell} V_\ell$. By Lemma \ref{lem:estimates:approx:groundstate:3x}, each replacement of $K_\ell$ by $V_\ell$ comes at  the cost of an error proportional to $(\Lambda^{-1} \ln \Lambda)^{1/2}$,
\begin{align}
& \big\| ( \mb W_K -  \mb W_V ) \zeta \big\|  \le C(b) \alpha^{-3} (\Lambda^{-1} \ln \Lambda)^{1/2} \| (\mathcal N+1)^{(2b+4)} \zeta \|.
\end{align}
To proceed we sort all terms in $\mb W_V$ according to powers of $\alpha^{-1}$,
\begin{align}
\sum_{\ell =3}^{b+2} \alpha^{-\ell} PV_\ell P &  = \sum_{\ell = 3}^{b+2} \alpha^{-\ell} \sum_{\substack{ \peps \in \mb N \\ |\peps | = \ell } } P V_{\epsilon_1} P \\ 
\sum_{\substack { k,\ell =1 \\ k+\ell \ge 3 } }^{b+2} \alpha^{-\ell-k} PV_\ell R V_k P & = \sum_{\ell = 3}^{2(b+2)} \alpha^{-\ell} \sum_{\substack{ \peps \in \mb N^2 \\ |\peps | = \ell } } P V_{\epsilon_1} (R V_{\epsilon_2})  P
\end{align}
and similarly for the last term,
\begin{align} 
 \sum_{i=2}^{2b+3} P V  ( R V )^i P & = \sum_{i= 3  }^{2b+4}  \sum_{\ell = 3 }^{(b+2)(2b+4)} \alpha^{-\ell}  \sum_{\substack { \peps \in \mb N^i  \\ |\peps | =\ell }}  P V_{\varepsilon_1} (R V_{\varepsilon_2}) \ldots (R V_{\varepsilon_{i}}) P , 
\end{align}
where we used that all summands with $3 \le \ell \le i-1$ are zero. Adding the three contributions and using \eqref{eq:def:mbV:a}, we obtain 
\begin{align}\label{eq:W:V:identity}
\mb W_V & =  \sum_{\ell=3}^{b+2} \alpha^{-\ell} \mb V_{\ell-2} + \sum_{i=2}^{2b+4} \sum_{\ell = b+3}^{(b+2)(2b+4)} \alpha^{-\ell} \sum_{\substack { \peps \in \mb N^i  \\ |\peps | =\ell }}  P V_{\varepsilon_1} (R V_{\varepsilon_2}) \ldots (R V_{\varepsilon_{i}}) P,
\end{align}
where we also used that $PV_{\varepsilon_1}RV_{\varepsilon_2} R = 0$ for $\varepsilon_1+\varepsilon_2 \ge 3$. With $\alpha^{-2}\mb V =  \sum_{\ell=3}^{b+2} \alpha^{-\ell} \mb V_{\ell-2}$ and since the remainder can be estimated by means of Eq. \eqref{eq:XKY:bound}, we arrive at
\begin{align}
& \big\| ( \mb W_V - \alpha^{-2} \mb V)\zeta  \big\|  \le \frac{C(b)}{\alpha^{b+3}}  \big \| (\mathcal N+1)^{(b+2)(2b+4)/2} \zeta \big\|  .
\end{align}
Finally, we use Lemma \ref{lem:number:op:bounds} and the assumption on $\zeta$ to estimate
\begin{align}
\| (\mathcal N+1)^{(b+2)(2b+4)/2} \zeta \| \le C(b)\| (\mathcal N+1)^{(b+2)(2b+4)/2} \mb U \zeta \| \le C(b).
\end{align}
This concludes the proof of the lemma.
\end{proof}

We are now prepared to prove Proposition \ref{prop:asym:expansion:HG:1}.

\begin{proof}[Proof of Proposition \ref{prop:asym:expansion:HG:1}] The starting point is Eq. \eqref{eq:heuristic:discussion:mbK:1} for $\xi=\xi_b$, 
\begin{align}
&  (H_0 +K) \Psi_{b} =    \bigg( \sum_{i=0}^{2b+3} PK (  R K  )^i \bigg)   \psi^\P \otimes \xi_{b} +  Q  K (  R K  )^{2b+3}   \psi^\P \otimes \xi_b, \label{eq:heuristic:argument:1}
\end{align}
where we apply Lemma  \ref{lem:PKRKP:approximation} to the first term (note that $\xi_b$ satisfies the required condition for some $b$-dependent value of $m$) and Lemma \ref{lem:estimates:approx:groundstate:3x} to the second term,
\begin{align}\label{eq:bound:H_0+K:1}
\| ( H_0 + K) \Psi_b \| & \le \alpha^{-2}  \big \|  (\mb H_0 - \eB E^{(n_s)} + \mb V ) \psi^\P \otimes \xi_b \big\| \notag\\
&  \hspace{2cm} + C(b) \bigg(\, \sqrt{\frac{\ln \Lambda}{\Lambda}} \alpha^{-3} + \alpha^{-b-3} + \frac{\sqrt \Lambda  }{\alpha^{2b+4} }   \bigg) .
\end{align}
Choosing $\Lambda = \alpha^{2b+2}$, the error in the second line is bounded by $ C(b) \alpha^{-b-3}$.

In the first term, we use $\mathds{1} = \mb P + \mb Q$, $\mb P(\mb H_0 - \eB E^{(n_s)}) = 0$ and $\mb R (\mb H_0 - \eB E^{(n_s)}) = - \mb Q$ to obtain
\begin{align}
 (\mb H_0 - \eB E^{(n_s)} + \mb V ) \psi^\P \otimes \xi_b = (P\otimes \mb P) \mb V \psi^\P \otimes \xi_b + \mb Q \mb V (\mb R \mb V)^{b} \psi^\P \otimes \ms U \Gamma^{(n_s)}.
\end{align}
With the aid of Lemmas \ref{lem:estimates:approx:groundstate:3x} and \ref{lem:number:op:bounds} and using that $\ms U \Gamma^{(n_s)} = \mathbb U^* \ms U \gamma^{(n_s)} $ with $\mathds 1(\mathcal N\le  m) \ms U \gamma^{(n_s)}  = \ms U \gamma^{(n_s)} $ for some $m\ge 1$, we can estimate the second term by
\begin{align}\label{eq:bound:H_0+K:2}
\big\|  \mb Q \mb V (\mb R \mb V)^{b} \psi^\P \otimes \ms U \Gamma^{(n_s)} \| \le  \frac{C(b)}{\alpha^{b+1}}  \| (\mathcal N+1)^{ (b+2+b^2)/2} \mb U^* \ms U  \gamma^{(n_s)}  \| \le  \frac{ C(b) }{  \alpha^{b+1}}.
\end{align}

In the remainder, we shall show that $ \| (P \otimes \mb P  ) \mb V \psi^\P \otimes \xi_b \|   \,  \le \, C(b) \alpha^{-b-1} $. We recall that $\xi_b =  \sum_{j=0}^{b}  ( \mb R  \mb V )^j   \ms U \Gamma^{(n_s)}$ for $s\in \{1,\ldots, \d\}$ with $\{ \Gamma^{(n_1)} , \ldots, \Gamma^{(n_\d)}\}$ an orthonormal basis of $\text{Ran}\mb P$ and $\ms U$ a unitary acting on this subspace. We first observe that
\begin{align}
& \| (P \otimes \mb P  ) \mb V \psi^\P \otimes \xi_b \|  \notag \\
&\qquad =  \sum_{r=1}^\d  \bigg|\sum^{b(b+1)}_{\ell=1} \alpha^{- \ell }   \bigg\langle \ms U \Gamma^{(n_r) } ,  \bigg[  \sum_{i=1}^{\min \{ \ell , b+1 \} } \sum_{\substack{ \peps \in \mb N^i \\ |\peps|=\ell } } \,  \mb V_{ \varepsilon_1 } (  \mb R  \mb V_{ \varepsilon_2} ) \ldots ( \mb R \mb V_{ \varepsilon_i})  \bigg] \ms U \Gamma^{(n_s) }    \bigg\rangle \bigg| , \label{eq:R:k+2:identity:b:1}
\end{align}
which follows  from expanding $\mb V \xi_b$ in powers of $\alpha^{-1}$,
\begin{align}
\mb P \mb V   \xi_b   & =    \sum_{\ell=1}^{b(b+1)} \frac{1}{\alpha^\ell} \sum_{i=1}^{\min\{ \ell , b+1 \} } \sum_{\substack{ \peps \in \mb N^i \\ |\peps|=\ell } }  (  \mb P  \mb V_{ \varepsilon_1} )( \mb R \mb V_{ \varepsilon_2}) \ldots ( \mb R \mb V_{ \varepsilon_i})  \ms U \Gamma^{(n_s)} ,
\end{align}
and using $ \mb P = \sum_{r=1}^\d  | \ms U \Gamma^{(n_r)} \rangle \langle 
\ms U \Gamma^{(n_r)}|$. Note that in  \eqref{eq:R:k+2:identity:b:1} we omit the trivial tensor product with $\psi^\P$ on both sides of the scalar product.

By separating the term with $i=1$ in \eqref{eq:R:k+2:identity:b:1} and using $\mb V_\ell = - E_\ell + \widetilde {\mb V}_\ell$, cf. \eqref{eq:def:mbV:b}, we can write
\begin{align}\label{PPVxi-term}
 \| (P \otimes \mb P  ) \mb V \psi^\P \otimes \xi_b \|   & = \sum_{r=1}^\d \bigg|  \Ls \ms  U \Gamma^{(n_r)} , \sum_{\ell = 1}^{b(b+1)} \alpha^{-\ell} (M_\ell - E_\ell ) \ms U \Gamma^{(n_s)} \Rs \bigg| 
\end{align}
with
\begin{align}
M_\ell : = \mb P \bigg[  \widetilde {\mb V}_\ell + \sum_{i=2}^{\min\{ \ell , b+1\} } \sum_{\substack{ \peps \in \mb N^i \\ |\peps|=\ell } } \,  \mb V_{ \varepsilon_1 } (  \mb R  \mb V_{ \varepsilon_2} ) \ldots ( \mb R \mb V_{ \varepsilon_i})  \bigg]  \mb P .
\end{align} 
The terms with $\ell \ge b+1$ in \eqref{PPVxi-term} can be estimated similarly as in \eqref{eq:bound:H_0+K:2}, using Lemmas \ref{lem:estimates:approx:groundstate:3x} and \ref{lem:number:op:bounds}, by 
\begin{align}
\bigg\| \sum_{\ell = b+1}^{b(b+1)} \alpha^{-\ell} ( M_\ell - E_\ell)  \ms U \Gamma^{(n_s)} \bigg\| \le \frac{C(b)}{\alpha^{b+1}}  \| (\mathcal N+1)^{(b+2+b^2)/2} \mb U^* \ms U \gamma^{(n_s)}  \| \le  \frac{C(b)}{ \alpha^{b+1}} .
\end{align}
For the contributions with $1 \le  \ell \le b$, we write $
M_\ell = \sum_{r,t=1}^\d ( M_\ell)_{rt} |\Gamma^{(n_r)} \rangle \langle \Gamma^{(n_t)}|$
with $(M_\ell)_{rt}$ the matrix coefficients defined in \eqref{eq:matrix:definitions}. Moreover, we set $M^{(b)} := \sum_{\ell =1 }^b \alpha^{-\ell }M_\ell$ and choose $\ms U$ such that $\ms U^*  M^{(b)} \ms U = \sum_{s=1}^\d \mu^{(s)}(M^{(b)}) |\Gamma^{(n_s)} \rangle \langle \Gamma^{(n_s)}|$ with eigenvalues $\mu^{(1)}(M^{(b)})\le \ldots \le \mu^{(\d)}(M^{(b)})$. Thus, we obtain 
\begin{align}
 \sum_{r=1}^\d \bigg|  \Ls \ms  U \Gamma^{(n_r)} , \sum_{\ell = 1}^{b} \alpha^{-\ell} (M_\ell - E_\ell ) \ms U \Gamma^{(n_s)} \Rs \bigg|   & =   \bigg| \mu^{(s)}(  M^{(b)}  ) - \sum_{\ell=1}^b \alpha^{-\ell} E_\ell \bigg|.
\end{align}
As explained in Remark \ref{rem:eigenvalue:M}, the eigenvalue $\mu^{(s)}(M^{(b)})$ has a power series expansion that converges for $\alpha \ge \alpha(b)$ for some large enough $\alpha(b)$ with coefficients $\mu^{(s)}_\ell$. Since $\lim_{\alpha \to \infty } \alpha^b ( \mu^{(s)}(M^{(b)}) - \sum_{\ell=1}^{b} \alpha^{-\ell} E_\ell ) = 0$ by definition of $E_b$, we have
\begin{align}
 \mu^{(s)}(M^{(b)}) - \sum_{\ell =1 }^b \alpha^{-\ell}  E_\ell  =  \sum_{\ell =b+1}^\infty \alpha^{-\ell}  \mu^{(s)}_\ell ,
\end{align}
where the remainder is bounded by $C(b) \alpha^{-b-1}$. (For completeness let us note that the remainder bound holds trivially for all $\Lambda$ since $M^{(b)}$ is $\Lambda$-independent by definition). This competes the proof of the lemma.
\end{proof}

\subsection{Proofs of Theorems \ref{thm:asymptotic:expansion:1} and \ref{thm:recursive:formula:generalcase}}

In this section we use Propositions \ref{prop:two:term:expansion} and \ref{prop:asym:expansion:HG:1} to prove Theorems \ref{thm:asymptotic:expansion:1} and \ref{thm:recursive:formula:generalcase}. Since $E_\ell$ is chosen as specified in Eq. \eqref{eq:limit:formula}, it is sufficient to prove Theorem \ref{thm:asymptotic:expansion:1}.

Before we begin, we recall the following simple fact. Consider a self-adjoint operator $A$ on a Hilbert space $\ms H$ with discrete spectrum and count the eigenvalues as $\mu^{(1)}(A) \le \mu^{(2)}(A)\le \ldots$ Assume that a set of normalized states $u_1,\ldots,u_n \in D(A) $ and two positive numbers $\varepsilon,\delta$ satisfy $| \langle u_i , u_j\rangle | \le  \varepsilon$ for all $i\neq j \in \{ 1,\ldots , n\} $ and $\| A u _i \| \le \delta$ for all $i \in \{ 1,\ldots , n\}$. Then, there exists a $\delta$-independent constant $C_\varepsilon>0$ such that $A$ has at least $n$ (possibly equivalent) eigenvalues in the interval $[-C_\varepsilon \delta ,C_\varepsilon \delta]$. Note that $C_\varepsilon =1$ in case $n=1$.

\begin{proof}[Proof of Theorem \ref{thm:asymptotic:expansion:1}] To enhance the clarity of the argument, let us first focus on eigenvalues $\ms E^\n(\alpha)$ such that $\eB E^\n$ is non-degenerate. We introduce $
\varepsilon^{(1)} : =  \eB E^{(2)} - \eB E^{(1)} >0$ and $
\varepsilon^\n  := \min \{ \eB E^{(n+1)} - \eB E^\n ,  \eB E^\n - \eB E^{(n-1)} \} > 0$
for $n\ge 2$. Moreover, we recall the unitary equivalence $\mf H_{\alpha,\Lambda}^{\textnormal{G}} - e^\Pek \cong \mf H_\alpha $ by \eqref{eq:Gross:transformed:Froehlich} and observe that $\| \Psi_b \| \ge 1$. With this at hand, we can apply Proposition \ref{prop:asym:expansion:HG:1} to see that the spectrum $\sigma (\mathfrak H_\alpha - e^\Pek)$ is non-empty around the approximate eigenvalue $\ms E_b^\n (\alpha) :=  \alpha^{-2} \eB E^\n +  \sum_{\ell=1}^b \alpha^{-\ell-2} E_\ell $. More precisely,
\begin{align}
 \sigma \big(\mf H_\alpha-e^\Pek \big)  \cap \Big[ \ms E_b^\n(\alpha)  - \frac{C(b)}{ \alpha^{b+3} } ,  \ms E_b^\n(\alpha)+ \frac{ C(b) }{\alpha^{b+3}} \Big]  
\end{align}
contains at least one eigenvalue.
For $\alpha \ge \alpha(b)$ with large enough $\alpha(b)$, we further have
\begin{align}
\Big[ \ms E_b^\n(\alpha)  - \frac{C(b)}{ \alpha^{b+3} } ,  \ms E_b^\n(\alpha)+ \frac{ C(b) }{\alpha^{b+3}} \Big]   \subset \Big[ \alpha^{-2} \big( \eB E^\n -  \tfrac13 \varepsilon^\n  \big) , \alpha^{-2} \big( \eB E^\n +  \tfrac13 \varepsilon^\n \big) \Big].
\end{align}
By Proposition \ref{prop:two:term:expansion}, on the other hand, we know that the only eigenvalue in the interval on the right side is $\ms E^\n(\alpha) - e^\Pek$. Concretely, we apply Proposition \ref{prop:two:term:expansion} to the neighboring eigenvalues to obtain
\begin{align}
 \ms E^{(n+1)} (\alpha) - e^\Pek & \ge  \alpha^{-2} \eB E^\n +\alpha^{-2} \varepsilon^\n - C \alpha^{-2-\delta} \ge  \alpha^{-2} \big( \eB E^\n  + \tfrac12 \varepsilon^\n \big) \\
 \ms E^{(n-1)}(\alpha) - e^\Pek & \le  \alpha^{-2}   \eB E^\n + \alpha^{-2} \varepsilon^\n  -  C \alpha^{-2-\delta} \le  \alpha^{-2} \big( \eB E^\n - \tfrac12 \varepsilon^\n \big) 
\end{align} 
for all $\alpha \ge \alpha_0$ for suitable $\alpha_0>0$ and $\delta>0$. Thus, we can conclude that $\sigma ( \mf H_\alpha-e^\Pek )  \cap [ \ms E^\n_b (\alpha)    - C(b)  \alpha^{-b-3} ,   \ms E^\n_b (\alpha)  + C(b)  \alpha^{b+3} ] = \{ \ms E^\n(\alpha) -e^\Pek \}$, which completes the proof of Theorem \ref{thm:asymptotic:expansion:1} in the non-degenerate case.

Since the general case can be treated in complete analogy, we shall present the argument for a two-fold degenerate eigenvalue, that is, for $ \eB E^{(n-1)} < \eB E^\n = \eB E^{(n+1)}  <  \eB E^{(n+2)}$.
In this case, we have two sets of coefficients and two approximate eigenstates \eqref{eq:Phi:n:m:1}, which we denote by $E^\n_{\ell}$, $ E^{(n+1)}_{\ell}$ and $\Psi^\n_b$, $\Psi_b^{(n+1)}$, respectively. Using Lemma \ref{lem:estimates:approx:groundstate:3x}, it follows that $ \| \Psi^{(j)}_b
 - \psi^\P \otimes \ms U \Gamma^{(j)} \| \to 0$ as $\alpha \to \infty$ for $j\in \{n,n+1\}$ and thus, $|\langle \Psi^{(n)}_b , \Psi^{(n+1)}_b \rangle |  \le \tfrac12 \| \Psi^{(n)}_b \| \, \| \Psi^{(n+1)}_b \| $ for large $\alpha$, since $\langle \Gamma^{(n)} , \Gamma^{(n+1)}\rangle = 0 $ by assumption and $\| \Psi^{(j)}_b\| \ge 1$. 
 
Now assume that $E_{\ell}^\n = E_{\ell}^{(n+1)}$ for $1\le \ell \le b$ and set $\ms E_b^\n(\alpha) = \sum_{\ell=0}^b  \alpha^{-2-\ell} E^{(n)}_{\ell}  $. As a consequence of Proposition \ref{prop:asym:expansion:HG:1}, we know that there is a constant $C(b)$ such that 
\begin{align}
 \sigma\big(\mf H_\alpha-e^\Pek \big)  \cap  \Big[ \ms E_b^\n(\alpha)  - \frac{C(b)}{ \alpha^{b+3} } ,  \ms E_b^\n(\alpha)+ \frac{ C(b) }{\alpha^{b+3}} \Big]  
\end{align}
contains at least two eigenvalues. In analogy to the non-degenerate case, we know by Proposition \ref{prop:two:term:expansion} that the only possible elements of the spectrum in the above set are the eigenvalues $\ms E^{(n)} (\alpha)-e^\Pek$, $\ms E^{(n+1)} (\alpha)-e^\Pek $. This proves the statement if $E_{\ell}^\n = E_{\ell}^{(n+1)}$ for all $1\le \ell \le b$. 

On the other hand, if there is an integer $1\le \ell_0 \le  b$ with $E^\n_{\ell_0} \neq E_{\ell_0}^{(n+1)}$, then the two intervals 
\begin{align}
I^{(j)}:= \Big[\ms E^{(j)}_b(\alpha) - \frac{C(b)}{ \alpha^{3+b}} , \ms E^{(j)}_b(\alpha) + \frac{ C(b) } { \alpha^{3+b} } \Big] , \quad j\in \{ n , n+1\}
\end{align}
with $\ms E^{(j)}_b(\alpha) =  \sum_{\ell=0}^b  \alpha^{-2-\ell} E^{(j)}_{\ell}  $ are disjoint for large $\alpha$. By Proposition \ref{prop:asym:expansion:HG:1}, we know that each set $\sigma(\mf H_\alpha - e^\Pek) \cap I^{(j)}$, for $j=n$ and $j=n+1$, contains at least one eigenvalue.  By Proposition \ref{prop:two:term:expansion}, we know again that the only possible elements of the spectrum in $\sigma(\mf H_\alpha - e^\Pek) \cap (I^{(n)} \cup I^{(n+1)}) $ are the eigenvalues $\ms E^\n(\alpha)-e^\Pek$, $\ms E^{(n+1)}(\alpha)-e^\Pek$. Since $E^\n_{\ell_0} < E^{(n+1)}_{\ell_0}$ by definition and $\ms E^\n(\alpha) \le \ms E^{(n+1)}(\alpha)$ by assumption, this implies that $\sigma (\mf H_\alpha - e^\Pek) \cap I^{(j)} = \{ \ms E^{(j)}(\alpha) - e^\Pek \}$ for $j\in \{ n,n+1\}$. This  completes the proof of the theorem in the two-fold degenerate case. Extending the argument to higher degeneracy is straightforward.
\end{proof}

\section{Improved remainder estimates}

\label{sec:improved:remainder:estimates}

In this section, we provide the proof of Theorem \ref{thm:asymptotic:expansion:2}. To this end, we optimize the form of the approximate eigenstates and establish better combinatorial estimates than the ones used in Section \ref{sec:approximate:eigenstates}.

For this, we assume that $\eB E^\n$ is a non-degenerate eigenvalue of $\mb H_0$ and $\Gamma$ the corresponding normalized eigenstate. Throughout this section, we define the sequence $(E_\ell)_{\ell \in \mb N_0}$ in the following way. Let $K_\ell$ be defined by \eqref{eq:def:K_1}--\eqref{eq:def:K:ell+1} with yet-to-be-determined coefficients $E_\ell$ and consider the Fock space operators
\begin{align}\label{def:mb:V:k}
\mb K_\ell  = \mb K_\ell \big(  E_0, \ldots, E_{\ell } \big) & :=   \sum_{i=1}^{\ell+2} \sum_{\substack{ \peps \in \mathbb N^i \\ |\peps|  = \ell+2 }}  P  K_{\varepsilon_1}  (R K_{\varepsilon_2} )\ldots  (R K_{\varepsilon_i}) P 
\end{align}
and $\widetilde {\mb K}_\ell = \widetilde{ \mb K}_\ell (E_0 , \ldots, E_{\ell-2})   : = \mb K_\ell + E_{\ell}$ for $\ell\ge 1$. For $b\in \mb N_0$, we then iteratively define $E_0 := \eB E^\n$ and for $\ell \ge 1$
\begin{align}
\label{eq:def:E:Lambda}
E_{\ell} & : = \bigg\langle \psi^\P \otimes \Gamma  , \bigg[  \widetilde {\mb K}_{\ell}  + \sum_{j=2}^{\ell} \sum_{ \substack { \peps \in \mb N^j \\ |\peps| = \ell } }  \mb K_{ \varepsilon_1 } ( \mb R  \mb K_{ \varepsilon_2} ) \ldots ( \mb R \mb K_{ \varepsilon_j})\bigg]  \psi^\P \otimes \Gamma  \bigg\rangle.
\end{align}
We shall derive the statement of Theorem \ref{thm:asymptotic:expansion:2} for this choice of $(E_\ell)_{\ell \in \mb N_0}$. By uniqueness of the asymptotic series, we can conclude a posteriori that the coefficients \eqref{eq:def:E:Lambda} are identical to the ones given in Theorem \ref{thm:non-degnerate-formula}.

For $b\in \mb N_0$, we consider the state
\begin{align}\label{eq:approximate:ground:state}
\Phi_b = \Phi_b (E_0 ,\ldots,E_{b}) = \sum_{\ell = 0 }^{b+2 } \frac{1}{\alpha^\ell } \Upsilon_{\ell} 
\end{align}
with $\Upsilon_0 = \psi^\P \otimes  \Gamma $
and
\begin{align}
\Upsilon_{\ell} & =  \sum_{i,j=0}^{\ell} \underset{ |\peps| + |\pmb k| = \ell}{ \sum_{\peps \in \mathbb N^i  }\sum_{  \pmb k \in \mathbb N^j }  } (R K _{\varepsilon_1})\ldots  ( R K_{\varepsilon_i})   ( \mb R  \mb K_{k_1})\ldots  (  \mb R \mb K_{k_j})  \Upsilon_0, \label{eq:def:Upsilon:ell}
\end{align}
where $\pmb k = (k_1,\ldots ,k_j)$, $|\pmb k| = k_1 + \ldots + k_j$ and similarly for $\peps$. 

For later use, let us also introduce the state
\begin{align}\label{eq:deg:widetilde:Phi}
\widetilde \Phi_b = \widetilde \Phi_b(E_0,\ldots, E_b) =    \bigg( \sum_{i=0}^{b+2} ( R  K )^i \, \bigg)  \psi^\P \otimes \zeta_b \quad \text{with} \quad \zeta_b =   \sum_{j=0}^{b+2}  ( \mb R  \mb K )^j   \Gamma
\end{align}
and
\begin{align}\label{eq:def:K:mbK}
K = \sum_{\ell=1}^{b+2} \alpha^{-\ell} K_\ell, \qquad
\mb K := \sum_{\ell=1}^{b} \alpha^{-\ell} \mb K_\ell.
\end{align}
The state $\Phi_b$ is obtained by expanding $\widetilde \Phi_b$ in powers of $\alpha^{-1}$ and keeping only the terms up to order $O(\alpha^{-b-2})$. Consequently, $ \Phi_b - \widetilde \Phi_b$ has only contributions of order $O(\alpha^{-b-3})$ and higher, a property that will be repeatedly used in the proof.

\begin{prop} \label{prop:asym:expansion:HG} Let $\mf H_{\alpha,\Lambda}^{\textnormal{G}}$ be defined by \eqref{eq:Gross:transformed:Froehlich}. Choose $n\in\mb N$ such that the $nth$ eigenvalue of $\mb H_0$ is non-degenerate with normalized eigenfunction $\Gamma \in \mathcal F$. There exist constants $C,\alpha_0>0$ so that the sequence $(E_\ell)_{\ell\in \mb N_0}$ defined by \eqref{eq:def:E:Lambda} and the state $\Phi_b$ given by \eqref{eq:approximate:ground:state} satisfy
\begin{align}\label{eq:bound:ev:equation:H_G}
 \bigg\| \bigg( \mf H_{\alpha,\Lambda}^{\textnormal{G}} -     \sum_{\ell =1}^b \alpha^{-\ell-2} E_{\ell}  \bigg)   \Psi_b \bigg\| \le \frac{C^{b+1} (b+1)!}{\alpha^{b+3}} \big( 1 + \sqrt \Lambda \big)
\end{align}
for all $\Lambda \ge 2$, $b\in \mb N_0$, $\alpha \ge \alpha_0$.
\end{prop} 

With the proposition at hand, we can prove the remainder estimate in Theorem \ref{thm:asymptotic:expansion:2}. The bounds on the coefficients $E_\ell$ are given in Lemma \ref{lem:estimates:approx:groundstate:2} below.

\begin{proof}[Proof of Theorem \ref{thm:asymptotic:expansion:2}] Applying Proposition \ref{prop:asym:expansion:HG} for $\Lambda =2$ and following the same steps as in the proof of Theorem \ref{thm:asymptotic:expansion:1}, one obtains the desired estimate for all $\alpha \ge b /\varepsilon$ with $\varepsilon>0$ sufficiently small. To extend the statement to all $\alpha \ge \alpha_0$, one uses $| \alpha^2 \ms E^\n(\alpha) - \alpha^2 e^\Pek -  \eB E^\n| \le 1$ by Proposition \ref{prop:two:term:expansion} and $|E_{\ell}| \le C^\ell \sqrt{ \ell !} $ by Lemma \ref{lem:estimates:approx:groundstate:2} below, so that 
\begin{align}
 \bigg|  \alpha^2 \ms E^\n(\alpha) - \alpha^2 e^\Pek -  \sum_{\ell=0}^b \alpha^{-\ell} E_{\ell} \bigg|  \le 1 + \sum_{\ell=1}^b \alpha^{-\ell} C^\ell \sqrt{ \ell!}  \le \frac{C^{b+1} (b+1)! }{\alpha^{b+1}}
\end{align}
for $\alpha \le b /\varepsilon$.
\end{proof}

\subsection{Proof of Proposition \ref{prop:asym:expansion:HG}}\label{sec:proof:asym:expansion:HG}

The next lemma provides a formula for the state $(H_0+K)\Phi_b$. For completeness, let us note that the statement holds independently of the specific choice of $(E_\ell)_{\ell \in \mb N_0}$. 

\begin{lem}\label{lem:identity} For $b\in \mb N_0$ let $\Phi_b$ and $K$ be defined by \eqref{eq:approximate:ground:state} and \eqref{eq:def:K:mbK} and consider the two polynomials $\mathfrak p_j :\mb R_+ \to L^2(\Omega) \otimes \mathcal F$ given by $\mathfrak p_1(1/\alpha) = (H_0 +K) \Phi_b$ and 
\begin{align}
\mf p_2(1/\alpha) = ( P\otimes \mathbb P)  K \Phi_b + ( Q +  P \otimes \mb Q )\sum_{\substack{ \ell , k =1 \\ k+l \ge b+3} }^{b+2} \frac{1}{\alpha^{k+l}}  K_\ell  \Upsilon_k.
\end{align}
It holds that $\mf p_1 = \mf p_2$.
\end{lem}

\begin{proof} Write $\mf p_1(1/\alpha) = (P\otimes \mb P + P \otimes \mb Q + Q )(H_0+K)\Phi_b$. Since the three terms are pairwise orthogonal, it is sufficient to show the identity for each term separately. For the term along $P\otimes \mb P$ the identity follows from $P H_0=0$. For the other two terms, we will employ the fact that a Hilbert space valued polynomial $\psi : \mb R_+ \to \ms H$, $\alpha^{-1} \mapsto  \psi(\alpha^{-1}) = \sum_{\ell=1}^n \alpha^{-\ell} \psi_\ell$ vanishes identically if and only if $\psi_\ell=0$ for $\ell=1,\ldots, n$.

Recalling the auxiliary state $\widetilde \Phi_b$ introduced in \eqref{eq:deg:widetilde:Phi}, we compute
\begin{align}
Q(H_0+K) \Phi_b & = Q(H_0 + K )  \widetilde \Phi_b + Q ( H_0 + K ) (\Phi_b - \widetilde \Phi_b)  \notag\\[1mm]
& = QK (RK)^{b+2}  \psi \otimes \zeta_b + Q ( H_0 + K )  (\Phi_b - \widetilde \Phi_b) =: \sum^{m}_{\ell \ge b+3}\alpha^{-\ell} \varphi_\ell
\end{align}
for some integer $m \ge b+3$ and $\alpha$-independent $\varphi_\ell \in L^2(\Omega) \otimes \mathcal F$. The second step followed from $H_0 R = -Q$, while in the last one, we used $K=O(\alpha^{-1})$ and $\Phi_b - \widetilde \Phi_b = O(\alpha^{-b-3})$. This shows that all coefficients of $Q (H_0 +K ) \Phi_b$ proportional to $\alpha^{-\ell}$ for $\ell=0,\ldots ,b+2$ vanish, and thus 
\begin{align}
Q (H_0 +K ) \Phi_b = \sum_{\ell = 0 }^{b+2} \frac{1}{\alpha^\ell}  Q \Big( H_0 + \sum_{k=1}^{b+2} K_k \Big)  \Upsilon _{\ell}  =   \sum_{\substack{ k,l=1 \\ k + \ell  \ge b+3 } }^{b+2}\frac{1}{\alpha^{k+\ell }}  Q K_k  \Upsilon_\ell.
\end{align}
This proves the identity for the contribution along $Q$.

For the last term, we use the identity
\begin{align}\label{eq:PVRV:Bod:identity}
\sum_{i=0}^{b+2}   P K ( R K )^i P   = \alpha^{-2} \big(  \mb H_0 - \eB E^\n + \mb K  \big) + \mb K_{\textnormal{rem}} ,
\end{align}
where $\mb K = \sum_{\ell = 1}^{b} \alpha^{-\ell} \mb K_\ell $ and $\mb K_{\textnormal{rem}}$ is defined as the remainder term in \eqref{eq:W:V:identity} for $\Lambda =\infty$. In particular, $\mb K_{\textnormal{rem}}$ is of order $O(\alpha^{-b-3})$ and higher. Combining this with $PH_0=0$ and using $\widetilde \Phi_b = \sum_{i=0}^{b+1} (RK)^i \psi^\P \otimes \zeta_b$ with $\zeta_b =   \sum_{j=0}^{b+2}  ( \mb R  \mb K )^j  \Gamma $, c.f. \eqref{eq:deg:widetilde:Phi}, one finds
\begin{align}
& ( P\otimes \mb Q ) (H_0 +K ) \Phi_b \notag = ( P\otimes \mb Q ) K  \widetilde \Phi_b + (
P\otimes \mb Q )   K   (\Phi_b - \widetilde \Phi_b ) \notag \\[1mm]
 & \qquad =\psi^\P \otimes \mb Q \big(  \alpha^{-2} ( \mb H_0 - \eB E^\n + \mb K ) +  \mb K_{\textnormal{rem}} \big)  \zeta_b + 
( P\otimes \mb Q ) K (\Phi_b - \widetilde \Phi_b) \notag \\[2mm]
 &\qquad = 
\alpha^{-2} \psi^\P \otimes \mb Q    \mb K  ( \mb R \mb K )^{b+2}  \Gamma  + \psi^\P   \otimes \mb Q  \mb K_{\textnormal{rem}} \zeta_b  + (P\otimes \mb Q )  K (\Phi_b- \widetilde \Phi_b),
\end{align}
where we used $(\mb H_0- \eB E^\n) \mb R = -\mb Q$. Since $\mb K=O(\alpha^{-1})$, $\mb K_{\textnormal {rem}} =O(\alpha^{-b-3})$, $K=O(\alpha^{-1})$ and $\Phi_b-\widetilde \Phi_b=O(\alpha^{-b-3})$, the right-hand side has no contributions of order $\alpha^{-\ell}$ for $\ell =0 ,\ldots,b+2$. A comparison of the coefficients thus yields that
\begin{align}
( P\otimes \mb Q ) (H_0 + K) \Phi_b = (P\otimes \mb  Q )  \sum_{\substack{ k,\ell=1 \\ k+\ell \ge b+3 } }^{b+2} \alpha^{-k-\ell }   K_k \Upsilon_{\ell}. \label{eq:PQVPsi}
\end{align}
This shows the identity for the last term and thus completes the proof of the lemma.
\end{proof}

\begin{lem}\label{lem:energy:identity} For $\ell \ge 1$, the coefficients $E_\ell$ defined by \eqref{eq:def:E:Lambda} satisfy
\begin{align}
E_\ell =  \Ls \Upsilon_0, \widetilde K_{\ell+2} \Upsilon_0 \Rs + \sum_{k=1}^{\ell +1 } \Ls \Upsilon_0, K_k \Upsilon_{\ell+2-k} \Rs
\end{align}
with $\Upsilon_\ell$ as in \eqref{eq:def:Upsilon:ell}, $ \widetilde K_1 = K_1$, $\widetilde K_\ell = K_\ell +E_{\ell+2}$ for $\ell \in \{2,3,4\}$ and $\widetilde K_\ell = 0$ otherwise.
\end{lem}
\begin{proof} We first compute
\begin{align}\label{eq:coefficients:comparison}
\Ls \Upsilon_0, K \Phi_b \Rs & = \sum_{\ell=1}^{b} \alpha^{-\ell-2}  \sum_{k=1}^{\ell+2} \Ls \Upsilon_0, K_k \Upsilon_{\ell+2-k} \Rs + \sum_{\substack{k=1,\ell=0 \\ k+\ell \ge b+3}}^{b+2}    \frac{1}{\alpha^{k+\ell}}  \Big\langle \Upsilon_0 , K_k \Upsilon_\ell \Rs ,
\end{align}
where we used $ \langle \Upsilon_0, K_1 \Upsilon_0 \rangle = 0 = \langle \Upsilon_0, ( K_2 + K_1RK_1) \Upsilon_0 \rangle $ as a consequence of  Lemma \ref{lem:Bogoliubov:identity}. Using $\langle  \Upsilon_0  , K \Phi_b \rangle = \langle  \Upsilon_0  , K \widetilde \Phi_b \rangle + O(\alpha^{-b-3})$ with $\widetilde \Phi_b$ given by \eqref{eq:deg:widetilde:Phi}, we can compare this with
\begin{align}
 \Ls  \Upsilon_0  , K  \widetilde \Phi_b \Rs 
 & =  \Ls  \Upsilon_0 ,\big(  \tfrac{1}{\alpha^2}(  \mb H_0 -  \eB E^\n + \mb K )  + \mb K_{\textnormal{rem}} \big)  \bigg( \sum_{j=0}^{b+2} (\mb R \mb K )^j \bigg) \Upsilon_0 \Rs    \notag\\
 & =  \Ls \Upsilon_0  , \big(  \tfrac{1}{\alpha^2} \mb K +  \mb K_{\textnormal{rem}} \big)   \bigg( \sum_{j=0}^{b+2} (\mb R \mb K )^j \bigg)  \Upsilon_0 \Rs \notag\\
& =   \sum_{\ell =1}^{b} \frac{1}{\alpha^{\ell+2} } \Big\langle \Upsilon_0  , \sum_{j=1}^{\ell } \sum_{\substack{ \peps \in \mb N^j \\ |\peps|= \ell } } \mb K_{ \varepsilon_1} (\mb R \mb K_{ \varepsilon_2} ) \ldots (\mb R \mb K_{ \varepsilon_j}) \Upsilon_0  \Big\rangle + O(\alpha^{-b-3}),
\end{align}
where we used \eqref{eq:PVRV:Bod:identity}, $(\mb H_0 - \eB E^\n) \mb P =0$ and $\mb K_{\textnormal {rem}} = O(\alpha^{-b-3})$, and where $ O(\alpha^{-b-3})$ stands for all terms of order $\alpha^{-b-3}$ and higher. Since the coefficients of $\alpha^{-\ell-2}$ for $\ell\in \{1,\ldots , b\}$ must coincide with the ones from \eqref{eq:coefficients:comparison}, we can use \eqref{eq:def:E:Lambda} together with $\widetilde {\mb K}_\ell = \mb K_\ell + E_\ell $ to obtain
\begin{align}
E_\ell  &= \Big\langle \Upsilon_0   , 
\sum_{j=1}^{\ell } \sum_{\substack{ \peps \in \mb N^j \\ |\peps|= \ell } } \mb K_{ \varepsilon_1} (\mb R \mb K_{ \varepsilon_2} ) \ldots (\mb R \mb K_{ \varepsilon_j})  \Upsilon_0  \Big\rangle  + E_\ell  =  \sum_{k=1}^{\ell+2} \Ls \Upsilon_0, K_k \Upsilon_{\ell+2-k} \Rs +  E_\ell .
\end{align}
This completes the proof of the lemma since by definition $ K_{\ell+2} + E_\ell = \widetilde K_{\ell+2} $.
\end{proof}

The proof of the next lemma is somewhat more technical and thus postponed to Section \ref{sec:proof:combinatorial:estimates}.

\begin{lem}\label{lem:estimates:approx:groundstate:2} For $(E_\ell)_{\ell \in \mb N_0}$ defined as in \eqref{eq:def:E:Lambda}, there exists a constant $C>0$ auch that $|E_\ell | \le C^\ell \sqrt{(\ell +2)!}$ for all $\ell \ge 1$ and $\Lambda \ge 2 $. Moreover, there exist constants $\alpha_0, C>0$ such that $\| K _k  \Upsilon_{\ell+2} \| \le (1+  \sqrt \Lambda\, )C ^{\ell+k} \sqrt{(\ell+k+2)!}$ for all $\alpha\ge \alpha_0$ and $\Lambda \ge 2$.

% Let $K_\ell$ for $\Lambda\ge 0$ be defined as in \eqref{eq:def:K_1}--\eqref{eq:def:K:ell+1} for the sequence $(E_\ell)_{\ell \in \mb N_0}$ specified by \eqref{eq:def:E:Lambda} and consider $\Upsilon_\ell$ as in \eqref{eq:def:Upsilon:ell}. Then there exist constants $\alpha_0,C>0$ such that 
%\begin{align} 
%\sum_{k=1}^{\ell+1}\big\| P  K _k  \Upsilon_{\ell-k+2} \big\| \le C^{\ell} \sqrt{(\ell+2)!} 
%\end{align}
%and $|E_\ell | \le C^\ell \sqrt{\ell!}$ for all $\ell \ge 1 $ and $\Lambda \ge 2$ and $\big\| K _k  \Upsilon_{\ell+2} \big\| \le (1+  \sqrt \Lambda\, )C ^{\ell+k} \sqrt{(\ell+k+2)!}$.
\end{lem}

The statement of Proposition \ref{prop:asym:expansion:HG} is now a consequence of Lemmas \ref{lem:identity}--\ref{lem:estimates:approx:groundstate:2}.

\begin{proof}[Proof of Proposition \ref{prop:asym:expansion:HG}] Since $\mf H_{\alpha,\Lambda}^{\textnormal{G}} - \sum_{\ell=0}^{b}\alpha^{-\ell-2} E_\ell = H_0 +K $, see \eqref{eq:H:G:formula:sum:K}, we can apply Lemma \ref{lem:identity} to obtain
\begin{align}
\bigg( \mf H_{\alpha,\Lambda}^{\textnormal{G}} - \sum_{\ell=0}^{b}\alpha^{-\ell-2} E_\ell \bigg) \Phi_b =  ( P\otimes \mathbb P)  K \Phi_b + ( Q +  P \otimes \mb Q )\sum_{\substack{ \ell , k =1 \\ k+l \ge b+3} }^{b+2} \frac{1}{\alpha^{k+l}}  K_\ell  \Upsilon_k.
\end{align}
The bound for the second term follows from the second statement of Lemma \ref{lem:estimates:approx:groundstate:2}, that is
\begin{align}\label{eq:remainder:bounds}
\bigg\|\sum_{\substack{ \ell , k =1 \\ k+l \ge b+3} }^{b+2}   K_\ell  \Upsilon_k \bigg\| & \le  (b+2)^2  \max_{ 1\le k,\ell \le b+2 } \| K_k \Upsilon_\ell \| \notag\\
& \le  (1+\sqrt \Lambda) C^{2b+4} \sqrt{ (2b+4)!} \le (1+\sqrt \Lambda )  C^{b+1} (b +1)!
\end{align}
where the last step can be verified with the aid of Stirling's formula.

To show that $\|  ( P \otimes \mb P  ) K  \Phi_b \| \le C^{b+1} (b+1)! \alpha^{-b-3}$, we use \eqref{eq:coefficients:comparison} and bound the second term in \eqref{eq:coefficients:comparison} similarly to the one above by means of Lemma \ref{lem:estimates:approx:groundstate:2},
\begin{align}
\bigg| \sum_{\substack{k=1,\ell=0 \\ k+\ell \ge b+3}}^{b+2}    \frac{1}{\alpha^{k+\ell}}  \Big\langle \Upsilon_0 , K_k \Upsilon_\ell \Rs  \bigg|  \le \frac{C^{b+1} (b +1)!}{\alpha^{b+3}} .
\end{align}
The first term in \eqref{eq:coefficients:comparison} is zero, since for $\ell \in \{1,\ldots,b\}$ and with $K_{\ell+2} = \widetilde K_{\ell+2} -  E_\ell $,
\begin{align}
 \sum_{k=1}^{\ell+2} \Ls \Upsilon_0, K_k \Upsilon_{\ell+2-k} \Rs = -E_\ell +  \Ls \Upsilon_0, \widetilde K_\ell \Upsilon_0 \Rs + \sum_{k=1}^{\ell+1} \Ls \Upsilon_0, K_k \Upsilon_{\ell+2-k} \Rs =0
\end{align}
by Lemma \ref{lem:energy:identity}. This completes the proof of the lemma.

\end{proof}
\subsection{Proof of Lemma \ref{lem:estimates:approx:groundstate:2}}

\label{sec:proof:combinatorial:estimates}

To prove Lemma \ref{lem:estimates:approx:groundstate:2}, we derive the following implication: If there are constants $C_2>0$, $\ell \in \mb N$ such that $|E_k| \le C_2^k \sqrt{(k+2)!}$ for all $k\in \{ 1,\ldots, \ell-1\}$, then
\begin{align}
\label{eq:PKUpsilon:bound}
\sum_{k=1}^{\ell+1} \| P \mb P K_k \Upsilon_{l+2-k} \| \le C_2^\ell \sqrt{ (\ell +2)!}.
\end{align}
Note that for notational convenience, we omit the tensor product in operators of the form $P\otimes \mb P$ throughout this section.  With Lemma \ref{lem:energy:identity}, we can use this implication to arrive at $|E_\ell| \le C_2^\ell \sqrt{(\ell +2 )!} $ for all $\ell \ge 1$ by iteration. Having established the bounds on $E_\ell$, the proof of the second statement of the lemma follows in close analogy to the derivation of \eqref{eq:PKUpsilon:bound}. 

The proof of \eqref{eq:PKUpsilon:bound} requires some preparation. We start by introducing the auxiliary operators
\begin{align}\label{eq:def:V:ell:beta}
 V_{\ell,\beta}:=\begin{cases}
-E_{\ell-2}&\text{ for }\beta=\ell\geq 1,\\
K_\ell+E_{\ell-2}&\text{ for }\beta=0\text{ and } \ell\geq 1,\\
0&\text{ otherwise, } 
 \end{cases}
\end{align}
where we set $E_{-1}=0$. With these at hand, we define
\begin{align}
\mb T_{\ell,\beta} := \sum_{p=1}^{\ell} \sum_{\substack { \ell_1,\ldots, \ell_p \ge 1  \\ \ell_1 + \ldots + \ell_p = \ell }}  \sum_{\substack { \beta_1,\dots ,\beta_p \\  0\leq \beta_i\leq \ell_i \\ \beta_1+\ldots + \beta_p = \beta  } }  V_{\ell_1,\beta_1}\! \left(R V_{\ell_2,\beta_2}\right)\dots \left(R V_{\ell_p,\beta_p}\right)\! P.
\end{align}
Inserting $K_\ell =V_{\ell,0}+  V_{\ell,\ell} $ in \eqref{def:mb:V:k}, it is straightforward to see that
\begin{align}
\mathbb{K}_\ell= P\sum_{\beta=0}^{\ell+2} \mathbb{T}_{\ell+2,\beta}\quad \text{and}\quad \widetilde{\mathbb{K}}_\ell=P\sum_{\beta=0}^{\ell+1} \mathbb{T}_{\ell+2,\beta},
\end{align}
where the second identity follows from the first by $\mb T_{\ell+2,\ell+2} = V_{\ell+2,\ell+2}P = - E_{\ell}P$ and $\widetilde {\mathbb K}_{\ell} = \mb K_\ell  - E_\ell $. While this way of writing $\mb K_\ell$ and $\widetilde{\mb K}_\ell$ involves many summands that are zero by definition, it turns out useful to keep track of the combinatorial factors. 

Furthermore, let us define the spaces $\mathcal{H}_k:=\mathbb U^* \! \left(\mathcal{F}_{\leq k}\right)$, where $\mathbb U$ is the unitary defined in Eq.~(\ref{eq: def of U}) and $\mathcal{F}_{\leq k}\subseteq \mathcal F$ is the subspace of states involving at most $k$ particles. Before we prove Lemma \ref{lem:estimates:approx:groundstate:2}, we need some auxiliary results.

\begin{lem}\label{lem:fock_space_results}
Assume that $|E_k| \le C_2^k \sqrt{(k+2)!}  $ for some $C_2>0$ and all $k\in \{1,\ldots, \ell-1 \}$. Then there exists a constant $D>0$, such that
\begin{align}
\label{eq:basic_fock_1}
&\|(Y V_{\ell,0} Z)|_{\mathcal{H}_r}\|\leq D \sqrt{(r+1)\dots (r+\ell)},\\
\label{eq:basic_fock_2}
&\|\left(V_{\ell,0} Z\right)|_{\mathcal{H}_{r}}\|\leq \left(1+\sqrt{\Lambda}\right)D\sqrt{(r+1)\dots (r+\ell)},\\
\label{eq:basic_fock_3}
&\|(\mathbb{R} P V_{\ell_1,0}R V_{\ell_2,0}R^{\frac{1}{2}})|_{\mathcal{H}_r}\|\leq D\sqrt{(r+1)\dots (r+\ell_1+\ell_2-2)}
\end{align}
with $Y,Z\in \{(-R)^{\frac{1}{2}},P\}$. 
\end{lem}

\begin{proof}
The bounds \eqref{eq:basic_fock_1} and \eqref{eq:basic_fock_2} follow immediately from Lemma \ref{lem:estimates:approx:groundstate:3x} and Lemma \ref{lem:number:op:bounds}. In order to verify Eq.~(\ref{eq:basic_fock_3}), note that $\mathbb{R}^2\leq C \mathbb{U}^\dagger (\mathcal{N}+1)^{-2} \mathbb{U}$, and therefore
\begin{align}
\|(\mathbb{R} P V_{\ell_1,0}R V_{\ell_2,0}R^{\frac{1}{2}})|_{\mathcal{H}_r}\|\leq \sqrt{C}\|(\mathbb{U}^\dagger (\mathcal{N}+1)^{-1} \mathbb{U} P V_{\ell_1,0}R V_{\ell_2,0}R^{\frac{1}{2}})|_{\mathcal{H}_r}\|.
\end{align}
To process this expression further, let us decompose $V_{\ell,0}|_{\mathcal{H}_r}=\sum_{i=-2}^2 W_{\ell}^i$, where $W_{\ell}^i$ is defined by $W_{\ell}^{-i}= \mathds 1(\ell \le 4)\, \pi_{\mathcal{H}_{r+i}} V_{\ell,0}|_{\mathcal{H}_r}$ and $\pi_{\mathcal{H}_{r}}$ is the orthogonal projection onto $\mathcal{H}_{r}$. Making use of Lemma \ref{lem:estimates:approx:groundstate:3x}, we know that $\|(Y W_\ell^i Z)|_{\mathcal{H}_r}\|\leq D\sqrt{(r+1)\dots (r+\ell)}$ and consequently 
\begin{align}
\label{eq:composed_W_estimate}
\|(P W_{\ell_1}^{i_1}R W_{\ell_2}^{i_2}R^{\frac{1}{2}})|_{\mathcal{H}_r}\|\leq D^2 \sqrt{(r+1)\dots (r+\ell_1+\ell_2)}.
\end{align}
Further note that $Z_{\ell_1,\ell_2}^{i_1,i_2}:=P W_{\ell_1}^{i_1}R W_{\ell_2}^{i_2}R^{\frac{1}{2}}$ satisfies $Z_{\ell_1,\ell_2}^{i_1,i_2} \mathcal{H}_{=r} \subset \mathcal{H}_{=r+i_1+i_2}$, where $\mathcal{H}_{=r}:=\mathcal{H}_{r}\cap \mathcal{H}_{r-1}^\perp$, and consequently $Z_{\ell_1,\ell_2}^{i_1,i_2}\mathcal{H}_{=r}\perp Z_{\ell_1,\ell_2}^{i_1,i_2} \mathcal{H}_{=s} $ for $r\neq s$. Decomposing an arbitrary state $\Psi\in \mathcal{H}_r$ into $\Psi=\sum_{s=0}^r \Psi_s$ such that $\Psi_s\in \mathcal{H}_{=s}$, we obtain
\begin{align}
&\|\mathbb{U}^\dagger (\mathcal{N}+1)^{-1} \mathbb{U}P W_{\ell_1}^{i_1}R W_{\ell_2}^{i_2}R^{\frac{1}{2}}\Psi\|^2=\sum_{s=s_0}^r \|\mathbb{U}^\dagger (\mathcal{N}+1)^{-1} \mathbb{U}P W_{\ell_1}^{i_1}R W_{\ell_2}^{i_2}R^{\frac{1}{2}}\Psi_s\|^2 \notag \\
& \ \ =\sum_{s=s_0}^r \frac{1}{(s+i_1+i_2+1)^2}\|P W_{\ell_1}^{i_1}R W_{\ell_2}^{i_2}R^{\frac{1}{2}}\Psi_s\|^2\leq \sum_{s=s_0}^r D^2\frac{(s+1)\dots (s+\ell_1+\ell_2)}{(s+i_1+i_2+1)^2}\|\Psi_s\|^2
\end{align}
with $s_0:=\max\{0,-i_1-i_2\}$, where we have used Eq.~(\ref{eq:composed_W_estimate}). Since we have the inequality $\frac{(s+1)\dots (s+\ell_1+\ell_2)}{(s+i_1+i_2+1)^2}\leq 7(s+1)\dots (s+\ell_1+\ell_2-2) \leq 7 (r+1)\dots (r+\ell_1+\ell_2-2)$ we obtain $\|(\mathbb{U}^\dagger (\mathcal{N}+1)^{-1} \mathbb{U}P W_{\ell_1}^{i_1}R W_{\ell_2}^{i_2}R^{\frac{1}{2}})|_{\mathcal{H}_r}\|\leq \sqrt{7}D\sqrt{(r+1)\dots (r+\ell_1+\ell_2-2)}$, which concludes the proof. Here we used that $W_{\ell}^{-i} =0$ if $\ell \ge 5$.
\end{proof}

\begin{lem}\label{lem:estimates:approx:groundstate:2_Auxiliary} 
Assume that $|E_k| \le C_2^k \sqrt{(k+2)!}  $ for some $C_2>0$ and all $k\in \{1,\ldots, \ell-1 \}$. Then there exist constants $C_0,C_1>0$, such that 
\begin{align*}
\left\|(X \mb T_{\ell,\beta})|_{\mathcal{H}_s}\right\|\leq C_0 \, C_1^{\ell-\beta}C_2^{\max\{\beta-2,0\}} \sqrt{(s+1)\dots (s+\ell-\beta)}\sqrt{\beta!}
\end{align*}
for all $\beta\leq \ell$ and $X\in \{(-R)^{\frac{1}{2}},P\}$, where we use the convention that $(s+1)\dots (s+\ell-\beta)=1$ for $\beta=\ell$. Furthermore, $\left\|(\mathbb{R} P \mb T_{\ell+2,0})|_{\mathcal{H}_s}\right\|\leq C_1^{\ell} \sqrt{(s+1)\dots (s+\ell)}$.
\end{lem}
\begin{proof}[Proof of Lemma \ref{lem:estimates:approx:groundstate:2_Auxiliary}] Using the fact that $V_{\ell,\beta} \mathcal{H}_r \subset \mathcal{H}_{r+\ell-\beta}$ and that there exists a constant $D>0$ such that $\|(Y V_{\ell,\beta} Z)|_{\mathcal{H}_r}\|\leq D^{\ell-\beta} C_2^{\max\{\beta-2,0\}}  \sqrt{(r+1)\dots (r+\ell-\beta)}\sqrt{\beta!}$ for $Y,Z\in \{(-R)^{\frac{1}{2}},P\}$, see Lemma \ref{lem:fock_space_results}, we can bound $\left\|(X \mb T_{\ell,\beta})|_{\mathcal{H}_s}\right\|$ from above by
\begin{align*}
 & D^{\ell-\beta}C_2^{\max\{\beta-2,0\}} \sqrt{(s \! + \! 1)\dots (s \! + \! \ell \! - \! \beta)}\sum_{p=1}^\beta \sum_{\substack{\ell_1,\dots , \ell_p\in \mb N \\ \ell_1 + \ldots + \ell_p = \ell}} \sum_{\substack { \beta_1,\dots ,\beta_p: 0\leq \beta_i\leq \ell_i \\ \beta_1+ \ldots + \beta_p = \beta } } \sqrt{\beta_1!\dots \beta_p!}\\
&\  =D^{\ell-\beta} C_2^{\max\{\beta-2,0\}} \sqrt{(s \! + \! 1)\dots (s \! + \! \ell \! - \! \beta)}\sum_{p=1}^\beta \sum_{\substack{ \beta_1,\dots ,\beta_p\geq 0 \\ \beta_1 + \ldots + \beta_p = \beta }} \sqrt{\beta_1!\dots \beta_p!}
\left( \sum_{\substack{ \ell_1, \dots , \ell_p  \\
\ell_i \ge \max\{\beta_i,1\} \\ \ell_1 + \ldots  +\ell_p = \ell }} 1\right).
\end{align*}
For given $\beta_1,\dots , \beta_p$ satisfying $\beta_1+\dots  +\beta_p=\beta$, let us define the set $I=I(\beta_1,\dots ,\beta_p):=\{j:\beta_j\geq 1\}$, the number $p'=p'(\beta_1,\dots ,\beta_p):=|I|$ and the indicator function $\chi=\chi(\beta_1,\dots ,\beta_p):=\mathds{1}_{p'\geq p-\ell+\beta}$. Note that the existence of a sequence $\ell_1,\dots ,\ell_p$ satisfying $\ell_1+\dots +\ell_p=\ell$ and $\ell_j\geq \max\{\beta_i,1\}$ immediately implies 
\begin{align*}
p-p'=|\{1,\dots,p\}\setminus I|\leq \sum_{i\in \{1,\dots,p\}\setminus I}\ell_i=\sum_{i\in \{1,\dots,p\}\setminus I}(\ell_i-\beta_i)\leq \sum_{i=1}^p (\ell_i-\beta_i)=\ell-\beta,
\end{align*}
and therefore $\chi(\beta_1,\dots,\beta_p)=1$. As a consequence
\begin{align}
\nonumber
\sum_{\substack{ \ell_1,\dots , \ell_p: \ell_i\geq \max\{\beta_i,1\} \\ \ell_1 + \ldots + \ell_p = \ell }}  1 & \leq \sum_{\substack{ \ell_1,\dots , \ell_p: \ell_i\geq \beta_i \\ \ell_1 + \ldots + \ell_p = \ell }} \chi = \sum_{\substack{ \delta_1,\dots , \delta_p: \delta_i\geq 0 \\ \delta_1 + \ldots + \delta_p =  \ell - \beta} } \chi=\sum_{q=1}^p \binom{p}{q} \sum_{\substack{ \delta_1,\dots , \delta_{p} \ge 1  \\ \delta_1 + \ldots + \delta_p = \ell - \beta }} \chi\\
\label{eq: Estimate_Combinatoric_Sum}
&\leq \sum_{q=1}^p \binom{p}{q} 2^{\ell-\beta}\chi\leq 2^p 2^{\ell-\beta}\chi.
\end{align}
With this at hand, we can bound
\begin{align*}
& \sum_{p=1}^\beta \sum_{\substack{ \beta_1,\dots ,\beta_p\geq 0 \\ \beta_1 + \ldots +\beta_p = \beta }} \sqrt{\beta_1!\dots \beta_p!} \left( \sum_{\substack{ \ell_1,\dots , \ell_p: \ell_i \ge \max\{\beta_i,1\} \\ \ell_1 + \ldots + \ell_p = \ell } } 1\right) \notag\\
&  \le 2^{\ell-\beta} \! \sum_{p=1}^\beta 2^p \sum_{\substack{ \beta_1,\dots ,\beta_p\geq 0 \\ \beta_1 + \ldots +  \beta_p = \beta }}  \sqrt{\beta_1!\dots \beta_p!}\, \chi =2^{\ell-\beta}\sum_{p=1}^\beta \sum_{p'=p-\ell+\beta}^p \!  \!  \!  \!  2^p\binom{p}{p'} \sum_{\substack{ \beta_1,\dots ,\beta_{p'}\geq 1 \\ \beta_1 + \ldots + \beta_{p'} =\beta }} \sqrt{\beta_1!\dots \beta_{p'}!}\, .
\end{align*}
In order to estimate this further, we are going to verify 
\begin{align}
\sum_{\substack{ \beta_1,\dots ,\beta_{q}\geq 1 \\ \beta_1+ \ldots + \beta_q = \beta } } \sqrt{\beta_1!\dots \beta_{q}!}\leq \frac{12^q}{\sqrt{q!}}\sqrt{\beta!}
\end{align}
by induction. For $q=1$ the statement is trivial. Regarding the step $q\mapsto q+1$, note that
\allowdisplaybreaks
\begin{align}
&\sum_{\substack{ \beta_1,\dots ,\beta_{q+1}\geq 1 \\ \beta_1 + \ldots + \beta_{q+1} =\beta } } \sqrt{\beta_1!\dots \beta_{q+1}!}=\sum_{t=1}^{\beta-q}\sqrt{t!}\sum_{\substack{ \beta_1,\dots ,\beta_{q}\geq 1 \\ \beta_1 + \ldots + \beta_q = \beta -t }} \sqrt{\beta_1!\dots \beta_{q}!}\leq \frac{12^q}{\sqrt{q!}}\sum_{t=1}^{\beta-1}\sqrt{t!}\sqrt{(\beta-t)!}\notag \\
& \ \ \ \ =\frac{12^q}{\sqrt{q!}}\sqrt{\beta!}\sum_{t=1}^{\beta-1}\binom{\beta}{t}^{-\frac{1}{2}}=\frac{12^q}{\sqrt{q!}}\sqrt{\beta!}\left(\frac{2}{\sqrt{\beta}}+\frac{4}{\sqrt{\beta(\beta-1)}}+\sum_{t=3}^{\beta-3}\binom{\beta}{t}^{-\frac{1}{2}}\right) \notag \\
& \ \ \ \ \leq \frac{12^q}{\sqrt{q!}}\sqrt{\beta!}\left(\frac{2}{\sqrt{\beta}}+\frac{4}{\sqrt{\beta(\beta-1)}}+(\beta-5)\binom{\beta}{3}^{-\frac{1}{2}}\right)\leq \frac{12^{q+1}}{\sqrt{q!}\sqrt{\beta}}\sqrt{\beta!}, 
\end{align}
where we have used $\binom{\beta}{t}^{-\frac{1}{2}}\leq \binom{\beta}{3}^{-\frac{1}{2}}$ for $3\leq t\leq \beta-3$. Since $\sum_{\substack{ \beta_1,\dots ,\beta_{q+1}\geq 1 \\ \beta_1+ \ldots + \beta_{q+1} = \beta }} \sqrt{\beta_1!\dots \beta_{q+1}!}=0$ for $\beta<q+1$, we can assume $\beta\geq q+1$, which concludes the induction. Combining what we have so far allows us to bound $\left\|(X \mb T_{\ell,\beta})|_{\mathcal{H}_s}\right\|$ from above by
\begin{align*}
(2D)^{\ell-\beta} C_2^{\max\{\beta-2,0\}} \sqrt{(s \! + \! 1)\dots (s \! + \! \ell \! - \! \beta)}\sqrt{\beta!}\sum_{p=1}^\beta \sum_{p'=\max\{p-\ell+\beta,0\}}^p \!  \!  \!  \!  2^p\binom{p}{p'} \frac{12^{p'}}{\sqrt{p'!}}.
\end{align*}
Making use of the estimate
\begin{align*}
\sum_{p=1}^\beta \sum_{p'=\max\{p-\ell+\beta,0\}}^p \!  \!  \!  \!  2^p\binom{p}{p'} \frac{12^{p'}}{\sqrt{p'!}}&=\sum_{p'=0}^\beta \frac{12^{p'}}{\sqrt{p'!}}\sum_{p=\max\{1,p'\}}^{\min\{p'+\ell-\beta,\beta\}} \!  \!  \!  \!  2^p\binom{p}{p'}\leq \sum_{p'=0}^\beta \frac{12^{p'}}{\sqrt{p'!}}\sum_{p=1}^{p'+\ell-\beta} \!  \!  \!  \!  4^p\\
& \leq \sum_{p'=0}^\beta \frac{12^{p'}}{\sqrt{p'!}}\frac{4}{3}4^{p'+\ell-\beta}\leq C_0 4^{\ell-\beta}
\end{align*}
with $C_0:=\frac{4}{3}\sum_{p'=0}^\beta \frac{(48)^{p'}}{\sqrt{p'!}}<\infty$ concludes the proof of the first statement of the lemma with $C_1:=8D$.

Regarding the proof of the second statement of the lemma, let us write 
\begin{align}
\mathbb{R} P \mb T_{\ell+2,0}=\sum_{\ell_1,\ell_2=1}^4 \mathbb{R} P V_{\ell_1,0}R V_{\ell_2,0}R^{\frac{1}{2}}R^{\frac{1}{2}} \mb T_{\ell+2-\ell_1-\ell_2,0}.
\end{align}
Using $\|(\mathbb{R} P V_{\ell_1,0}R V_{\ell_2,0}R^{\frac{1}{2}})|_{\mathcal{H}_r}\|\leq D\sqrt{(r+1)\dots (r+\ell_1+\ell_2-2)}$ for a suitable constant $D>0$, see Lemma \ref{lem:fock_space_results}, yields
\begin{align}
\left\|(\mathbb{R} P \mathbb T_{\ell+2,0})|_{\mathcal{H}_s}\right\|\leq D C_0 C_1^{\ell+2-\ell_1-\ell_2}\sqrt{(s+1)\dots (s+\ell)}.
\end{align}
\end{proof}

We are now prepared to prove Lemma \ref{lem:estimates:approx:groundstate:2}.

\begin{proof}[Proof of Lemma \ref{lem:estimates:approx:groundstate:2}] As explained at the beginning of this section, the main step is to derive \eqref{eq:PKUpsilon:bound}. To this end, let us write $P \mb P K_k\Upsilon_{\ell+2} =P \mb P   Z_1+P  \otimes \mb P Z_2$ with $Z_1$ defined as
\begin{align}
\label{eq: representation eigenstate_1}
\sum_{p=1}^\ell \!  \sum_{\beta=0}^{\ell+2p}\! \sum_{\substack{ \ell_1,\dots , \ell_p\geq 1 \\ \ell_1 + \ldots + \ell_p = \ell }}  \sum_{\substack{ \beta_1,\dots ,\beta_p \\ 0\leq \beta_i\leq \ell_i+2 \\ \beta_1+ \ldots + \beta_p = \beta }} \! \! \!  V_{k,0} ( R \mb T_{\ell_1+2,\beta_2} ) ( \mathbb{R} P \mb T_{\ell_2+2,\beta_1} )  \dots \! ( \mathbb{R} \! P \mb T_{\ell_p+2,\beta_p} )  \Upsilon_0 \! \! 
\end{align}
and $Z_2$ defined as
\begin{align}
\label{eq: representation eigenstate_2}
\sum_{p=1}^{ \ell +2 }\!   \sum_{\beta=0}^{\ell+2+2p}\! \sum_{\substack{ \ell_1,\dots , \ell_p\geq 1 \\ \ell_1 + \ldots + \ell_p = \ell +2 }} \sum_{\substack{ \beta_1,\dots ,\beta_p \\ 0\leq \beta_i\leq \ell_i+2 \\ \beta_1 + \ldots + \beta_p = \beta }} \! \!   V_{k,0} ( \mathbb{R}  P \mb T_{\ell_1+2,\beta_1} ) ( \mathbb{R} P \mb T_{\ell_2+2,\beta_1} )   \dots ( \mathbb{R}  P \mb T_{\ell_p+2,\beta_p} ) \Upsilon_0,
\end{align}
where we have used that $P   K_kR=P   V_{k,0}R$ and $ \mb P   K_k\mathbb{R} = \mb P  V_{k,0}\mathbb{R}  $, see \eqref{eq:def:V:ell:beta}.  According to Lemma \ref{lem:estimates:approx:groundstate:2_Auxiliary} we have $\| ( R^{\frac{1}{2}}\mb T_{\ell_1+2,0} ) |_{\mathcal{H}_{3s}}\|\leq (\sqrt{3}C_1)^{\ell_1+2} \sqrt{(s+1)\dots (s+\ell_1+2)}$ as well as for $\beta_i\geq 3$
\begin{align*}
\|(X \mb T_{\ell_i+2,\beta_i})|_{\mathcal{H}_{3s}}\|\leq \frac{C_0}{\tau} (\sqrt{3}C_1)^{\ell_i+2-\beta_i}C_2^{\beta_i-2}\sqrt{(s+1)\dots (s+\ell_i+2-\beta_i)}
\end{align*}
where $X\in \{\mathbb{R} \! P,(-R)^{\frac{1}{2}}\}$ and $\tau:=\min\{\|\mathbb{R} \|^{-1},1\}$. For $\beta_i=0$ we use the estimate
\begin{align*}
\left\|(\mathbb{R} P \mb T_{\ell+2,0})|_{\mathcal{H}_{3s}}\right\|\leq (\sqrt{3}C_1)^{\ell} \sqrt{(s+1)\dots (s+\ell)}
\end{align*}
by Lemma \ref{lem:estimates:approx:groundstate:2_Auxiliary}. Further, $\|\left(P V_{k,0} Y\right)|_{\mathcal{H}_{3s}}\|\leq C_1^{k}\sqrt{(s+1)\dots (s+k)}$ for $Y\in \{(-R)^{\frac{1}{2}},P\}$, see Lemma \ref{lem:fock_space_results}. Note at this point that we only need to consider values of $q$ bounded by $p-1$, since $X \mb T_{\ell_p+2,\beta_p}\Upsilon_0 = 0$ for $X\in \{\mathbb{R} P,(-R)^{\frac{1}{2}}\}$ allowing us to restrict the sum in Eq.~(\ref{eq: representation eigenstate_1}) and Eq.~(\ref{eq: representation eigenstate_2}) to $\beta_{p}=0$. Furthermore, the index $\beta$ is in both Equations bounded by $\beta_*:=\ell+k-(p-q)+2q$, since $\beta=\sum_{j=1}^p \beta_j\leq \sum_{j:\beta_j\neq 0} (\ell_j+2)=\sum_{j:\beta_j\neq 0}\ell_j+2q$ and in the case of Eq.~(\ref{eq: representation eigenstate_1}) we have $\sum_{j:\beta_j\neq 0}\ell_j\leq \ell-(p-q)\leq \ell+k-(p-q)$, while we have $\sum_{j:\beta_j\neq 0}\ell_j\leq \ell+2-(p-q)\leq \ell+k-(p-q)$ in the case of Eq.~(\ref{eq: representation eigenstate_2}). Note that we have used the fact that $PV_{k,0}P=0$ for $k=1$ and therefore we can assume $k\geq 2$ in Eq.~(\ref{eq: representation eigenstate_2}). Making use of $\mb T_{\ell+2,0} \mathcal{H}_{3r}\subset \mathcal{H}_{3(r+\ell)}$ , $\mb T_{\ell+2,\beta} \mathcal{H}_{3r}\subset \mathcal{H}_{3(r+\ell+2-\beta)}$ and $\Upsilon_0\in \mathcal{H}_s$, we obtain for a suitable constant $\widetilde{C}_0$ and both $Z_1$ and $Z_2$ the estimate
\begin{align}
& \|Z_i\| \le  \sum_{p=1}^{\ell+2} \sum_{q=0}^{p-1} \sum_{\beta=0}^{\beta_*}\binom{p}{q}  \tilde{C}_0^q (\sqrt{3}C_1)^{\ell+k+2p-\beta}C_2^{\beta-2q}  \notag \\
\nonumber
& \quad \times \sqrt{(s \! + \! 1)\dots (s \! + \! \ell \! + \! k \! + \! 2(q \! + \! 1) \! - \! \beta)}\, \sum_{\substack{ \beta_1,\dots ,\beta_q \ge  3 \\ \beta_1 +\ldots + \beta_q = \beta  }} \sqrt{\beta_1!\dots \beta_q!} \sum_{\substack{ \ell_1,\dots , \ell_p\geq \max\{1,\beta_i-2\} \\ \ell_1 + \ldots + \ell_p = \ell+2 }} 1\\
\nonumber
\qquad \quad  &  \le C \!  \sum_{p=1}^{\ell+2} \sum_{q=0}^{p-1} \sum_{\beta=0}^{\beta_*}\binom{p}{q}  \tilde{C}_0^q (\sqrt{3(s+1)}C_1)^{\ell+k+2p-\beta}C_2^{\beta-2q}\\
 \nonumber
 &  \ \ \times \sqrt{(\ell+k+2(q+1)-\beta)!}\,  2^p 2^{\ell+2p-\beta}  \sum_{\substack{ \beta_1,\dots ,\beta_q \ge 3 \\ \beta_1 + \ldots \beta_q = \beta }} \, \sqrt{\beta_1!\dots \beta_q!},
\end{align}
see also Eq.~(\ref{eq: Estimate_Combinatoric_Sum}) regarding the last estimate, where we have further used $(s+1)\dots (s+\ell+k+2(q+1)-\beta)\leq (s+1)^{\ell+k+2(q+1)-\beta}(\ell+k+2(q+1)-\beta)!$. We next show by induction with respect to $q$ that the remaining sum over $\beta_1,\ldots ,\beta_q$ is bounded by $\frac{C_*^q}{\sqrt{q!}} \sqrt{(\beta+2-2q)!}$ for a suitable constant $C_*$. The case $q=1$ is clear. Regarding the induction step $q\mapsto q+1$ we estimate, similarly as in the proof of Lemma \ref{lem:estimates:approx:groundstate:2_Auxiliary},
\begin{align*}
&  \sum_{\substack{ \beta_1,\ldots ,\beta_{q+1} \ge 3 \\ \beta_1 + \ldots + \beta_{q+1} = \beta }} \sqrt{\beta_1!\dots \beta_{q+1}!}\leq \frac{C_*^q}{\sqrt{q!}}\sum_{t=3q}^{\beta-3} \sqrt{(t+2-2q)!}\sqrt{(\beta-t)!}\\
& \ \ \leq \frac{C_*^q}{\sqrt{q!}}\sqrt{(\beta \! - \! 2q \! - \! 1)!}\left(2\sqrt{6}+2\sqrt{\frac{24}{\beta \! - \! 2q \! - \! 1}}+(\beta \! - \! 3q \! - \! 2)\sqrt{\frac{120}{(\beta \! - \! 2q \! - \! 2)(\beta \! - \! 2q \! - \! 1)}}\right)\\
& \ \ \leq \frac{C_*^{q+1}}{\sqrt{q!}}\sqrt{(\beta  -  2q  -  1)!}\leq \frac{C_*^{q+1}}{\sqrt{(q+1)!}}\sqrt{(\beta  -  2q )!},
\end{align*}
where we have used that $t+2-2q$ and $\beta-t$ are bounded from above by $\beta-2q-1$. Since $\beta+2-2q$ as well as $\ell+k+2(q+1)-\beta$ are bounded from above by $\ell+k+2$ we further obtain $\sqrt{(\ell+k+2(q+1)-\beta)!}\sqrt{(\beta+2-2q)!}\leq \binom{\ell+k+4}{\ell+k+2}^{-\frac{1}{2}}\sqrt{(\ell+k+4)!}=\sqrt{2}\sqrt{(\ell+k+2)!}$, and therefore we can bound 
\begin{align*}
 &  \|P  \mb P K_k\Upsilon_{\ell+2} \|\! \leq \!  D \! \sum_{p=1}^{\ell+2} \sum_{q=0}^{p-1} \sum_{\beta=0}^{\beta_*} \!2^p\binom{p}{q} \! \!    \left( \! 2\sqrt{3(s \! + \! 1)}C_1 \! \right)^{\ell+k+  2p  -  \beta} \! \!  \!  \! C_2^{\beta \! - \! 2q}  \! \frac{(\tilde{C}_0C_*)^q}{\sqrt{q!}}\sqrt{(\ell \! + \! k \! + \! 2  )!}\\
& \ \ \ \ \leq D\sqrt{(\ell+k+2)!} \sum_{q=0}^{\ell+1} \frac{(\tilde{C}_0C_*)^q}{\sqrt{q!}}\sum_{p=q+1}^{\ell+2} \sum_{\beta=0}^{\beta_*}   \left(8\sqrt{3(s+1)}C_1\right)^{\ell+k+2p-\beta}C_2^{\beta-2q} 
\end{align*}
for a suitable constant $D$, where we have used that $p\leq \frac{1}{2}(\ell+k+2p-\beta)$ and therefore $2^p\binom{p}{q}\leq 4^p\leq 2^{\ell+k+2p-\beta}$. Using the shorthand $C_3 :=4\sqrt{3(s+1)}C_1$, we obtain for $C_2>C_3$
\begin{align*}
\sum_{\beta=0}^{\beta_*}  & \!  \! \left(4\sqrt{3(s  \! + \! 1)}C_1\right)^{\ell+k+2p-\beta} \!  \! C_2^{\beta-2q} \! = \! K^{\ell+k+2p}C_2^{-2q}\frac{\left(\frac{C_2}{K}\right)^{\beta_*}-\frac{K}{C_2}}{1-\frac{K}{C_2}}  \leq \!  \frac{C_2^{\ell+k}}{1-\frac{K}{C_2}} \left(\frac{K^3}{C_2}\right)^{p-q}.
\end{align*}
Furthermore, we have $\sum_{p=q+1}^{\ell+2} \left(\frac{C_3^3}{C_2}\right)^{p-q}C_2^{\ell+k} \leq \frac{1}{1-\frac{C_3^3}{C_2}}C_3^3 C_2^{\ell+k-1}$ for $C_2>C_3^3$, and consequently we obtain for a suitable constant $\widetilde{D}$
\begin{align*}
\|P  \mb P K_k\Upsilon_{\ell+2} \|\leq \frac{D K^3 \sqrt{(\ell+k+2)!} C_2^{\ell+k-1}}{\left(1-\frac{K}{C_2}\right)\left(1-\frac{K^3}{C_2}\right)}\sum_{q=0}^{\ell+1} \frac{(\tilde{C}_0C_*)^q}{\sqrt{q!}}\leq \widetilde{D}\sqrt{(\ell+k+2)!}C_2^{\ell+k-1}.
\end{align*}
Consequently we obtain for $C_2\geq 3\widetilde{D}$
\begin{align*}
\sum_{k=1}^{\ell+1}\big\|   P  \mb P  K _k  \Upsilon_{\ell-k+2}^{(n)} \big\|= \!  \!  \!  \!  \! \sum_{k=1}^{\max\{3,\ell+1\}} \!  \!  \! \big\| P \otimes \mb P K _k  \Upsilon_{\ell-k+2} \big\| \! \leq  \! 3\widetilde{D}\sqrt{(\ell+2)!}C_2^{\ell-1} \! \leq  \! C_2^\ell \sqrt{(\ell+2)!}.
\end{align*}
This completes the derivation of \eqref{eq:PKUpsilon:bound}. As a consequence of Lemma \ref{lem:energy:identity}, we thus obtain the first part of the lemma.

The second statement of the lemma follows in close analogy, by using the established bounds for $|E_\ell|$. The only major difference is to use the bound $\|\left(V_{k,0} Y\right)|_{\mathcal{H}_{3s}}\|\leq (1+\sqrt{\Lambda}\, ) C_1^{k}\sqrt{(s+1)\dots (s+k)}$ instead of the previously used estimate $\|\left(P V_{k,0} Y\right)|_{\mathcal{H}_{3s}}\|\leq C_1^{k}\sqrt{(s+1)\dots (s+k)}$, cf. Lemma \ref{lem:fock_space_results}. This completes the proof of Lemma \ref{lem:estimates:approx:groundstate:2}.
\end{proof}

\bigskip\noindent\textbf{Financial support.} M.B. gratefully acknowledges funding from the ERC Advanced Grant ERC-AdG CLaQS, grant agreement n. 834782.

\end{spacing}

Mathematical Physics, Analysis and Geometry 


\begin{thebibliography}{}

\bibitem{Auberson82} G. Auberson. Borel summability for a nonpolynomial potential. \emph{Comm. Math. Phys.} 84, 531--546. (1982)

\bibitem{BetzP22} V. Betz and S. Polzer. A functional central limit theorem for polaron path measures \emph{Comm. Pure. Appl. Math.} 75, 11, 2345--2392. (2022)

\bibitem{BetzP23} V. Betz and S. Polzer. Effective mass of the polaron: a lower bound. \emph{Comm. Math Phys.} 399, 173--188. (2023)

\bibitem{Bog50} N.N. Bogoliubov. On a new form of the adiabatic theory of disturbances in the problem of interaction
of particles with a quantum field. (Russian) \emph{Ukr. Mat. Zh 2,3-24}. (1950)

\bibitem{Bossmann2019} L. Bo\ss mann, S. Petrat, P. Pickl and A. Soffer. Beyond Bogoliubov dynamics. \emph{Pure Appl. Anal.} 3(4), 677--726. (2021)

\bibitem{BPS2021} L.~Bo{\ss}mann, S.~Petrat and R.~Seiringer. Asymptotic expansion of low-energy excitations for weakly interacting bosons. \emph{Forum of Mathematics, Sigma} 9, E28. (2021)

\bibitem{BS1} M. Brooks and R. Seiringer. The Fr\"ohlich polaron at strong coupling -- part I: The quantum correction to the classical energy. Preprint. \href{https://arxiv.org/abs/2207.03156}{\emph{arXiv:2207.03156}}. (2022)

\bibitem{BS2} M. Brooks and R. Seiringer. The Fr\"ohlich polaron at strong coupling -- part II: Energy-momentum relation and effective mass. Preprint. \href{https://arxiv.org/abs/2211.03353}{\emph{arXiv:2211.03353}}. (2022)

\bibitem{DonskerV83} M. Donsker and S.R.S. Varadhan. Asymptotics for the polaron. \emph{Comm. Pure Appl. Math.} 36, 505--528. (1983)

\bibitem{DybalskiS2020}
W. Dybalski and H. Spohn. Effective mass of the polaron -- revisited. \textit{Ann. Henri Poincar\'{e} 21, 1573--1594}. (2020)

\bibitem{Feliciangeli20} D. Feliciangeli and R. Seiringer. Uniqueness and non-degeneracy of minimizers of the Pekar functional on a ball. \emph{SIAM J. Math. Anal.} 52, 1. (2020)

\bibitem{Feliciangeli21} D. Feliciangeli and R. Seiringer. The strongly coupled polaron on the torus: Quantum corrections to the Pekar asymptotics. \emph{Arch. Rat. Mech. Anal.} 242(3), 1835--1906. (2021)

\bibitem{FrankSchlein14} R. Frank and B. Schlein. Dynamics of a strongly coupled polaron. \emph{Lett. Math. Phys.} 104 (8), 911--929. (2014)

\bibitem{FrankS21} R. Frank and R.~Seiringer.~Quantum corrections to the Pekar asymptotics of a strongly coupled polaron. \emph{Comm.~Pure~Appl.~Math.}~74(3), 544--588. (2021)

\bibitem{Froehlich37} H. Fr\"ohlich. Theory of electrical breakdown in ionic crystals. \emph{Proc. R. Soc. Lond. A} 160, 230--241. (1937)

\bibitem{GGS70} S. Graffi, V. Grecchi and B. Simon. Borel summability: Application to the anharmonic oscillator. \emph{Phys. Lett. B} 32(7), 631--634. (1970)

\bibitem{GrechS13} P. Grech and R. Seiringer. The excitation spectrum for weakly interacting bosons in a trap. \emph{Commun. Math. Phys.} 322(2):559--591. (2013)

\bibitem{GriesemerW18} M. Griesemer and A. W\"unsch. On the domain of the Nelson Hamiltonian. \emph{J. Math. Phys}. 59(4), 042111. (2018)

\bibitem{Gross76} E.P. Gross. Strong coupling polaron theory and translational invariance. \emph{Ann. Phys. 99} 1--29. (1976)

\bibitem{Kato66} T. Kato. Perturbation theory for linear operators. Springer. (1995)

\bibitem{LMM23} J. Lampart, D. Mitrouskas and K. My\'sliwy. On the global minimium of the energy-momentum relation of the polaron. \emph{Math. Phys. Anal. Geom.} 26, 17. (2023)

\bibitem{Lenzmann09} E. Lenzmann. Uniqueness of ground states for pseudorelativistic Hartree equations. \emph{Anal. PDE} 2, 1--27. (2009)

\bibitem{LNSS} M. Lewin, P. Nam, S. Serfaty and J. Solovej. Bogoliubov spectrum of interacting
Bose gases. \emph{Comm.~Pure~Appl.~Math.} 68, 413--471. (2015)

\bibitem{Lieb77} E.H. Lieb. Existence and uniqueness of the minimizing solution of Choquard’s non-linear equation. \emph{Studies in Appl. Math.} 57, 93--105. (1977)

\bibitem{LS} E. Lieb and J. Solovej. Ground state energy of the one-component charged Bose gas. \emph{Comm.~Math.~Phys.} 217, 127--163. (2001)

\bibitem{LiebT97} E.H. Lieb and L.E. Thomas. Exact ground state energy of the strong-coupling polaron. \textit{Comm. Math. Phys. 183(3), 511--519}. (1997)

\bibitem{LiebY58} E.H. Lieb and K. Yamazaki. Ground-state energy and effective mass of the polaron. \emph{Phys. Rev. 111, 728}. (1958)

\bibitem{MMS23} D. Mitrouskas, K. My\' sliwy and R. Seiringer. Optimal parabolic upper bound for the energy-momentum relation of a strongly coupled polaron. \emph{Forum of Mathematics, Sigma}, 11, E49. (2023)

\bibitem{MitS2022} D. Mitrouskas and R. Seiringer. Ubiquity of bound states for the strongly coupled polaron. Preprint. \href{https://arxiv.org/abs/2211.03606}{\emph{arXiv:2211.03606}}. (2022)

\bibitem{Miyake76} S.J. Miyake. The ground state of the optical polaron in the strong-coupling case. \emph{J. Phys. Soc. Jpn.} 41, 747--752. (1976)

\bibitem{Moeller06} J.S. M\o ller, The polaron revisited. \emph{Rev. Math. Phys.} 18, 485--517. (2006)

\bibitem{Mukherjee} C. Mukherjee, S.\- R.\- S. Varadhan. Identification of the Polaron Measure I: Fixed Coupling Regime and the Central Limit Theorem for Large Times \emph{Comm. Pure. Appl. Anal.} 73, 2, 350--383. (2020)

\bibitem{Polzer22} S. Polzer. Renewal approach for the energy-momentum relation of the Fr\"ohlich polaron. \emph{Lett. Math. Phys.} 113, 90. (2023)

\bibitem{Seiringer11} R. Seiringer. The excitation spectrum for weakly interacting bosons. \emph{Commun. Math. Phys.} 306(2):565--578. (2011)

\bibitem{Seiringer21} R. Seiringer. The polaron at strong coupling. \emph{Math. Rev. Phys.} 33, 1. (2021)

\bibitem{Selke} M. Selke. Almost quartic lower bound for the Fr\"ohlich polaron's effective mass via Gaussian domination. Preprint. \href{https://arxiv.org/abs/2212.14023}{\emph{arXiv:2212.14023}}. (2022)
 
\bibitem{Simon70} B. Simon. Borel summability of the ground-state energy in spatially cutoff $(\phi^4)_2$. \emph{Phys. Rev. Lett.} 25, 1583. (1970)

\bibitem{Simon71} B. Simon. 
Determination of eigenvalues by divergent perturbation series. \emph{Adv. Math.} 7, 240--253. (1971)

\bibitem{Simon82} B. Simon. Large orders and summability of eigenvalue perturbation theory: A mathematical overview. \emph{Int. J. Quantum Chem.} 21(1), 3--25. (1982)

\bibitem{Sokal80} A. Sokal. An improvement of Watson’s theorem on Borel summability. \emph{J. Math. Phys.} 21, 261. (1980)

\bibitem{Spohn1988} H. Spohn. The polaron at large total momentum. \emph{J. Phys. A: Math. Gen.} 21, 1199. (1988)

\bibitem{Tjablikow54} S.W. Tjablikow. Adiabatische Form der St\"orungstheorie im Problem der Wechselwirkung eines Teilchens mit einem gequantelten Feld. \emph{Abhandl. Sowj. Phys. 4, 54--68}. (1954)

\end{thebibliography}
\end{document}